\documentclass[11pt]{article} 
\usepackage{fullpage,amsmath,amsthm,amssymb,hyperref,cleveref,graphicx,url,rotating}
\usepackage{cancel,xcolor}

\graphicspath{{./},{./Fig/},{./figD/}}

\usepackage{listings}
\definecolor{mydarkblue}{rgb}{0,0.08,0.45}

\usepackage{scalerel}
\usepackage{tikz}
\usetikzlibrary{svg.path}
\definecolor{orcidlogocol}{HTML}{A6CE39}
\tikzset{
  orcidlogo/.pic={
    \fill[orcidlogocol] svg{M256,128c0,70.7-57.3,128-128,128C57.3,256,0,198.7,0,128C0,57.3,57.3,0,128,0C198.7,0,256,57.3,256,128z};
    \fill[white] svg{M86.3,186.2H70.9V79.1h15.4v48.4V186.2z}
                 svg{M108.9,79.1h41.6c39.6,0,57,28.3,57,53.6c0,27.5-21.5,53.6-56.8,53.6h-41.8V79.1z M124.3,172.4h24.5c34.9,0,42.9-26.5,42.9-39.7c0-21.5-13.7-39.7-43.7-39.7h-23.7V172.4z}
                 svg{M88.7,56.8c0,5.5-4.5,10.1-10.1,10.1c-5.6,0-10.1-4.6-10.1-10.1c0-5.6,4.5-10.1,10.1-10.1C84.2,46.7,88.7,51.3,88.7,56.8z};
  }
}

\newcommand\orcidicon[1]{\href{https://orcid.org/#1}{\mbox{\scalerel*{
\begin{tikzpicture}[yscale=-1,transform shape]
\pic{orcidlogo};
\end{tikzpicture}
}{|}}}}

\hypersetup{ pdftitle={},
    pdfauthor={},
    pdfsubject={},
    pdfkeywords={},
    pdfborder=0 0 0,
    pdfpagemode=UseNone,
    colorlinks=true,
    linkcolor=mydarkblue,
    citecolor=mydarkblue,
    filecolor=mydarkblue,
    urlcolor=mydarkblue,
    pdfview=FitH}
\def\dnu{\mathrm{d}\nu}
\def\dmu{\mathrm{d}\mu}
\def\calB{\mathcal{B}}
\def\TV{\mathrm{TV}}
\def\Hankel{\mathrm{Hankel}}		
\def\IS{\mathrm{IS}}		
\def\bartheta{{\bar\theta}}
\def\bareta{{\bar\eta}}
\def\floor#1{\lfloor{#1}\rfloor}
\def\ceil#1{\lceil{#1}\rceil}
\def\bbR{\mathbb{R}}
\def\kl{\mathrm{kl}}
\def\st{\ :\ }
\def\AC{\mathrm{AC}}
\def\calF{\mathcal{F}}
\def\iid{\mathrm{i.i.d.}}
\def\JS{\mathrm{JS}}
\def\KL{\mathrm{KL}}
\def\calS{\mathcal{S}}
\def\calX{\mathcal{X}}
\def\bbN{\mathbb{N}}
\def\MLE{\mathrm{MLE}}
\def\SME{\mathrm{SME}}
\def\ML{\mathrm{MLE}}
\def\SM{\mathrm{SME}}
\def\LS{\mathrm{LS}}
\def\LMS{\mathrm{LMS}}
\def\calE{\mathcal{E}}
\def\Mtheta{{M_1}}
\def\Meta{{M_2}}
\def\Var{\mathrm{Var}}
\def\colvec#1{{\left[\begin{array}{l}#1\end{array}\right]}}
\def\dx{\mathrm{d}x}
\def\calM{\mathcal{M}}
\def\calA{\mathcal{A}}

\newtheorem{Example}{Example}

\newtheorem{Corollary}{Corollary}
\newtheorem{Theorem}{Theorem}
\newtheorem{Lemma}{Lemma}
\newtheorem{Proposition}{Proposition}

\def\Var{\mathrm{Var}}

\def\colvec#1{{\left[\begin{array}{l}#1\end{array}\right]}}
\def\dx{\mathrm{d}x}

\begin{document}

\title{Fast approximations of the Jeffreys divergence between univariate Gaussian mixture models via exponential polynomial densities\footnote{This paper has been published after peer-reviewed in a polished and revised form in~\cite{e23111417}.}}  

\author{Frank Nielsen\orcidicon{0000-0001-5728-0726}\\ Sony Computer Science Laboratories Inc.\\ Tokyo, Japan}

\date{ }

\maketitle

\def\calB{\mathcal{B}}
\def\TV{\mathrm{TV}}
\def\Hankel{\mathrm{Hankel}}		
\def\IS{\mathrm{IS}}		
\def\bartheta{{\bar\theta}}
\def\bareta{{\bar\eta}}
\def\floor#1{\lfloor{#1}\rfloor}
\def\ceil#1{\lceil{#1}\rceil}
\def\bbR{\mathbb{R}}
\def\kl{\mathrm{kl}}
\def\st{\ :\ }
\def\AC{\mathrm{AC}}
\def\calF{\mathcal{F}}
\def\iid{\mathrm{i.i.d.}}
\def\JS{\mathrm{JS}}
\def\KL{\mathrm{KL}}
\def\calS{\mathcal{S}}
\def\calX{\mathcal{X}}
\def\bbN{\mathbb{N}}
\def\MLE{\mathrm{MLE}}
\def\SME{\mathrm{SME}}
\def\ML{\mathrm{MLE}}
\def\SM{\mathrm{SME}}
\def\LS{\mathrm{LS}}
\def\LMS{\mathrm{LMS}}
\def\calE{\mathcal{E}}
\def\Mtheta{{M_1}}
\def\Meta{{M_2}}
\def\Var{\mathrm{Var}}
\def\colvec#1{{\left[\begin{array}{l}#1\end{array}\right]}}
\def\dx{\mathrm{d}x}
\def\calM{\mathcal{M}}
\def\calA{\mathcal{A}}

\begin{abstract}
The Jeffreys divergence is a renown symmetrization of the oriented Kullback-Leibler divergence broadly used in information sciences.
Since the Jeffreys divergence between Gaussian mixture models is not available in closed-form, various techniques with pros and cons have been proposed in the literature to either estimate, approximate, or lower and upper bound this divergence. 
In this paper, we propose a simple yet fast heuristic to approximate the Jeffreys divergence between two univariate Gaussian mixtures with arbitrary number of components. 
Our heuristic relies on converting the mixtures into pairs of dually parameterized probability densities belonging to an exponential family.
In particular, we consider the versatile polynomial exponential family densities, and design a divergence to measure in closed-form the  goodness of fit between a Gaussian mixture and its polynomial exponential density approximation.
This goodness-of-fit divergence is a generalization of the Hyv\"arinen divergence used to estimate models with computationally intractable normalizers. It allows us to perform model selection by choosing the orders of the polynomial exponential densities used to approximate the mixtures.
We demonstrate experimentally that our heuristic to approximate the Jeffreys divergence improves by several orders of magnitude the computational time of stochastic Monte Carlo estimations while approximating reasonably well the Jeffreys divergence, specially when the mixtures have a very small number of modes. 
Besides, our mixture-to-exponential family conversion techniques may prove useful in other settings.
\end{abstract}



\section{Introduction}

\subsection{Statistical mixtures and statistical divergences}

In this work, we consider the problem of approximating the Jeffreys divergence~\cite{jeffreys1946invariant} between two univariate continuous mixture models~\cite{mclachlan1988mixture} 
$m(x)=\sum_{i=1}^k w_i p_i(x)$ and $m'(x)=\sum_{i=1}^{k'} w_i' p_i'(x)$ with continuous component distributions $p_i$'s and $p_i'$'s defined on a coinciding support $\calX$. 
The mixtures $m(x)$ and $m'(x)$ may have different number of components (i.e., $k\not=k'$).
Historically, Pearson~\cite{pearson1894contributions} first considered a univariate Gaussian mixture of two components for modeling the distribution of ratio of forehead breadth to body length of a thousand crabs in 1894 (the obtained mixture was unimodal).

Although our work applies to any continuous mixtures of exponential families  (e.g., Rayleigh mixtures~\cite{seabra2011rayleigh} with restricted support $\calX=\bbR_+$), we explain our method for the most prominent family of mixtures encountered in practice: The Gaussian mixture models, or GMMs for short (also abbreviated as MoG for Mixtures of Gaussians~\cite{MoG-2006}). 
In the remainder, a univariate GMM $m(x)=\sum_{i=1}^k w_i p_{\mu_i,\sigma_i}(x)$  with $k$ Gaussian components 
$$
p_i(x)=p_{\mu_i,\sigma_i}(x):=\frac{1}{\sigma_i\sqrt{2\pi}}\exp\left(-\frac{(x-\mu_i)^2}{2\sigma_i^2}\right),
$$ 
is called a $k$-GMM.

The  Kullback-Leibler divergence~\cite{kullback1997information} $D_\KL$ (KLD) between two probability density functions $m$ and $m'$ is:
\begin{equation}
D_\KL[m:m'] := \int_\calX m(x)\log\left(\frac{m(x)}{m'(x)}\right)\dx.
\end{equation}
The KLD is an oriented divergence since $D_\KL[m:m']\not=D_\KL[m':m]$.

The Jeffreys divergence~\cite{jeffreys1946invariant} (JD) $D_J$ is the arithmetic symmetrization of the forward and reverse KLDs:
\begin{eqnarray}
D_J[m,m'] &:=& D_\KL[m:m'] +D_\KL[m':m],\\
&=& \int_\calX (m(x)-m'(x))\log\left(\frac{m(x)}{m'(x)}\right)\dx.
\end{eqnarray}
The JD is a symmetric divergence: $D_J[m,m']=D_J[m',m]$.
In the literature, the Jeffreys divergence~\cite{vitoratou2017thermodynamic} has also been called the 
$J$-divergence~\cite{kannappan1988axiomatic,burbea2004j}, the symmetric Kullback-Leibler divergence~\cite{tabibian2015speech} and sometimes the symmetrical Kullback-Leibler divergence~\cite{veldhuis2002centroid,nielsen2013jeffreys}. 
Notice that there are many other ways to symmetrize the KLD~\cite{nielsen2019jensen} beyond the usual Jeffreys divergence and renown Jensen-Shannon divergence~\cite{Lin-1991} $D_\JS$:
$$
D_\JS[m:m'] := D_\KL\left[m:\frac{m+m'}{2}\right]+D_\KL\left[m':\frac{m+m'}{2}\right].
$$

In general, it is provably hard to calculate in closed-form the integral of the KLD between two continuous mixtures:
For example, the KLD between two GMMs has been shown to be non-analytic~\cite{KLnotanalytic-2004}.
One recent notable exception to this hardness result of calculating KLD between mixtures is the closed-form analytic formula (albeit being large) reported for the KLD between two Cauchy mixtures of two components~\cite{nielsen2021dually}.
Thus in practice, when calculating the JD between two GMMs, one can either  approximate~\cite{garcia2010simplification,cui2015comparison}, estimate~\cite{sreekumar2021non}, or  bound~\cite{durrieu2012lower,nielsen2016guaranteed} the KLD between mixtures.
Another approach  to bypass the computational intractability of calculating the KLD between mixtures consists in designing new types of divergences taylored to mixtures which admit closed-form expressions. 
See~\cite{jenssen2006cauchy,nielsen2012closed,nielsen2019statistical} for some examples of statistical divergences (e.g., Cauchy-Schwarz divergence~\cite{jenssen2006cauchy}) well-suited to mixtures.

In practice, the vanilla Monte Carlo (MC) estimator of the KLD between mixtures consists in first rewriting the KLD as
$$
D_\KL[m:m'] = \int_\calX \left(m(x)\log\left(\frac{m(x)}{m'(x)}\right)+m'(x)-m(x)\right)\dx =\int_\calX D_\kl(m(x):m'(x))\dx
$$
where $D_\kl(a:b):=a\log\frac{a}{b}+b-a$ is a scalar Bregman divergence~\cite{Bregman-1967,KLBD-2001,BregmanKmeans-2005,nock2015conformal} (hence non-negative), and then performing Monte Carlo stochastic integration:
$$
\hat{D}_\KL^{\calS_s}[m:m'] :=  \frac{1}{s} \sum_{i=1}^s \frac{1}{m(x_i)} D_\kl(m(x_i):m'(x_i))= \frac{1}{s} \sum_{i=1}^s \left( \log\left(\frac{m(x_i)}{m'(x_i)}\right) + \frac{m'(x_i)}{m(x_i)}-1\right),
$$
where $\calS_s=\{x_1,\ldots, x_s\}$ is $s$  independent and identically distributed (iid.) samples from $m(x)$. 
This MC estimator is by construction always non-negative (a weighted sum of non-negative Bregman scalar divergences), and furthermore consistent. 
That is, we have 
$\lim_{s\rightarrow \infty} \hat{D}_\KL^{\calS_s}[m:m']=D_\KL[m:m']$ (under mild conditions~\cite{MC-2013}).

Similarly, we estimate the Jeffreys divergence via MC sampling as follows:
\begin{equation}\label{eq:MCmix}
\hat{D}_J^{\calS_s}[m,m'] :=   \frac{1}{s} \sum_{i=1}^s \frac{2}{m(x_i)+m'(x_i)} \left(m(x_i)-m'(x_i)\right)\log\left(\frac{m(x_i)}{m'(x_i)}\right),
\end{equation}
where $\calS_s=\{x_1,\ldots, x_s\}$ are $s$ iid. samples from the middle mixture $m_{12}(x):=\frac{1}{2}(m(x)+m'(x))$.
Since the scalar Jeffreys divergence $D_j(p,q):=(p-q)\log \frac{p}{q}=D_\kl(p:q)+D_\kl(q:p)\geq 0$, we have $\hat{D}_J^{\calS_s}[m,m']\geq 0$.
By choosing the middle mixture $m_{12}(x)$ for sampling, we ensure that we keep the symmetric property of the JD:
That is, $\hat{D}_J^{\calS_s}[m,m']=\hat{D}_J^{\calS_s}[m',m]$. 
We also have consistency under mild conditions: $\lim_{s\rightarrow \infty} \hat{D}_J^{\calS_s}[m,m']=D_J[m,m']$.
Thus the time complexity to stochastically estimate  the JD is $\tilde O((k+k')s)$, with $s$ typically ranging from $10^4$ to $10^6$ in applications.
Notice that the number of components can be very large (e.g., $k=O(n)$ for $n$ input data when using Kernel Density Estimators~\cite{simplifykde-2013}).
KDEs build mixtures by setting a mixture component at each data location.
Those KDE mixtures have a large number of components and may potentially exhibit many spurious modes visualized as small bumps when plotting the mixture densities.

\subsection{Polynomial exponential families and Jeffreys divergence}
In this work, we shall consider the approximation of the JD by converting continuous mixtures into densities of exponential families~\cite{BN-2014} also called tilted families~\cite{Efron-2016} (i.e., densities obtained by tilting the Lebesgue base measure).
A continuous exponential family (EF) $\calE_t$ of order $D$ is defined as a family of probability density functions with support $\calX$ and probability density function:
$$
\calE_t:=\left\{ p_\theta(x):=\exp\left(\sum_{i=1}^D \theta_i t_i(x)-F(\theta)\right) \ :\ \theta\in\Theta  \right\},
$$ 
where $F(\theta)$ is called the log-normalizer (also called log-Laplace transform) which ensures normalization of
 $p_\theta(x)$ (i.e., $\int_\calX p_\theta(x)\dx=1$):
$$
F(\theta)=\log\left(\int_\calX \exp\left(\sum_{i=1}^D \theta_it_i(x)\right)\dx\right).
$$
The log-normalizer is also called the cumulant function because the
 cumulant  generator function (CGF) $K(u)=\log E_{p_\theta}\left[\exp(\sum_{i=1}^D u_i t_i(x_i))\right]=F(\theta+u)-F(\theta)$ is related to function $F(\theta)$ (see \S\ref{sec:rawmom}).
Parameter $\theta\in\Theta\subset\bbR^D$ is called the natural or canonical parameter, and the vector $t(x)=[t_1(x)\ \ldots\ t_D(x)]^\top$ is called the sufficient statistics~\cite{BN-2014}. Let $\Theta$ denotes the natural parameter space: $\Theta:=\{\theta\ :\ F(\theta)<\infty\}$, an open convex domain for regular exponential families~\cite{BN-2014}.
 
Polynomial exponential families~\cite{Cobb-1983,clutton1990density,barron1991approximation,NN-2016} (PEF) $\calE_D$ are exponential families with  polynomial sufficient statistics $t_i(x)=x^i$ for $i\in\{1,\ldots, D\}$.
For example, the exponential distributions $\{p_\lambda(x)=\lambda\exp(-\lambda x)\}$ form a PEF with $D=1$, $t(x)=x$ and $\calX=\bbR_+$, and the normal distributions form a PEF with $D=2$, $t(x)=[x\ x^2]^\top$ and $\calX=\bbR$, etc. 
PEFs are also called exponential-polynomial densities~\cite{Demidenko-2010,PEF-holonomic2016}.
PEDs have positive densities by construction, and this contrasts with modeling the density by polynomials~\cite{buckland1992fitting} which may yield densities which are potentially negative at some values. 
Clutton-Brock~\cite{clutton1990density} estimated densities from iid. observations using exponentials of orthogonal series (which ensures positivity of the densities). The considered densities are $\exp\left(\sum_{i=1}^D \theta_i P_i(x) -F(\theta)\right)$, where the $P_i$'s are orthogonal polynomials (e.g., Legendre, Chebyshev, Gegenbauer, Hermite, Laguerre  polynomials, etc.)
Historically, Neyman~\cite{Neyman-1937} used an exponential of a series of Legendre polynomials in 1937 to develop his ``smooth test''.
Yet another versatile methodology to estimate density in statistics is spline density smoothing~\cite{wahba1990spline}.

The log-normalizes $F(\theta)$ can be obtained in closed-form for lower order PEFs (e.g., $D=1$ or $D=2$) or very special subfamilies of PEFs.
However, no-closed form formula are available for $F(\theta)$ of PEFs in general as soon $D\geq 4$, and the cumulant function $F(\theta)$ is said to be computationally intractable.
See also the exponential varieties~\cite{michalek2016exponential} related to polynomial exponential families.
Notice that when $\calX=\bbR$ the coefficient $\theta_D$ is negative  for even integer order $D$. 

PEFs are attractive because these families can universally model smooth multimodal distributions~\cite{Cobb-1983}, and require fewer parameters in comparison to GMMs:
Indeed, a univariate $k$-GMM $m(x)$   (at most $k$ modes and $k-1$ antimodes) requires $3k-1$ parameters to specify $m(x)$ (or $k+1$ for a KDE~\cite{simplifykde-2013} with constant kernel width $\sigma$ or $2k-1$ for varying kernel widths, but $k=n$ observations).
A density of a PEF of order $D$ (called a Polynomial Exponential Density, PED) requires $D$ parameters to specify $\theta$ but has at most $\frac{D}{2}$ modes and $\frac{D}{2}-1$ antimodes.

The case of the quartic (polynomial) exponential densities $\calE_4$ ($D=4$) has been extensively investigated in~\cite{o1933method,aroian1948fourth,QuarticEF-1978,zellner1988calculation,mccullagh1994exponential,PEFMaxEnt-1984}. 
Armstrong and Brigo~\cite{PEFfiltering-2013} discussed order-$6$ PEDs,
and Efron and Hastie reported and order-$7$ PEF in their textbook (see Figure 5.7 of~\cite{Efron-2016}).

Let $P_\theta(x):=\sum_{i=1}^D \theta_ix^i$ be a homogeneous polynomial defining the shape of the PEF:
$$
p_\theta(x)=\exp\left(P_\theta(x)-F(\theta)\right).
$$
When $P_\theta$ is a monomial, the cumulant function is available in closed-form and the PEF is called a Monomial Exponential Family.
This closed-form property has been used to devise a sequence of maximum entropy upper bounds for GMMs~\cite{nielsen2017maxent}. 
Appendix~\ref{sec:MEF} describes the characteristics of the MEFs.
Since the logarithm function is strictly increasing, the stationary points $\gamma_j$ of $p_\theta(x)>0$
 (satisfying $p_\theta'(\gamma_j)=0$) are equivalent to the stationary points of
$\log p_\theta(x)=P_\theta(x)-F(\theta)$, i.e., the points $\gamma_j$ such that $P_\theta'(\gamma_j)=\sum_{i=1}^D i\theta_i\gamma_j^{i-1}=0$.
Reciprocally, polynomial $P_\theta(x)$ can also be expressed using stationary points $\gamma_i$'s of its derivative as
$$
P_\theta(x)=\int_0^x \prod_j (u-\gamma_j) \mathrm{d}u.
$$
For example, consider $P_\theta(x)$ with $P_\theta'(x)=\gamma_4(x-\gamma_1)(x-\gamma_2)(x-\gamma_3)(x-\gamma_4)$.
$P_\theta(x)$ has stationary points $\gamma_i$.
Then we get the corresponding polynomial exponential family:
$$
p_\theta(x) =: \exp(\sum_{i=1}^D \theta_ix^i-F(\theta)).
$$

In the context of deep learning~\cite{DL-2016}, the PEFs can be interpreted as a simple class of Energy-Based Models~\cite{lecun2006tutorial,grathwohl2019your} (EBMs).
Thus let us write
$$
p_\theta(x)=\frac{q_\theta(x)}{Z(\theta)},
$$
where $q_\theta(x):=\exp\left(\sum_{i=1}^D \theta_ix^i\right)$ is the  unnormalized density, and $Z(\theta)=\exp(F(\theta))$ is called
 the {\em partition function} in statistical physics.
 Hence, $F(\theta)$ is also called log-partition function since $F(\theta)=\log Z(\theta)$.
We can define and equivalence class $\sim$ such that $p_1(x)\sim p_2(x)$ iff. there exists $\lambda>0$ such that $p_1(x)=\lambda p_2(x)$.
In the literature, $\tilde{p}_\theta(x)$ often denote one representative of the equivalence class $p_\theta(x)/\sim$, the distribution $q_\theta(x)$: $\tilde{p}_\theta(x)=q_\theta(x)$.

PEFs like any other exponential family admit a dual parameterization~\cite{BN-2014} $\eta=\eta(\theta):=E_{p_\theta}[t(x)]=\nabla F(\theta)$,  called the moment parameterization (or mean parameterization). Let $H$ (pronounced Eta using the greek alphabet) denote the moment parameter space. Let us use the subscript and superscript notations to emphasize the coordinate system used to index a PEF: 
In our notation, we thus write
 $p_\theta(x)=p^\eta(x)$.

It is known that the KLD between any two densities of an exponential family amounts to a reverse Bregman divergence ${B_F}^*$ induced by the cumulant  of the EF~\cite{KLBD-2001,BregmanKmeans-2005}: 
$$
D_\KL[p_\theta:p_{\theta'}]={B_F}^*(\theta:\theta'):=B_F(\theta':\theta),
$$
where the Bregman divergence for a strictly convex and smooth generator $F$ is defined by:
\begin{eqnarray}
B_F(\theta_1:\theta_2)&:=&F(\theta_1)-F(\theta_2)-(\theta_1-\theta_2)^\top \nabla F(\theta_2),\\
&=& F(\theta_1)-F(\theta_2)-(\theta_1-\theta_2)^\top \eta_2,
\end{eqnarray}
with $\eta_2:=\nabla F(\theta_2)$.

Thus the JD between two PEDS of a PEF can be written equivalently as:
\begin{eqnarray}
D_J(p_\theta,p_{\theta'}) &=& D_\KL[p_\theta:p_{\theta'}] + D_\KL[p_{\theta'}:p_{\theta}],\\
&=& B_F(\theta':\theta)+B_F(\theta:\theta'),\\
&=&   B_F(\theta':\theta)+B_{F^*}(\eta':\eta),\\
&=& (\theta'-\theta)^\top(\eta'-\eta),\\
&=& (\theta'-\theta)^\top (\nabla F(\theta')-\nabla F(\theta)),\\
&=& (\nabla F^*(\eta')-\nabla F^*(\eta))^\top (\eta'-\eta),
\end{eqnarray}
where $F^*(\eta)$ denote the Legendre-Fenchel convex conjugate:
\begin{equation}
F^*(\eta):=\sup_{\theta\in\Theta} \{\theta^\top\eta-F(\theta)\},
\end{equation}
and the dual Bregman divergence is $B_{F^*}(\eta':\eta)=B_F(\theta:\theta')$.

Figure~\ref{fig:SBgeometricinterpretation} illustrates geometrically the symmetrized Bregman divergence for a univariate generator.
Let us visually notice that we can read the following dual Bregman divergences as areas described by the following definite integrals:
$$
B_F(\theta_1:\theta_2) = \int_{\theta_2}^{\theta_1} (F'(\theta)-F'(\theta_2))\, \mathrm{d}\theta,
$$
and
$$
B_{F^*}(\eta_2:\theta_1) = \int_{\eta_1}^{\eta_2} ({F^*}'(\eta)-{F^*}'(\eta_1))\, \mathrm{d}\eta.
$$

Thus by combining these two definite integral areas,
 we obtain the symmetrized Bregman divergence $S_F(\theta_1,\theta_2)=S_F(\theta_2,\theta_1)$ as the area rectangle $(\theta_1-\theta_2)^\top\, (\eta_1-\eta_2)$:
\begin{eqnarray*}
S_F(\theta_1,\theta_2) &=& B_F(\theta_1:\theta_2)+B_F(\theta_2:\theta_1),\\
&=& B_F(\theta_1:\theta_2)+B_{F^*}(\eta_1:\eta_2),\\
&=& (\theta_1-\theta_2)^\top\, (\eta_1-\eta_2).
\end{eqnarray*}

\begin{figure}%
 \centering
\includegraphics[width=0.8\columnwidth]{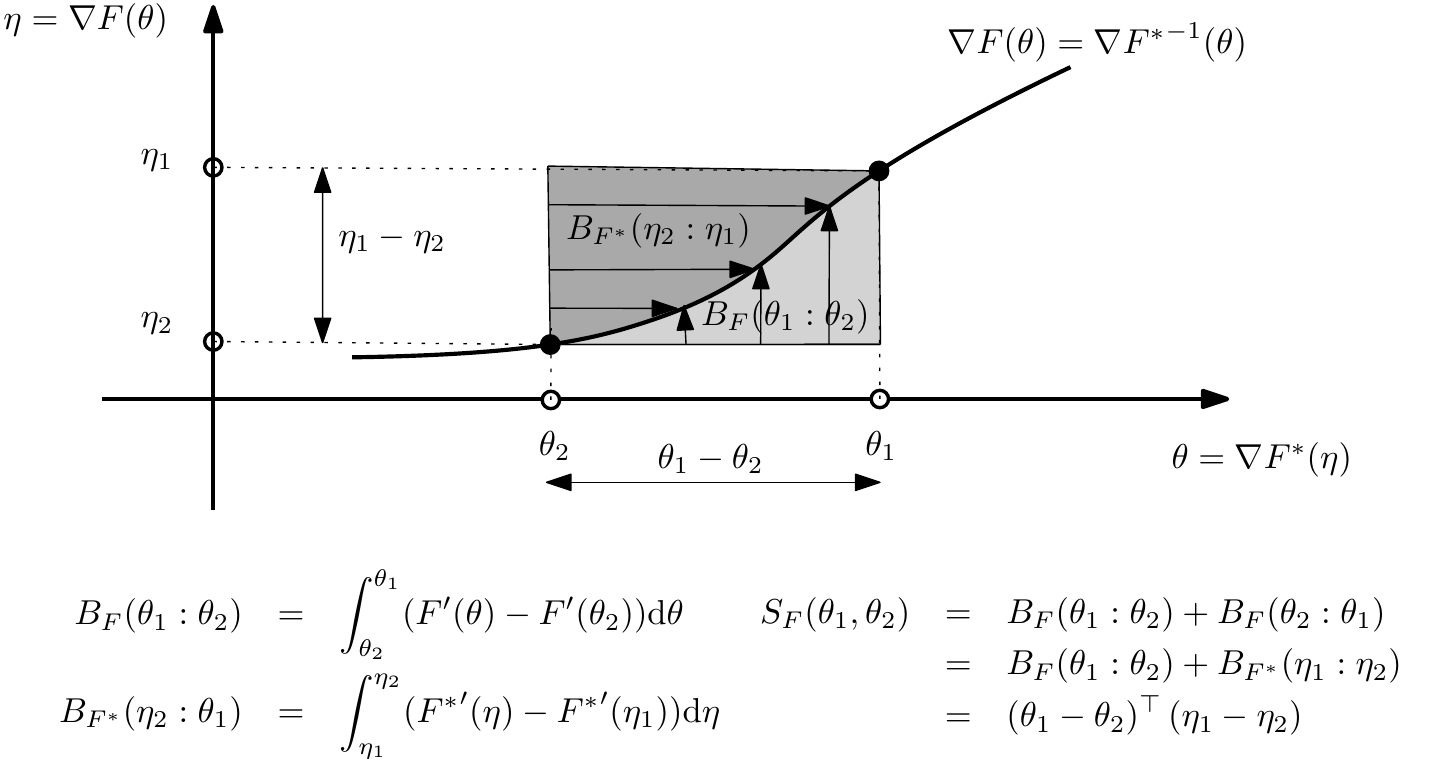}
\caption{Geometric interpretation of the symmetrized Bregman divergence as a rectangle area.}\label{fig:SBgeometricinterpretation}
\end{figure}

It follows from the Legendre transform that we have $\eta=\nabla F(\theta)$ and $\theta=\nabla F^*(\eta)$.
Thus the density of an exponential family expressed using the moment parameterization~\cite{nielsen2012k} is:
\begin{eqnarray}
p^\eta(x) &=&\exp\left(\nabla F^*(\eta)^\top t(x)-F(\nabla F^*(\eta))\right),\\
&=& \exp\left(-B_{F^*}(t(x):\eta)+F^*(t(x))\right),
\end{eqnarray}
where $B_{F^*}$ is the dual Bregman divergence.
This dual parameterization of the density emphasizes the bijection between regular exponential families 
and ``regular'' Bregman divergences~\cite{BregmanKmeans-2005}.

Using the mixed natural and moment parameterizations (with $\eta=\nabla F(\theta)$ and $\eta'=\nabla F(\theta')$), we get the following expression of the Jeffreys divergence:

\begin{Proposition}[Jeffreys divergence between densities of an exponential family]\label{prop:JD}
Let $p_\theta$ and $p_{\theta'}$ be two densities of an exponential family $\calE$. 
Then the Jeffreys divergence is:
\begin{equation}\label{eq:JDEF}
D_J[p_\theta,p_{\theta'}] = (\theta'-\theta)^\top (\eta'-\eta).
\end{equation}
\end{Proposition}

\begin{proof}
The proof is straightforward:
\begin{eqnarray*}
D_J[p_\theta,p_{\theta'}] &=& D_\KL[p_\theta:p_{\theta'}]+D_\KL[p_{\theta'}:p_\theta],\\
&=& B_F(\theta':\theta)+B_F(\theta:\theta'),\\
&=& \cancel{F(\theta')}-\cancel{F(\theta)}-(\theta'-\theta)^\top \eta + \cancel{F(\theta)}-\cancel{F(\theta')}-(\theta-\theta')^\top \eta',\\
&=& (\theta'-\theta)^\top (\eta'-\eta).
\end{eqnarray*}
\end{proof}

Interestingly, observe that the cumulant function $F(\theta)$ does not appear explicitly in Eq.~\ref{eq:JDEF} of Proposition~\ref{prop:JD} (although it occurs implicitly in the moment parameters $\eta=\nabla F(\theta)$ or dually in the natural parameters $\theta=\nabla F^*(\eta)$).
An alternative way to derive Eq.~\ref{eq:JDEF} is to consider the Legendre-Fenchel divergence~\cite{wmixtures-2018} $L_F$ which is equivalent to a Bregman divergence but which uses both mixed natural and moment parameterizations:
\begin{equation}\label{eq:LJDEF}
L_F(\theta_1:\eta_2):=F(\theta_1)+F^*(\eta_2)-\theta_1^\top\eta_2=L_{F^*}(\eta_2:\theta_1)=B_F(\theta_1:\theta_2)=B_{F^*}(\eta_2:\eta_1).
\end{equation}
Then we have for PDFs of an EF:
\begin{eqnarray}
D_J[p_{\theta_1}:p_{\theta_2}]&=& L_F(\theta_1:\eta_2)+L_F(\theta_2:\eta_1),\\
&=& \underbrace{F(\theta_1)+F^*(\eta_1)}_{\theta_1^\top\eta_1}+\underbrace{F(\theta_2)+F^*(\eta_2)}_{\theta_2^\top\eta_2}
-\theta_1^\top\eta_2-\theta_2^\top\eta_1,\\
&=& (\theta_2-\theta_1)^\top (\eta_2-\eta_1).
\end{eqnarray}

\subsection{A simple approximation heuristic}

In view of Proposition~\ref{prop:JD}, our method to approximate the Jeffreys divergence between mixtures $m$ and $m'$ consists in first converting those mixtures $m$ and $m'$ into pairs of polynomial exponential densities (PEDs) in \S~\ref{sec:mm2pef}. 
To convert a mixture $m(x)$ into a pair $(p_{\bartheta_1},p^{\bareta_2})$ dually parameterized (but not dual because $\bareta_2\not=\nabla F(\bartheta_1)$), 
 we shall consider  ``integral extensions'' of the  Maximum Likelihood Estimator~\cite{BN-2014}
 (MLE which estimates in the moment parameter space $H=\{\nabla F(\theta)\ :\ \theta\in\Theta\}$) and of  the Score Matching Estimator~\cite{hyvarinen-2005}  
(SME which estimates in the natural parameter space $\Theta=\{\nabla F^*(\eta)\ :\ \eta\in H\}$).

\begin{figure}%
\begin{tabular}{cc}
\includegraphics[width=0.48\columnwidth]{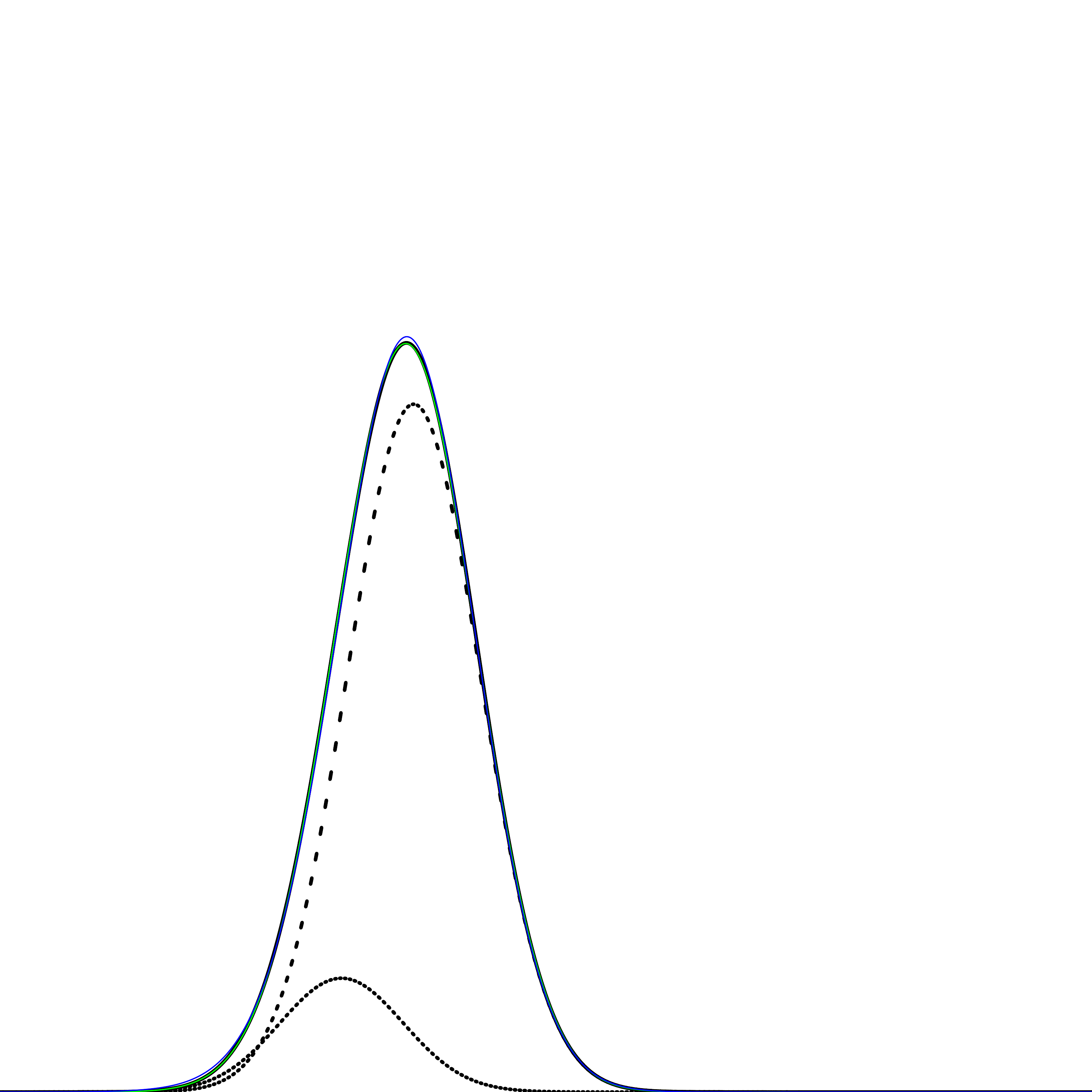}&
\includegraphics[width=0.48\columnwidth]{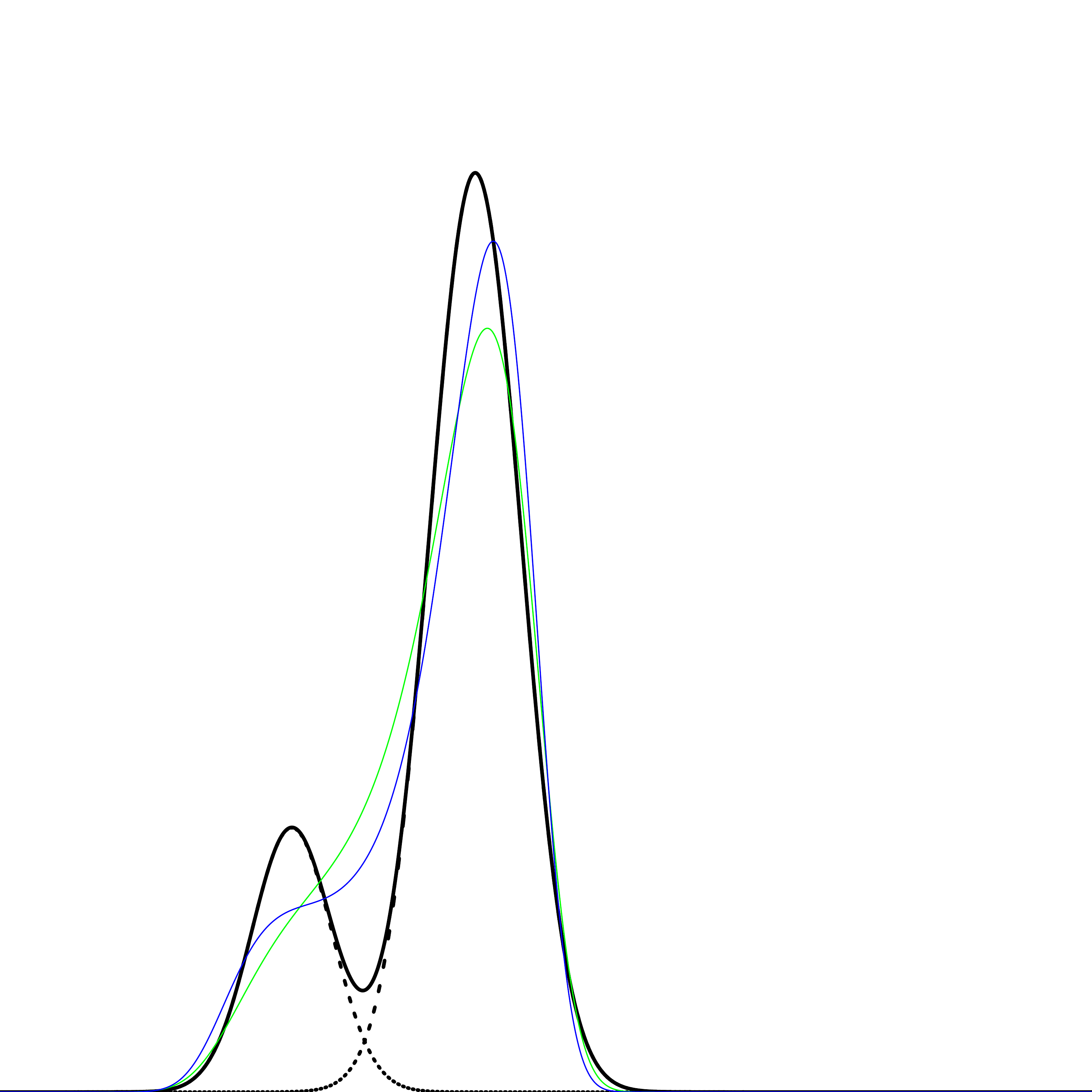}\\
Unimodal $2$-GMM  & Bimodal $2$-GMM
\end{tabular}
 
\caption{Two examples illustrating the conversion a GMM $m$ (black) of $k=2$ components (dashed black) into a pair of polynomial exponential densities of order $D=4$ $(p_{\bartheta_\SME},p^{\bareta_\MLE})$. 
PED $p_{\bartheta_\SME}$ is displayed in green and PED $p^{\bareta_\MLE}$ is displayed in blue. 
To display $p^{\bareta_\MLE}$, we first converted $\bareta_\MLE$ to $\tilde\bartheta_\MLE$ using an iterative linear system descent method (ILSDM), and we numerically estimated the normalizing factors $Z(\bartheta_\SME)$ and $Z(\bareta_\MLE)$ to display the normalized PEDs.   }%
\label{fig:convertpair}
\end{figure}

Then by converting both mixture $m$ and mixture $m'$ into  pairs of dually natural/moment parameterized unnormalized PEDs,  
 i.e., $m\rightarrow (q_{\bartheta_{\SME}},q_{\bareta_{\MLE}})$ and $m'\rightarrow (q_{\bartheta_{\SME}'},q_{\bareta_{\MLE}}')$, we approximate the  JD  between mixtures $m$ and $m'$ by using the four parameters of the PEDsas
\begin{equation}\label{eq:JDPEF}
D_J[m,m']\approx (\bartheta'_{\SME}-\bartheta_{\SME})^\top (\bareta'_{\MLE}-\bareta_{\MLE}).
\end{equation}

Let $\Delta_J$ denote the approximation formula obtained from the two pairs of PEDs:
\begin{equation}\label{eq:JDGMMPEF}
{\Delta}_J[p_{\theta_\SME},p^{\eta_\MLE};p_{\theta_\SME'},p^{\eta'_\MLE}] :=
 (\theta_\SME'-\theta_\SME)^\top (\eta'_\MLE -\eta_\MLE).
\end{equation}
Let ${\Delta}_J({\theta_\SME},{\eta_\MLE};{\theta_\SME'},{\eta'_\MLE}):={\Delta}_J[p_{\theta_\SME},p^{\eta_\MLE};p_{\theta_\SME'},p^{\eta'_\MLE}]$.
Then we have
$$
D_J[m,m']\approx \tilde{D}_J[m,m']:= {\Delta}_J({\theta_\SME},{\eta_\MLE};{\theta_\SME'},{\eta'_\MLE}).
$$

Note that $\Delta_J$ is not a proper divergence as it may be negative since in general $\bareta_\ML\not=\nabla F(\bartheta_\SM)$.
That is, $\Delta_J$  may not satisfy the law of the indiscernibles.
Approximation ${\Delta}_J$ is exact when $k_1=k_2=1$ with both $m$ and $m'$ belong to an exponential family.

We show experimentally in~\S\ref{sec:jdexp} that the $\tilde{D}_J$ heuristic yields fast approximations of the JD compared to the MC  baseline estimations by several order of magnitudes while approximating reasonably well the JD when the mixtures have a small number of modes.

For example, Figure~\ref{fig:ex} displays the unnormalized PEDs obtained for two Gaussian mixture models 
($k_1=10$ components and $k_2=11$ components) into PEDs of a PEF of order $D=8$.
The MC estimation of the JD with $s=10^6$ samples yields $0.2633\dots$ while the PED approximation of Eq.~\ref{eq:JDPEF} on corresponding PEFs yields $0.2618\ldots$ (the relative error is $0.00585\dots$ or about $0.585\ldots \%$).
It took about $2642.581$ milliseconds (with $s=10^6$ on a Dell Inspiron 7472 laptop) to MC estimate the JD while it took about $0.827$ milliseconds with the PEF approximation.
Thus we obtained a speed-up factor of about $3190$ (three orders of magnitude) for this particular example.
We report the mixtures and PEF conversions used in Figure~\ref{fig:ex} in Appendix~\ref{sec:maximaexample}.
Notice that when viewing Figure~\ref{fig:ex}, we tend to visually evaluate the dissimilarity using the total variation distance~\cite{nielsen2018guaranteed} (a metric distance):
$$
D_\TV[m,m']:=\frac{1}{2} \int |m(x)-m'(x)|\dx,
$$
rather than by a dissimilarity relating to the KLD.
Using Pinsker's inequality~\cite{Pinsker-1960,Pinsker-2003}, we have $D_J[m,m']\geq D_\TV[m,m']^2$ and $D_\TV[m,m']\in [0,1]$.
Thus large TV distance (e.g., $D_\TV[m,m']=0.1$) between mixtures may have small JD since Pinsker's inequality yields $D_J[m,m']\geq 0.01$.

Let us point out that our approximation heuristic is deterministic while the MC estimations are stochastic:
That is, each MC run (Eq.~\ref{eq:MCmix}) returns a different result, and a single MC run may yield a very bad approximation of the true Jeffreys divergence.

\def\ttt{0.33}
\begin{figure}%
\centering

\begin{tabular}{ccc}
\includegraphics[width=\ttt\columnwidth]{example-m12.pdf} &
\includegraphics[width=\ttt\columnwidth]{example-q1-order8.pdf} &
\includegraphics[width=\ttt\columnwidth]{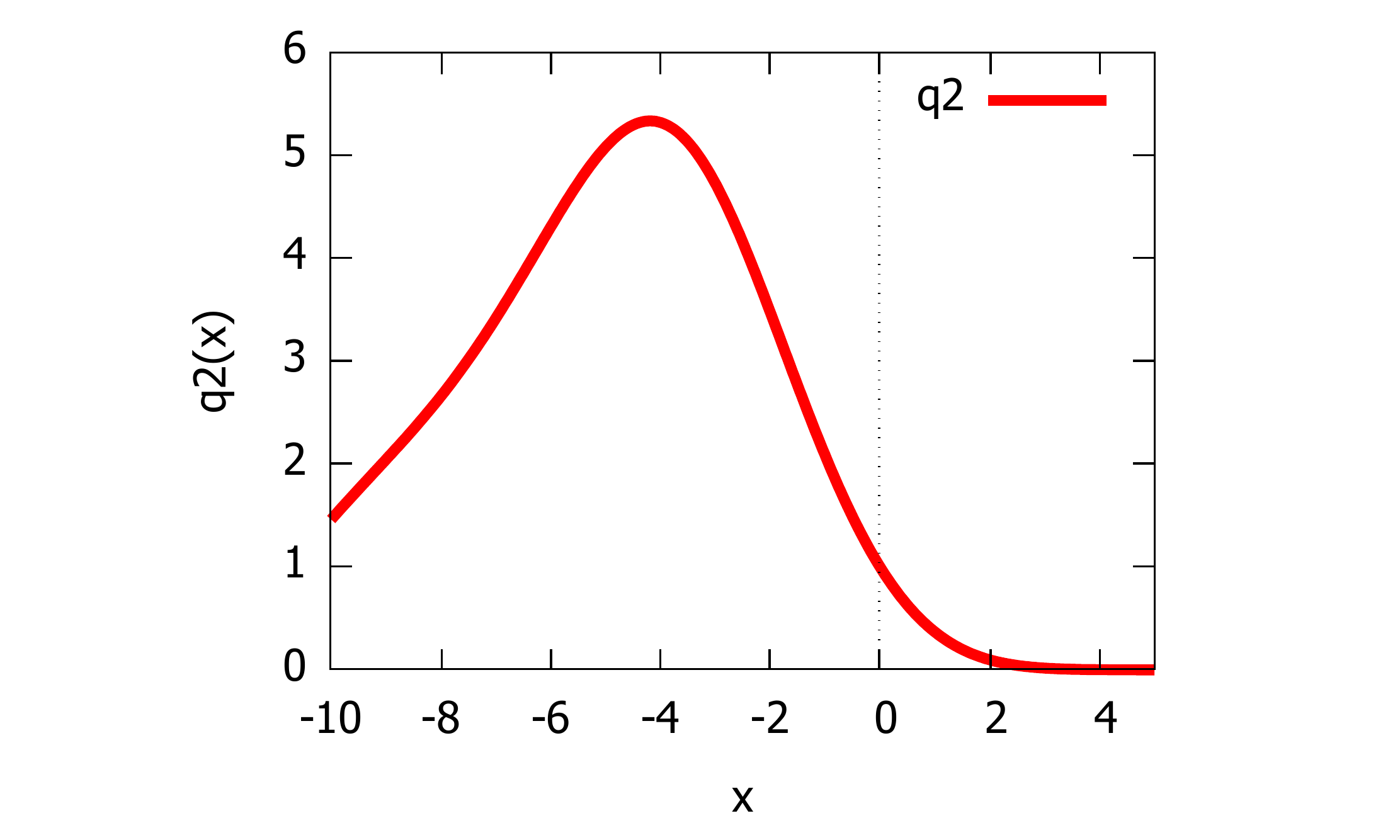}
\end{tabular}

\caption{Two mixtures $m_1$ (black) and $m_2$ (red) of $k_1=10$ components and $k_2=11$ components (left), respectively.
The unnormalized PEFs $q_{\bartheta_1}=\tilde{p}_{\bartheta_1}$ (middle) and $q_{\bartheta_2}=\tilde{p}_{\bartheta_2}$ (right) of order $D=8$. Jeffreys divergence (about $0.2634$) is approximated using PEDs within $0.6\%$ compared to the Monte Carlo estimate with a speed factor of about $3190$. Notice that displaying $p_{\bartheta_1}$ and $p_{\bartheta_2}$ on the same PDF canvas as the mixtures would require to calculate the partition functions $Z(\bartheta_1)$ and $Z(\bartheta_2)$ (which we do not in this figure). The PEDs  $q^{\bareta_1}$ and $q^{\bareta_2}$ of the pairs $(\bartheta_1,\bareta_1)$ and $(\bartheta_2,\bareta_2)$ parameterized in the moment space are not shown here.}%
\label{fig:ex}%
\end{figure}

We compare our fast heuristic $\tilde{D}_J[m,m']$ with two more costly methods relying on numerical procedures:
\begin{enumerate}
	\item Simplify GMMs $m_i$ into PEDs $p^{\eta_i^\MLE}$, and convert approximately the $\bar\eta_i^\MLE$'s into $\tilde\theta_i^\MLE$'s.
	Then approximate the Jeffreys divergence as
\begin{equation}
D_J[m_1,m_2]\simeq \tilde{\Delta}_J^\MLE[m_1,m_2]:=(\tilde\theta_2^\MLE-\tilde\theta_1^\MLE)^\top 	(\bar\eta_2^\MLE-\bar\eta_1^\MLE).
\end{equation}

\item Simplify GMMs $m_i$ into PEDs $p_{\bar\theta_i^\SME}$, and convert approximately the ${\bar\theta_i^\SME}$'s into $\tilde\eta_i^\SME$'s.
Then approximate the Jeffreys divergence as
\begin{equation}
D_J[m_1,m_2]\simeq  \tilde{\Delta}_J^\SME(m_1,m_2)=(\bar\theta_2^\SME-\bar\theta_1^\SME)^\top 	(\tilde\eta_2^\SME-\tilde\eta_1^\SME).
\end{equation}
\end{enumerate}

\subsection{Jeffreys divergence between GMMs of mixture families}\label{sec:wgmm}
A particular family of GMMs are GMMs sharing the same normal distribution components.
These families are called $w$-GMMs~\cite{wmixtures-2018}, and only the weights of prescribed Gaussian components are allowed to vary.
The family of $w$-GMMs with exactly $k$ prescribed and distinct components $p_i(x)=p_{\mu_i,\sigma_i}(x)$ form a mixture family $\{m_{\theta}\ :\ \theta\in\Theta \}$ of order $D=k-1$ in information geometry~\cite{IG-2016,MCIG-2019,EIG-2020}.
The underlying structure of a mixture family is a Bregman manifold (i.e., a Hessian manifold~\cite{Hessian-2007} with a single chart) with the Bregman generator being the negative entropy (a provably strictly convex function~\cite{MCIG-2019}).
Let 
$$
m_{\theta}=\sum_{i=1}^{k-1} \theta_i p_{\mu_i,\sigma_i}(x)+(1-\sum_{i=1}^{k-1} \theta_i) p_{\mu_0,\sigma_0}(x)
$$ 
and 
$$
m_{\theta'}=\sum_{i=1}^{k-1} \theta_i' p_{\mu_i,\sigma_i}(x)+(1-\sum_{i=1}^{k-1} \theta_i') p_{\mu_0,\sigma_0}(x)
$$ 
be two $w$-GMMs with $\theta$ and $\theta'$ belonging to the open $(k-1)$-dimensional simplex $\Theta=\Delta_{k-1}^\circ$.
The Bregman divergence for the negentropy generator amounts to calculate the KLD between the corresponding mixtures~\cite{EIG-2020}:
\begin{equation}
B_F(\theta:\theta')=D_\KL[m_{\theta}:m_{\theta'}],
\end{equation}
for $F(\theta):=-h[m_\theta]=\int_\calX m_\theta(x)\log m_\theta(x)\dx$, the differential negentropy.
It follows that the Jeffreys divergence between two $w$-GMMs is:
\begin{equation}
D_J[m_{\theta},m_{\theta'}]=(\theta'-\theta)^\top (\eta'-\eta), \label{eq:jdivmix}
\end{equation}
where the dual parameter $\eta$ is defined by~\cite{MCIG-2019}:
\begin{equation}
\eta=\nabla F(\theta)=\left[\begin{array}{c}\int (p_{\mu_1,\sigma_1}(x)-p_{\mu_0,\sigma_0}(x))(1+\log m_\theta(x))\dx\\
\vdots\\
\int (p_{\mu_i,\sigma_i}(x)-p_{\mu_0,\sigma_0}(x))(1+\log m_\theta(x))\dx\\
\vdots\\
\int (p_{\mu_{k-1},\sigma_{k-1}}(x)-p_{\mu_0,\sigma_0}(x))(1+\log m_\theta(x))\dx\\
\end{array}\right].
\end{equation}
These $\eta$-parameters are not available in closed-form, and need to be either numerically approximated or estimated via Monte Carlo methods~\cite{MCIG-2019}.

Furthermore, the $\alpha$-skewed Jensen-Shannon divergence between two $w$-GMMs amount to a $\alpha$-skewed Jensen divergence~\cite{wmixtures-2018}:

\begin{eqnarray}
D_{\mathrm{JS},{\alpha}}[m_{\theta_1}: m_{\theta_2}] &:=&
(1-\alpha) D_{\mathrm{KL}}\left[m_{\theta_1}: m_{\alpha}\right]+\alpha D_{\mathrm{KL}}\left[m_{\theta_2}: m_{\alpha}\right],\label{eq:jsmix}\\ 
&=& (1-\alpha) F\left(\theta_{1}\right)+\alpha F\left(\theta_{2}\right)-F\left((1-\alpha) \theta_{1}+\alpha \theta_{2}\right),\\
&=:& J_{F, \alpha}\left(\theta_{1}: \theta_{2}\right)
\end{eqnarray}
where $F(\theta)=-h[m_\theta]$ and
$$
m_{\alpha}(x):=(1-\alpha) m_{\theta_1}(x)+\alpha m_{\theta_2}(x)=m_{(1-\alpha)\theta_1+\alpha\theta_2},
$$
for any  $\alpha \in(0,1)$.

It has been proved that scaled $\alpha$-skewed Jensen divergences tend to Bregman divergences in limit cases $\alpha\rightarrow 0$ and $\alpha\rightarrow 1$~\cite{nielsen2011burbea,wmixtures-2018}:
\begin{equation} \label{eq:jdbd}
\lim _{\alpha \rightarrow 1^{-}} \frac{1}{\alpha(1-\alpha)} J_{F, \alpha}\left(\theta_{1}: \theta_{2}\right)
=B_{F}\left(\theta_{1}: \theta_{2}\right)=D_{\mathrm{KL}}\left[m_{1}: m_{2}\right].
\end{equation}
Thus we have $\lim _{\alpha \rightarrow 1^{-}} \frac{1}{\alpha(1-\alpha)} D_{\mathrm{JS},{\alpha}}[m_{\theta_1}: m_{\theta_2}]=D_{\mathrm{KL}}\left[m_{1}: m_{2}\right]$, as expected from Eq.~\ref{eq:jsmix}. Notice that the left-hand-side of Eq.~\ref{eq:jdbd} does not use explicitly the gradient $\nabla F(\theta)$ but only $F(\theta)=-h[m_\theta]$ while the right-hand-side requires the gradient $\nabla F(\theta)$.

\subsection{Contributions and paper outline}

Our contributions are summarized as follows:

\begin{itemize}
	\item We explain how to convert any continuous density $r(x)$ (including GMMs) into a polynomial exponential density in Section~\ref{sec:mm2pef} using  integral-based extensions of the 
	Maximum Likelihood Estimator~\cite{BN-2014} (MLE estimates in the moment parameter space $H$, 
	Theorem~\ref{thm:MLE} and Corollary~\ref{cor:msimplify}) and the Score Matching Estimator~\cite{hyvarinen-2005} (SME estimates in the natural parameter space $\Theta$, Theorem~\ref{thm:SMEmixPED}).   
	We show a connection of SME with the Moment Linear System Estimator~\cite{Cobb-1983} (MLSE) which is related to Stein's lemma for exponential families~\cite{EF-Hudson-1978} (see Lemma~\ref{lemma:Stein} in Appendix~\ref{sec:Stein}).
	
	\item We report a closed-form formula to evaluate the goodness-of-fit of a polynomial family density to a GMM in~\S\ref{sec:HyvDiv} using an extension of the Hyv\"arinen divergence~\cite{IG-2016} (Theorem~\ref{thm:hdpef}), and discuss the problem of model selection for choosing the order $D$ of the polynomial exponential family.
	
	\item We show how to approximate the Jeffreys divergence between GMMs using a pair of natural/moment parameter PED conversion, and present experimental results which displays a gain of several orders of magnitude of performance when compared to the vanilla Monte Carlo estimator in \S\ref{sec:jdexp}. We observe that the quality of the approximations depend on the number of modes of the GMMs~\cite{carreira2000mode}. However, calculating or counting the modes of a GMM is a difficult problem in its own~\cite{carreira2000mode}.
\end{itemize}

The paper is organized as follows: 
In Section~\ref{sec:mm2pef}, we show how to convert arbitrary probability density functions into polynomial exponential densities using integral-based Maximum Likelihood Estimator (MLE)   and Score Matching Estimator (SME).
We describe a Maximum Entropy method to convert iteratively moment parameters to natural parameters in \S\ref{sec:eta2theta}.
It is followed by Section~\ref{sec:HyvDiv} which shows how to calculate in closed-form the order-$2$ Hyv\"arinen divergence between a GMM and a polynomial exponential density. We use this criterion to perform model selection.
Section~\ref{sec:jdexp} presents our computational experiments which demonstrate a gain of several orders of magnitudes for GMMs with small number of modes. 
Finally, we conclude in Section~\ref{sec:concl}.

\section{Converting finite mixtures to exponential family densities}\label{sec:mm2pef}

We report two generic methods to convert a mixture $m(x)$ into a density $p_\theta(x)$ of an exponential family:
The first method extending the MLE in \S\ref{sec:mle} proceeds using the mean parameterization $\eta$ while the second method extending the SME in \S\ref{sec:convertnat} uses the natural parameterization of the exponential family.
We then describe how to convert the moments parameters to natural parameters (and vice-versa)  for polynomial exponential families in \S\ref{sec:ped:thetaeta}.
We show how to instantiate these generic conversion methods for GMMs: 
It requires to calculate in closed-form non-central moments of GMMs. 
The efficient computations of raw moments of GMMs is detailed in \S\ref{sec:rawmom}. 

\subsection{Conversion using the moment parameterization (MLE)}\label{sec:mle}

Let us recall that in order to estimate the moment or mean parameter $\hat{\eta}_\ML$ of a density belonging an exponential family  
$$
\calE_t:=\left\{p_\theta(x)=\exp\left(t(x)^\top\theta-F(\theta)\right)\right\}
$$ 
with sufficient statistic vector 
$t(x)=[t_1(x)\ \ldots\ t_D(x)]^\top$ from a i.i.d. sample set $x_1,\ldots, x_n$, the Maximum Likelihood Estimator (MLE)~\cite{Brown-1986,BN-2014} yields
\begin{eqnarray}
&& \max_{\theta} \prod_{i=1}^n p_\theta(x_i),\\
&\equiv& \max_{\theta} \sum_{i=1}^n \log p_\theta(x_i),\\
&=& \max_{\theta} E(\theta):=\left(\sum_{i=1}^n t(x_i)^\top\theta\right)-nF(\theta),\label{eq:estMLE}\\
&\Rightarrow& \hat\eta_\ML= \frac{1}{n}\sum_{i=1}^n t(x_i).\label{eq:MLEiid}
\end{eqnarray}

In statistics, Eq.~\ref{eq:estMLE} is called the estimating equation.
The MLE exists under mild conditions~\cite{BN-2014}, and is unique since the Hessian $\nabla^2 E(\theta)=\nabla^2 F(\theta)$ of the estimating equation is positive-definite (log-normalizers $F(\theta)$ are always strictly convex and real analytic~\cite{BN-2014}).
The MLE is consistent and asymptotically normally distributed~\cite{BN-2014}.
Furthermore, since the MLE satisfies the equivariance property~\cite{BN-2014}, we have $\hat\theta_\ML=\nabla F^*(\hat\eta_\ML)$, where $\nabla F^*$ denotes the gradient of the conjugate function $F^*(\eta)$ of the cumulant function $F(\theta)$ of the exponential family.
In general, $\nabla F^*$ is intractable for PEDs with $D\geq 4$ (except for the MEFs detailed in Appendix~\ref{sec:MEF}).

By considering the empirical distribution 
$$
p_e(x):=\frac{1}{n}\sum_{i=1}^s \delta_{x_i}(x),
$$ 
where $\delta_{x_i}(\cdot)$ denoting the Dirac distribution at location $x_i$, we can formulate the MLE problem as a minimum KLD problem between the empirical distribution and a density of the exponential family:
\begin{eqnarray*}
\min_\theta D_\KL[p_e:p_\theta] &=& \min -H[p_e]-E_{p_e}[\log p_\theta(x)],\\
&\equiv & \max_\theta \frac{1}{n} \sum_{i=1}^n \log p_\theta(x_i),
\end{eqnarray*}
since the entropy term $H[p_e]$ is independent of $\theta$.

Thus to convert an arbitrary smooth density $r(x)$ into a density $p_\theta$ of an exponential family $\calE_t$, 
we ask to solve the following minimization problem: 
$$
\min_{\theta\in\Theta} D_\KL[r:p_\theta].
$$

Rewriting the minimization problem as:
\begin{eqnarray*}
&& \min_\theta D_\KL[r:p_\theta] = -\int r(x)\log p_\theta(x)\dx+ \int r(x)\log r(x)\dx,\\
&\equiv& \min_\theta -\int r(x)\log p_\theta(x)\dx,\\
&=& \min_\theta \int r(x)(F(\theta)-\theta^\top t(x))\dx,\\
&=& \min_\theta \bar{E}(\theta)=F(\theta)-\theta^\top E_r[t(x)],
\end{eqnarray*}
we get 
\begin{equation}\label{eq:MLEdensity}
\bar\eta_\ML(r):= E_r[t(x)] =\int_\calX r(x)t(x) \dx.
\end{equation} 
The minimum is unique since $\nabla^2\bar{E}(\theta)=\nabla^2 F(\theta)\succ 0$ (positive-definite matrix).
This conversion procedure $r(x)\rightarrow p^{\bar\eta_\ML(r)}(x)$ can be interpreted as an integral extension of the MLE, hence the $\bar\dot$ notation in $\bareta_\ML$.
Notice that the ordinary MLE is $\hat\eta_\ML=\bar\eta_\MLE(p_e)$ obtained for the empirical distribution: $r=p_e$: $\bar\eta_\MLE(p_e)=\frac{1}{n}\sum_{i=1}^n t(x_i)$.
 
\begin{Theorem}\label{thm:MLE}
The best density $p^\bareta(x)$ of an exponential family $\calE_t=\{p_\theta\ :\ \theta\in\Theta\}$ minimizing the Kullback-Leibler divergence $D_\KL[r:p_\theta]$ between a density $r$ and a density $p_\theta$ of an exponential family $\calE_t$ is 
$\bareta=E_r[t(x)] =\int_\calX r(x)t(x) \dx$. 
\end{Theorem}

Notice that when $r=p_\theta$, we obtain $\bareta=E_{p_\theta}[t(x)]=\eta$, so that the method $\bareta_\MLE(r)$ is consistent (by analogy to the finite i.i.d. MLE case): $\bar\eta_\ML(p_\theta)=\eta=\nabla F(\theta)$.

The KLD right-sided minimization problem can be interpreted as an information projection~\cite{nielsen2018information} of $r$ onto $\calE_t$.
As a corollary of Theorem~\ref{thm:MLE}, we get:

\begin{Corollary}[Best right-sided KLD simplification of a mixture]\label{cor:msimplify}
The best right-sided KLD simplification of a homogeneous mixture of exponential families~\cite{mclachlan1988mixture} $m(x)=\sum_{i=1}^k w_i p_{\theta_i}(x)$ with $p_{\theta_i}\in\calE_t$, i.e. $\min_{\theta\in\Theta} D_\KL[m:p_\theta]$,
  into a single component $p^\eta(x)$ is given by $\eta=\hat\eta_\MLE(m)=E_m[t(x)]=\sum_{i=1}^k \eta_i=\bar\eta$.
\end{Corollary}

Eq.~\ref{eq:MLEdensity} allows us to greatly simplifies the proofs reported in~\cite{pelletier-2005,simplifykde-2013} for mixture simplifications which involved  the explicit use of the Pythagoras' theorem in the
 dually flat spaces of exponential families~\cite{IG-2016}.  
Figure~\ref{fig:simplGMM} displays the geometric interpretation of the best KLD simplification of a GMM with ambient space the probability space $(\bbR,\calB(\bbR),\mu_L)$ where $\mu_L$ denotes the Lebesgue measure and $\calB(\bbR)$ the Borel $\sigma$-algebra of $\bbR$.
 
\begin{figure}%
\centering
\includegraphics[width=0.85\columnwidth]{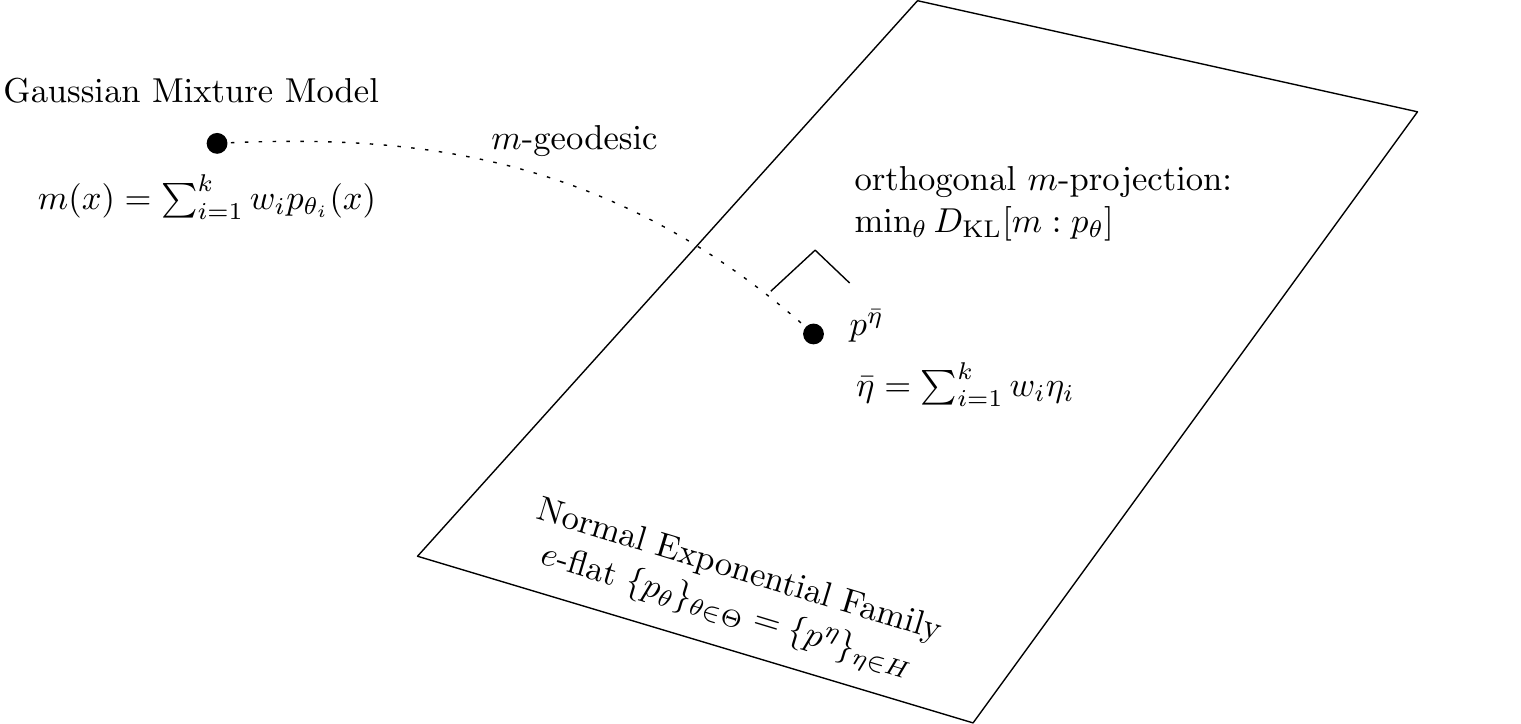}%

\caption{The best simplification of a GMM $m(x)$ into a single normal component $p_{\theta^*}$ ($\min_{\theta\in\Theta} D_\KL[m:p_\theta]=\min_{\eta\in H} D_\KL[m:p^\eta]$) is geometrically interpreted as the unique $m$-projection of $m(x)$ onto  the Gaussian family (a $e$-flat): We have $\eta^*=\bar\eta=\sum_{i=1}^k \eta_i$.}%
\label{fig:simplGMM}%
\end{figure}

Let us notice that Theorem~\ref{thm:MLE} yields an algebraic system for polynomial exponential densities, i.e., 
$E_m[x^i]=\bar\eta_i$ for $i\in\{1,\ldots, D\}$, 
to compute $\bar\eta_\MLE(m)$ for a given GMM $m(x)$ (since raw moments $E_m[x^i]$ are algebraic).
To contrast with this result,  the MLE of iid. observations is in general not an algebraic function~\cite{amendola2015maximum} but a transcendental function.

\subsection{Converting to a PEF using the natural parameterization (SME)}\label{sec:convertnat}

\subsubsection{Integral-based Score Matching Estimator (SME)}\label{sec:SM}
To convert density $r(x)$  to  an exponential density  with sufficient statistics $t(x)$, we can also  use the Score Matching Estimator~\cite{hyvarinen-2005,Hyvarinen-2007} (SME).
The score matching estimator minimizes the Hyv\"arinen divergence $D_H$ (Eq.~4 of~\cite{Hyvarinen-2007}):
$$
D_H[p:p_\theta]:=\frac{1}{2}\int \|\nabla_x \log p(x)-\nabla_x\log p_\theta(x)\|^2 p(x) \dx.
$$

That is, we convert a density $r(x)$ into an exponential family density $p_\theta(x)$ using the following minimizing problem: 
$$
\theta_\SM(r)=\min_{\theta\in\Theta} D_H[r:p_\theta].
$$

Beware that in statistics, the score $s_\theta(x)$ is defined by $\nabla_\theta\log p_\theta(x)$, but in score matching we refer to the ``data score''   defined by $\nabla_x\log p_\theta(x)$.
Hyv\"arinen~\cite{Hyvarinen-2007} gave an explanation of the naming ``score'' using a spurious location parameter.

\begin{itemize}
\item Generic solution: 
It can be shown that for exponential families~\cite{Hyvarinen-2007}, we get the following solution:
\begin{equation}\label{eq:SME}
\theta_\SM(r) = -\left(E_r[A(x)] \right)^{-1} \times \left( E_r[b(x)] \right),
\end{equation} 
where 
$$
A(x):=[t_i'(x)t_j'(x)]_{ij}
$$
is a $D\times D$ symmetric matrix, and 
$$
b(x)=[t_1''(x) \ldots\ t_D''(x)]^\top
$$ 
is a $D$-dimensional column vector.

\begin{Theorem}\label{thm:smef}
The best conversion of a density $r(x)$ into a density $p_\theta(x)$ of an exponential family minimizing the right-sided Hyv\"arinen divergence is 
$$
\theta_\SM(r) = -\left(E_r[[t_i'(x)t_j'(x)]_{ij}] \right)^{-1} \times \left( E_r[[t_1''(x) \ldots\ t_D''(x)]]^\top \right).
$$
\end{Theorem}

\item Solution instantiated for polynomial exponential families:

For polynomial exponential families of order $D$, we have  $t_i'(x)=ix^{i-1}$ and $t_i''(x)=i(i-1)x^{i-2}$, and therefore we have
$$
A_D=E_r[A(x)]= \left[ ij\, \mu_{i+j-2}(r)\right]_{ij},
$$ 
and 
$$
b_D=E_s[b(x)]= \left[j(j-1)\, \mu_{j-2}(r)\right]_j,
$$
where  $\mu_l(r):=E_r[X^l]$ denotes the $l$-th raw moment of distribution $X\sim r(x)$ (with the convention that $m_{-1}(r)=0$).
For a probability density function $r(x)$, we have $\mu_1(r)=1$.

Thus the integral-based SME of a density $r$ is:
\begin{equation}
\theta_\SM(r) = - \left(\left[ij \mu_{i+j-2}(r)\right]_{ij}\right)^{-1} \times \left[j(j-1) \mu_{j-2}(r)\right]_j.
\end{equation}

For example, matrix $A_4$ is 
$$
\left[
\begin{array}{cccc}
\mu_0 & 2\mu_1 & 3\mu_2 & 4\mu_3\cr
2\mu_1 & 4\mu_2 & 6\mu_3 & 8\mu_4\cr
3\mu_2 & 6\mu_3 & 9\mu_4 & 12\mu_5\cr
4\mu_3 & 8\mu_4 & 12\mu_5 & 16\mu_6
\end{array}
\right].
$$

\item Faster PEF solutions using Hankel matrices:
 
The method of Cobb et al.~\cite{Cobb-1983} (1983) anticipated the score matching method of Hyv\"arinen (2005).
It can be derived from Stein's lemma for exponential families (see Appendix~\ref{sec:Stein}).
The integral-based score matching method is consistent, i.e., if $r=p_\theta$ then $\bar\theta_\SM=\theta$: 
The probabilistic proof for $r(x)=p_e(x)$ is reported as Theorem~2 of~\cite{Cobb-1983}. 
The integral-based  proof is based on the property that arbitrary order partial mixed derivatives can be obtained from higher-order partial derivatives with respect to $\theta_1$~\cite{PEF-holonomic2016}: 
$$
\partial_1^{i_1}\ldots\partial_D^{i_D} F(\theta)=\partial_1^{\sum_{j=1}^D ji_j} F(\theta),
$$
where $\partial_i:=\frac{\partial}{\partial\theta_i}$.

The complexity of the direct SME method is $O(D^3)$ as it requires to inverse the $D\times D$-dimensional matrix $A_D$.

We show how to lower this complexity by reporting an equivalent method (originally presented in~\cite{Cobb-1983})  which relies on recurrence relationships between the moments of $p_\theta(x)$ for PEDs.
Recall that $\mu_l(r)$ denotes the $l$-th raw moment $E_r[x^l]$.

Let $A'=[a_{i+j-2}']_{ij}$ denote the $D\times D$ symmetric matrix with $a_{i+j-2}'(r)=\mu_{i+j-2}(r)$ (with $a_0'(r)=\mu_0(r)=1$),
 and $b'=[b_i]_i$ the $D$-dimensional vector with $b_i'(r)=(i+1)\mu_i(r)$.
We solve the system $A'\beta=b'$ to get $\beta={A'}^{-1}b'$.
We then get the natural parameter $\bartheta_\SME$ from the vector $\beta$ as
\begin{equation}
\bartheta_\SME=\left[
\begin{array}{c}
-\frac{\beta_1}{2}\\
\vdots\\
-\frac{\beta_i}{i+1}\\
\vdots\\
-\frac{\beta_D}{D+1}
\end{array}
\right].
\end{equation}

Now, if we inspect matrix $A_D'=\left[\mu_{i+j-2}(r)\right]$,  we find that matrix $A_D'$ is a Hankel matrix:
A Hankel matrix has constant anti-diagonals and can be inverted in quadratic-time~\cite{trench1965algorithm,heinig2011fast} instead of cubic time for a general $D\times D$ matrix. (The inverse of a Hankel matrix is a Bezoutian matrix~\cite{fuhrmann1986remarks}.)
Moreover, a Hankel matrix can be stored using linear memory (store $2D-1$ coefficients) instead of quadratic memory of regular matrices.

For example, matrix $A_4'$ is:
$$
A_4'=\left[
\begin{array}{cccc}
\mu_0 & \mu_1 & \mu_2 & \mu_3\cr
\mu_1 & \mu_2 & \mu_3 & \mu_4\cr
\mu_2 & \mu_3 & \mu_4 & \mu_5\cr
\mu_3 & \mu_4 & \mu_5 & \mu_6
\end{array}
\right],
$$
and requires only $6=2\times 4-2$ coefficients to be stored instead of $4\times 4=16$.
The order-$d$ moment matrix is
$$
A_d' :=[\mu_{i+j-2}]_{ij}=\left[
\begin{array}{cccc}
\mu_0 & \mu_1 & \ldots & \mu_d\cr
\mu_1 & \mu_2 & \ldots  & \vdots\cr
\vdots &   & \ddots & \vdots\cr
\mu_d & \ldots & \ldots & \mu_{2d}
\end{array}
\right],
$$
is a Hankel matrix stored using $2d+1$ coefficients:
$$
A_d'=:\Hankel(\mu_0,\mu_1,\ldots,\mu_{2d}).
$$

In statistics, those matrices $A'_d$ are called moment matrices and well-studied~\cite{lindsay1989determinants,lindsay1989moment,provost2009inversion}. 
The variance $\mathrm{Var}[X]$ of a random variable $X$ can be expressed as the determinant of the order-$2$ moment matrix:

$$
\mathrm{Var}[X]=E[(X-\mu)^2]=E[X^2]-E[X]^2=\mu_2-\mu_1^2=\mathrm{det}\left(\left[\begin{array}{cc}
1 & \mu_1 \cr
\mu_1 & \mu_2 
\end{array}\right]\right)\geq 0.
$$

This observation yields a generalization of the notion of variance to $d+1$ random variables: $X_1,\ldots, X_{d+1}\sim_{iid}  F_X \Rightarrow E\left[\prod_{j>i} (X_i-X_j)^2 \right] = (d+1)!\,\mathrm{det}(M_d)\geq 0$. The variance can be expressed as $E[\frac{1}{2}(X_1-X_2)^2]$ for $X_1, X_2\sim_{iid} F_X$. See~\cite{serfling2009approximation} (Chapter 5) for a detailed description related to $U$-statistics.

For GMMs $r$, the raw moments $\mu_l(r)$ to build matrix $A_D$ can be calculated in closed-form as explained in section \S\ref{sec:rawmom}.

\begin{Theorem}[Score matching GMM conversion]\label{thm:SMEmixPED}
The score matching conversion of a GMM $m(x)$ into a polynomial exponential density $p_{\theta_\SM(m)}(x)$ of order $D$ is obtained as
$$
\theta_\SM(m) = - \left(\left[ij\, m_{i+j-2}\right]_{ij}\right)^{-1} \times \left[j(j-1)\, m_{j-2}\right]_j,
$$
 where $m_i=E_{m}[x^i]$ denote the $i$th non-central moment of the GMM $m(x)$.
\end{Theorem}

\end{itemize}

\subsection{Converting numerically moment parameters from/to natural parameters}\label{sec:ped:thetaeta}

Recall that our fast heuristic approximates the Jeffreys divergence by
$$
\tilde{D}_J[m,m']:=(\bar\theta_\SME(m')-\bar\theta_\SME(m))^\top (\bar\eta_\MLE(m') -\bar\eta_\MLE(m)).
$$

Because $F$ and $\nabla F^*$ are not available in closed form (except for the case $D=2$ of the normal family), we cannot get $\theta$ from a given $\eta$ (using $\theta=\nabla F^*(\eta)$) nor $\eta$ from a given $\theta$ (using $\eta=\nabla F(\theta)$).

However, provided that we can approximate numerically   $\tilde\eta\simeq\nabla F(\theta)$ and $\tilde\theta\simeq\nabla F^*(\eta)$, we also consider these two approximations for the Jeffreys divergence:
$$
\tilde{\Delta}_J^\MLE[m_1,m_2]:=(\tilde\theta_2^\MLE-\tilde\theta_1^\MLE)^\top 	(\bar\eta_2^\MLE-\bar\eta_1^\MLE),
$$
and
$$
\tilde{\Delta}_J^\SME[m_1,m_2]=(\bar\theta_2^\SME-\bar\theta_1^\SME)^\top 	(\tilde\eta_2^\SME-\tilde\eta_1^\SME).
$$

In this section, we show how to numerically estimate  $\tilde\theta^\MLE\simeq \nabla F(\bar\eta^\MLE)$ from $\bar\eta^\MLE$ in \S\ref{sec:eta2theta}.
Next, in \S\ref{sec:theta2eta}, we show how to stochastically estimate $\tilde\eta^\SME\simeq \nabla F^*(\bar\theta^\SME)$.

\subsubsection{Converting moment  parameters to natural parameters using maximum entropy}\label{sec:eta2theta}

Let us report the iterative approximation technique of~\cite{MaxEnt-1992} (which extended the method described in~\cite{zellner1988calculation}) based on solving a maximum entropy problem (MaxEnt problem).
This method will be useful when comparing our fast heuristic $\tilde{D}_J[m,m']$ with the approximations 
$\tilde{\Delta}_J^\MLE[m,m']$ and  $\tilde{\Delta}_J^\SME[m,m']$.
 
The density $p_\theta$ of any exponential family can be characterized as a maximum entropy distribution given the $D$ moment constraints
 $E_{p_\theta}[t_i(x)]=\eta_i$: Namely,
$\max_p h(p)$ subject to the $D+1$ moment constraints $\int t_i(x) p(x)\dx=\eta_i$ for $i\in\{0,\ldots, D\}$, 
where we added by convention $\eta_0=1$ and $t_0(x)=1$ (so that $\int p(x)\dx=1$). 
The solution of this MaxEnt problem~\cite{MaxEnt-1992} is $p(x)=p_\lambda$ where $\lambda$ are the $D+1$ Lagrangian parameters.
Here, we adopt the the following canonical parameterization of the densities of an exponential family:
$$
p_\lambda(x):=\exp\left(-\sum_{i=0}^D \lambda_i t_i(x)\right).
$$
That is, $F(\lambda)=\lambda_0$ and $\lambda_i=-\theta_i$ for $i\in\{1,\ldots, D\}$.
Parameter $\lambda$ is a kind of augmented natural parameter which includes the log-normalizer in its first coefficient.

Let $K_i(\lambda):=E_{p_\theta}[t_i(x)]=\eta_i$ denote the set of $D+1$ non-linear equations for $i\in\{0,\ldots, D\}$. 
The Iterative Linear System Method~\cite{MaxEnt-1992} (ILSM) converts iteratively $p^\eta$ to $p_\theta$. 
We initialize $\lambda^{(0)}$ to $\bartheta_\SME$ (and calculate numerically $\lambda_0^{(0)}=F(\bartheta_\SME)$).

At iteration $t$ with current estimate $\lambda^t$, we use the following first-order Taylor approximation:
$$
K_i(\lambda)\approx K_i(\lambda^{(t)})+(\lambda-\lambda^{(t)})\nabla K_i(\lambda^{(t)}).
$$
Let $H(\theta)$ denote the $(D+1)\times (D+1)$ matrix:
$$
H(\lambda):=\left[\frac{\partial K_i(\lambda)}{\partial\lambda_j}\right]_{ij}.
$$
We have
$$
H_{ij}(\lambda)=H_{ji}(\lambda)=-E_{p_\theta}[t_i(x)t_j(x)].
$$

We update as follows:
\begin{equation}\label{eq:update}
\lambda^{(t+1)}=\lambda^{(t)}+H^{-1}(\lambda^{(t)})
\left[\begin{array}{c}
\eta_0-K_0(\lambda^{(t)})\\ \vdots\\ \eta_D-K_D(\lambda^{(t)}) \end{array}\right].
\end{equation}

For a PEF of order $D$, we have
$$
H_{ij}(\lambda)=-E_{p_\theta}[x^{i+j-2}]=-\mu_{i+j-2}(p_\theta).
$$
This yields a moment matrix $H_\lambda$ (Hankel matrix) which can be inverted in quadratic time~\cite{heinig2011fast}.
In our setting, the moment matrix is invertible because $|H|>0$, see~\cite{karlin1968total}.

Let $\tilde{\lambda}_T(\eta)$ denote  $\theta^{(T)}$ after $T$ iterations (retrieved from  $\lambda^{(T)}$), and let be the corresponding natural parameter of the PED.
We have the following approximation of the JD:
$$
D_J[m,m']\approx (\tilde{\theta}_T(\eta')-\tilde{\theta}_T(\eta))^\top (\eta'-\eta).
$$

The method is costly because we need to numerically calculate $\mu_{i+j-2}(p_\theta)$ and the $K_i$'s (e.g., univariate Simpson integrator).
Another potential method consists in estimating these expectations using acceptance-rejection sampling~\cite{vonNeumann1951,flury1990acceptance} (see Appendix~\ref{sec:ar}).
We may also consider the holonomic gradient descent~\cite{PEF-holonomic2016}.
Thus the conversion $\eta\rightarrow\theta$ method is costly. 
Our heuristic $\tilde\Delta_J$  bypasses this costly moment-to-natural parameter conversion by  converting each mixture $m$ to a pair $(p_{\theta_\SME},p_{\eta_\MLE})$ of PEDs parameterized in the natural and moment parameters
 (i.e., loosely speaking, we untangle these dual parameterizations).

\subsubsection{Converting natural parameters to moment parameters}\label{sec:theta2eta}

Given a PED $p_\theta(x)$, we ask to find its corresponding moment parameter $\eta$ (i.e., $p_\theta=p^\eta$).
Since $\eta=E_{p_\theta}[t(x)]$, we sample $s$ iid. variates $x_1,\ldots, x_s$ from $p_\theta$ using acceptance-rejection sampling~\cite{vonNeumann1951,flury1990acceptance} or any other Markov chain Monte Carlo technique~\cite{rohde2014mcmc}, and estimate $\hat{\eta}$ as:
$$
\hat\eta=\frac{1}{s}\sum_{i=1}^s t(x_i).
$$

\subsection{Raw non-central moments of normal distributions and GMMs}\label{sec:rawmom}

In order to implement the MLE or SME Gaussian mixture conversion procedures, we need to calculate the raw moments 
of a Gaussian mixture model.
The $l$-th moment raw moment $E[Z^l]$ of a standard normal distribution $Z\sim N(0,1)$ is $0$ when $l$ is odd (since the normal standard density is an even function) and
$(l-1)!!=2^{-\frac{l}{2}}\frac{l!}{(l/2)!}$ when $l$ is even,
where  $n!!=\sqrt{\frac{2^{n+1}}{\pi}}\Gamma(\frac{n}{2}+1)=\prod_{k=0}^{\ceil{\frac{n}{2}}-1} (n-2k)$ is the {\em double factorial} (with $(-1)!!=1$ by convention).
Using the binomial theorem, we deduce that a normal distribution $X=\mu+\sigma Z$ has finite moments:
$$
\mu_l(p_{\mu,\sigma})=E_{p_{\mu,\sigma}}[X^l]=E[(\mu+\sigma Z)^l]=E[(\mu+\sigma Z)^l]=\sum_{i=0}^l \binom{l}{i}\mu^{l-i}\sigma^i E[Z^i].
$$
That is, we have
\begin{equation}\label{eq:momz0}
\mu_l(p_{\mu,\sigma})= \sum_{i=0}^{\floor{\frac{l}{2}}} \binom{l}{i}(2i-1)!! \mu^{l-2i}\sigma^{2i},
\end{equation}
where $n!!$ denotes the double factorial:
$$
n!!=\prod_{k=0}^{\left\lceil\frac{n}{2}\right]-1}(n-2 k)=\left\{
\begin{array}{ll}
 \prod_{k=1}^{\frac{n}{2}}(2 k) & \mbox{$n$ is even},\\
 \prod_{k=1}^{\frac{n+1}{2}}(2 k-1)& \mbox{$n$ is odd}.
\end{array}
\right.
$$

By the linearity of the expectation $E[\cdot]$, we deduce the $l$-th raw moment of a GMM  $m(x)=\sum_{i=1}^k w_i p_{\mu_i,\sigma_i}(x)$: 
$$
\mu_l(m)=\sum_{i=1}^k w_i \mu_l(p_{\mu_I,\sigma_i}).
$$

Notice that by using~\cite{Barr-1999}, we can extend this formula to truncated normals and GMMs.
Thus computing the first $O(D)$ raw moments of a GMM with $k$ components can be done in $O(kD^2)$ using the Pascal triangle method for computing the binomial coefficients. See also~\cite{amendola2016moment}.

In general, the raw moments $\mu_l(\theta):=E_{p_\theta}[t(x)^l]$ of a probability density function $p_\theta$ belonging to an EF can be calculated from the $l$-fold (partial) derivatives of the moment generating function~\cite{Cumulant-1999} (MGF).
The MGF of a random variable $X$ is defined by $m_X(u):=E[e^{u^\top X}]$.
When $X\sim p_\theta(x)$, we get the MGF $m_{\theta}(u):=E[e^{u^\top t(X)}]$ which admits the following simple expression:
$$
m_{\theta}(u)=\exp(F(\theta+u)-F(\theta)).
$$

For example, we check that we have 
\begin{eqnarray*}
\mu_1(\theta) &:=& E_{p_\theta}[t(x)],\\
&=&\left.(\nabla_u m_{\theta}(u))\right|_{u=0},\\
&=& \left. (\nabla_u F(\theta+u)) \exp(F(\theta+u)-F(\theta))\right|_{u=0}\\
&=&\nabla_\theta F(\theta)=:\eta,
\end{eqnarray*}
the moment parameter, since $\left.\nabla_u F(\theta+u)\right|_{u=0}=\nabla_\theta F(\theta)$ and $\left.\exp(F(\theta+u)-F(\theta))\right|_{u=0}=e^0=1$.

For uniorder exponential families, we have $E_{p_\theta}[t(x)^l]=m_{\theta}^{(l)}(0)$.
For multiparameter exponential families, we have~\cite{Cumulant-1999}:
$$
E_{p_\theta}[t_1(x)^{n_1}\ldots t_D(x)^{n_D}]=
\left.\frac{\partial^{\sum_{i=1}^D n_i}}{\partial u_1^{n_1}\ldots \partial u_D^{n_D}} m_{\theta(u)}\right|_{u=0}.
$$

Thus we get the following proposition:

\begin{Proposition}
For a polynomial exponential family $p_\theta(x)=\exp(\sum_{i=1}^D \theta_ix^i-F(\theta))$ of order $D$ with $t_1(x)=x$, we get the $l$th moment $E_{p_\theta}[x^l]$ expressed in the natural coordinates $\theta$ as
$$
\mu_l(\theta)=\left.\frac{\partial^l}{\partial u_1^l} m_{\theta}(u)\right|_{u=0}.
$$
\end{Proposition}

\begin{Example}
For example, consider the univariate biorder polynomial exponential family commonly called the family of normal distributions.
The log-normalizer $F_N(\theta)$~\cite{nielsen2009statistical} is 
$$
F_N(\theta)=-\frac{\theta_1^2}{4\theta_2}+\frac{1}{2}\log\left(-\frac{\pi}{\theta_2}\right),
$$ 
for the  natural parameterization 
$$
\theta=(\theta_1,\theta_2)=\left(\frac{\mu}{\sigma^2},-\frac{1}{2\sigma^2}\right),
$$
with  the sufficient statistics is $t(x)=[t_1(x)\ t_2(x)]^\top=[x\ x^2]^\top$.
The MGF is $m_{\theta}(u)=\exp(F_N(\theta+u)-F_N(\theta))$.
We have 
$$
E[t_1(x)^2]=E[x^2]=\left.\left(\frac{\partial^2}{\partial u_1^2} m_{\theta}(u)\right)\right|_{u=0}=\frac{\theta_1^2}{4\theta_2^2}-\frac{1}{2\theta_2}.
$$
Plugging $\theta_1=\frac{\mu}{\sigma^2}$ and $\theta_2=-\frac{1}{2\sigma^2}$ in the above right-hand-side equation, we get
$$
E_{p_{\mu,\sigma}}[x^2]=\mu^2+\sigma^2.
$$
In general,  we get for normal distributions:
$$
E_{p_{\mu,\sigma}}[X^l]=\left.\left(\frac{\partial^l}{\partial u_1^l} m_{\theta}(u)\right)\right|_{u=0}.
$$
Thus using symbolic computing, we can automatically calculate any order moment in closed-form.
For example, following code in {\sc Maxima} (a free computer algebra system available online at \url{https://maxima.sourceforge.io/}) calculates the 5th raw normal moment $E_{p_{\mu,\sigma}}[X^5]$:

\begin{lstlisting}[backgroundcolor = \color{lightgray}]
F(theta1,theta2):=-(theta1*theta1/(4*theta2))+(1/2)*log(-%pi/theta2);
MGF(theta1,theta2,u1,u2):=exp(F(theta1+u1,theta2+u2)-F(theta1,theta2));
derivative(MGF(theta1,theta2,u1,u2),u1,5);
ev(%, u1=0,u2=0);
ev(%,theta1=(mu/(sigma**2)),theta2=-1/(2*(sigma**2)));
ratsimp(%);
\end{lstlisting}

Executing the above {\sc Maxima} code yields
$$
E_{p_{\mu,\sigma}}[x^5]=\mu^5+10\mu^3\sigma^2+15\mu\sigma^4.
$$

We can also check that since $t_1(x)=x$ and $t_2(x)=x^2=t_1(x)^2$, we have:
\begin{eqnarray}
E_{p_{\mu,\sigma}}[x^{2k}]&=& E[t_1^{2k}(x)] = \left.\left(\frac{\partial^{2k}}{\partial u_1^{2k}} m_{\theta}(u)\right)\right|_{u=0},\\
&=& E[t_2^{k}(x)] = \left.\left(\frac{\partial^k}{\partial u_2^{k}} m_{\theta}(u)\right)\right|_{u=0}.
\end{eqnarray}

We may also directly calculate symbolically the Gaussian moments in {\sc Maxima} as follows:
\begin{lstlisting}[backgroundcolor = \color{lightgray}]
normal(x,mu,sigma) := (1.0/(sqrt(2*%pi)*sigma))*exp(-((x-mu)**2)/(2*sigma**2) );
assume(s>0);
maxOrder:18;
for i:1 while (i<=maxOrder)  
do( m[i]: ratsimp(integrate(normal(x,m,s)*(x**i),x,minf,inf)), print(m[i]));
\end{lstlisting}

\end{Example}

Notice that we can rewrite the MGF as 
$$
m_{\theta}(u)=\exp(F(\theta+u)-F(\theta))=\exp\left(\int_\theta^{\theta+u} \nabla F(u)\mathrm{d}u\right).
$$
This highlights that the moment parametrization $\eta=\nabla F(\theta)=E_{p_\theta}[t(X)]$ specifies the MGF~\cite{sampson1975characterizing} (i.e., dual parameterizations $m_\theta(u)=m^\eta(u)$).

\section{Goodness-of-fit between GMMs and PEDs:  Higher order Hyv\"arinen divergences}\label{sec:HyvDiv}

Once we have converted a GMM $m(x)$ into an unnormalized PED $q_{\theta_m}(x)=\tilde{p}_{\theta_m}(x)$, we would like to evaluate the quality of the conversion, i.e., $D[m(x):q_{\theta_m}(x)]$, using a statistical divergence $D[\cdot:\cdot]$.
This divergence shall allow us to perform {\em model selection} by choosing the order $D$  of the PEF so that $D[m(x):p_\theta(x)]\leq \epsilon$ for $\theta\in\bbR^D$, 
where $\epsilon>0$ is a prescribed threshold.
Since PEDs have computationally intractable normalization constants, we consider a {\em right-sided projective divergence}~\cite{IG-2016} $D[p:q]$  
that satisfies $D[p:\lambda q]=D[p:q]=D[p:\tilde{q}]$ for {\em any} $\lambda>0$. 
For example, we may consider the $\gamma$-divergence~\cite{gammadiv-2008} that is a {\em two-sided projective divergence}: 
$D_\gamma[\lambda p:\lambda' q]=D[p:q]=D[\tilde{p}:\tilde{q}]$ for any $\lambda,\lambda'>0$ and converge to the KLD when $\gamma\rightarrow 0$. 
However, the $\gamma$-divergence between a mixture model and an unnormalized PEF does not yield a closed-form formula.
Moreover, the $\gamma$-divergence between two unnormalized PEDs is expressed using the log-normalizer function $F(\cdot)$ that is computationally intractable~\cite{NN-2016}.

In order to a get a closed-form formula for a divergence between a mixture model and an unnormalized PED, 
we consider  the order-$\alpha$ (for $\alpha>0$) Hyv\"arinen divergence~\cite{IG-2016}  as follows:
\begin{equation}
D_{H,\alpha}[p:q] := \int p(x)^\alpha \left(\nabla_x \log p(x) - \nabla_x \log q(x) \right)^2 \dx,\quad \alpha>0.
\end{equation}
The Hyv\"arinen divergence~\cite{IG-2016} has also been called the Fisher divergence~\cite{hyvarinen-2005,VariationalFisherDivergence-2019,RLFisher-2021,FisherAE-2021}.
The  Hyv\"arinen divergence is also known as half of the relative Fisher information in the optimal transport community (Equation~(8) of~\cite{otto2000generalization} or Equation (2.2) in~\cite{toscani1999entropy}), where it is defined for two measures $\mu$ and $\nu$ as follows:
$$
I[\mu:\nu]:=\int_\calX \left\|\nabla \log \frac{\dmu}{\dnu}\right\|^2\dmu
= 4\int_\calX  \left\|\nabla\sqrt{\frac{\dmu}{\dnu}}\right\|^2\dnu.
$$
Notice that when $\alpha=1$, $D_{H,1}[p:q]=D_H[p:q]$, the ordinary Hyv\"arinen divergence~\cite{hyvarinen-2005}.

The Hyv\"arinen divergences $D_{H,\alpha}$ is a right-sided projective divergence~\cite{nielsen2017holder} 
which satisfies $D_{H,\alpha}[p:q]=D_{H,\alpha}[p:\lambda q]$ for any $\lambda>0$. That is, we have $D_{H,\alpha}[p:q]=D_{H,\alpha}[p:\tilde{q}]$.
Thus we have $D_{H,\alpha}[m:p_\theta]=D_{H,\alpha}[m:q_\theta]$ for a unnormalized PED $q_\theta=\tilde{p}_\theta$.
For statistical estimation, it is enough to have a sided projective divergence since we need to evaluate the goodness of fit between the (normalized) empirical distribution $p_e$ and the (unnormalized) parameteric density.

For univariate distributions, $\nabla_x \log p(x)=\frac{p'(x)}{p(x)}$, and $\frac{p'(x)}{p(x)}=\frac{\tilde{p}'(x)}{\tilde{p}(x)}$ where $\tilde{p}(x)$ is
the unnormalized model. 
For PEDs with homogeneous polynomial $P_\theta(x)$, we have $\frac{p'(x)}{p(x)}= (\log P_\theta(x))'= \sum_{i=1}^D i\theta_i x^{i-1}$.

\begin{Theorem}\label{thm:hdpef}
The Hyv\"arinen divergence $D_{H,2}[m:q_\theta]$ of order $2$ between a Gaussian mixture $m(x)$ and a polynomial exponential family density $q_\theta(x)$ is available in closed form.
\end{Theorem}

\begin{proof}
We have $D_{H,2}[m:q]=\int m(x)^2 \left( \frac{m'(x)}{m(x)} - \sum_{i=1}^D i\theta_i x^{i-1} \right)^2 \dx$ with 
$$
m'(x)= -\sum_{i=1}^k w_i\frac{x-\mu_i}{\sigma_i^2}p(x_i;\mu_i,\sigma_i),
$$ 
denoting the derivative of the Gaussian mixture density $m(x)$. 
It follows that:

$$
D_{H,2}[m:q]=\int m'(x)\dx  -2\sum_{i=1}^D i\theta_i \int x^{i-1}m'(x)m(x) \dx     + \sum_{i,j=1}^D  ij\theta_i\theta_j \int x^{i+j-2} m(x)^2  \dx,
$$
where 
$$
\int x^i m'(x)m(x)\dx= -\sum w_a w_b\int  \frac{x-\mu_a}{\sigma_a^2} x^i p(x;\mu_a,\sigma_a)p(x;\mu_b,\sigma_b)\dx.
$$

Therefore we have
$$
D_{H,2}[m:q]=\int m'(x)\dx  -2\sum_{i=1}^D i\theta_i \int x^{i-1}m'(x)m(x) \dx 
   + \sum_{i,j=1}^D  ij\theta_i\theta_j \int x^{i+j-2} m(x)^2  \dx
	$$
with $m'(x)=-\sum w_a  \frac{x-\mu_a}{\sigma_a^2} p(x;\mu_a,\sigma_a)$.

Since
$p_a(x)p_b(x) = \kappa_{a,b} p(x;\mu_{ab},\sigma_{ab})$,
with 
\begin{eqnarray*}
\mu_{ab} &=&\sigma_a^2\sigma_b^2(\sigma_b^2\mu_a+\sigma_a^2\mu_b),\\
\sigma_{ab}&=&\frac{\sigma_a\sigma_b}{\sqrt{\sigma_a^2+\sigma_b^2}},\\
\kappa_{a,b}&=&\exp(F(\mu_{ab},\sigma_{ab})-F(\mu_a,\sigma_a)-F(\mu_b,\sigma_b)),
\end{eqnarray*}
and
$$
F(\mu,\sigma)=\frac{\mu^2}{2\sigma^2}+\frac{1}{2}\log(2\pi\sigma^2),
$$ 
the log-normalizer of the Gaussian exponential family~\cite{IG-2016}. 

Therefore we get
$$
\int p_a(x)p_b(x) x^l \dx=\kappa_{a,b}m_l(\mu_{ab},\sigma_{ab}).
$$

Thus the Hyv\"arinen  divergence $D_{H,2}$ of order $2$ between a GMM and a PED is available in closed-form.
\end{proof}

For example, when $k=1$ (i.e., mixture $m$ is a single Gaussian $p_{\mu_1,\sigma_1}$) and $p_\theta$ is a normal distribution (i.e., PED with $D=2$, $q_\theta=p_{\mu_2,\sigma_2}$), we obtain the following formula for the order-$2$ Hyv\"arinen divergence:
$$
D_{H,2}[p_{\mu_1,\sigma_1}:p_{\mu_2,\sigma_2}] = \frac{(\sigma_1^2-\sigma_2^2)^2+2(\mu_2-\mu_1)^2\sigma_1^2)}{8\sqrt{\pi} \sigma_1^3\sigma_2^4}.
$$

\section{Experiments: Jeffreys divergence between mixtures}\label{sec:jdexp}

In this section, we evaluate our heuristic to approximate the Jeffreys divergence between two mixtures $m$ and $m'$:
$$
\tilde{D}_J[m,m']:=(\bar\theta_\SME(m')-\bar\theta_\SME(m))^\top (\bar\eta_\MLE(m') -\bar\eta_\MLE(m)).
$$
Recall that stochastically estimating the JD between $k$-GMMs  with Monte Carlo sampling using $s$ samples (i.e., $\hat{D}_{J,s}[m:m']$) requires $\tilde O(ks)$ and is not deterministic. That is, different MC runs  yield fluctuating values which may be fairly different. 
In comparison, approximating $D_J$ by $\tilde{D}_J$ using $\Delta_J$ by converting mixtures to $D$-order PEDs require to $O(kD^2)$ time to compute the raw moments and $O(D^2)$ time to invert a Hankel moment matrix. 
Thus by choosing $D=2k$, we get a deterministic $O(k^3)$ algorithm which is faster than the MC sampling when $k^2\ll s$.
Since there are at most $k$ modes for a $k$-GMM, we choose order $D=2k$ for the PEDs.

To get quantitative results on the performance of our heuristic $\tilde{D}_J$, we build random GMMs with $k$ components as follows: 
$m(x)=\sum_{i=1}^k w_i p_{\mu_i,\sigma_i}(x)$, where $w_i\sim U_i$, $\mu_i\sim -10+10U_1'$ and $\sigma_i\sim 1+U_2'$, where the $U_i$'s and  $U_1'$ and $U_2'$ are independent uniform distributions on $[0,1)$. The mixture weights are then normalized to sum up to one.
For each value of $k$, we make $1000$ trial experiments to gather statistics, and use $s=10^5$ for evaluating the Jeffreys divergence $\hat{D}_J$ by Monte Carlo samplings. 
We denote by $\mathrm{error}:=\frac{|\hat{D}_J-\Delta_J|}{\hat{D}_J}$ the error of an experiment.
Table~\ref{tab:randGMM} presents the results of the experiments for $D=2k$: 
The table displays the average error, the maximum error (minimum error is very close to zero, of order $10^{-5}$), and the speed-up obtained 
by our heuristic $\Delta_J$. Those experiments were carried on a Dell Inspiron 7472 laptop (equipped with an Intel(R) Core(TM) i5-8250U CPU  at 1.60 GHz).

\begin{table}
\caption{Comparison of $\tilde{\Delta}_J(m_1,m_2)$ with $\hat{D}_J(m_1,m_2)$ for random GMMs.\label{tab:randGMM}}
\begin{center}
\begin{tabular}{ll|lll}\hline
$k$ & $D$ & average error & maximum error & speed-up\\ \hline
2 & 4 & 0.1180799978221536 & 0.9491425404132259 & 2008.2323536011806\\
3 & 6 & 0.12533811294546526 & 1.9420608151988419 & 1010.4917042114389\\
4 & 8 & 0.10198448868508087 & 5.290871019594698 & 474.5135294829539\\
5 & 10 &  0.06336388579897352 & 3.8096955246161848 & 246.38780782640987\\
6 & 12 &  0.07145257192133717 & 1.0125283726458822 & 141.39097909641052\\
7 & 14 & 0.10538875853178625 & 0.8661463142793943 & 88.62985036546912\\
8 & 16 & 0.4150905507007969 & 0.4150905507007969 & 58.72277575395611\\ \hline
\end{tabular}
\end{center}

\end{table}

Notice that the quality of the approximations of $\tilde{D}_J$ depend on the number of modes of the GMMs.
However, calculating the number of modes is  difficult~\cite{ModeGMM-2000,amendola2020maximum} even for simple cases~\cite{aprausheva2006bounds,aprausheva2013exact}.

Figure~\ref{fig:exp} displays several experiments of converting mixtures to pairs of PEDs to get approximations of the Jeffreys divergence.

\def\ttt{0.35}
\begin{figure}%
\centering
\begin{tabular}{lcc}
$D=2$ &\includegraphics[width=\ttt\columnwidth]{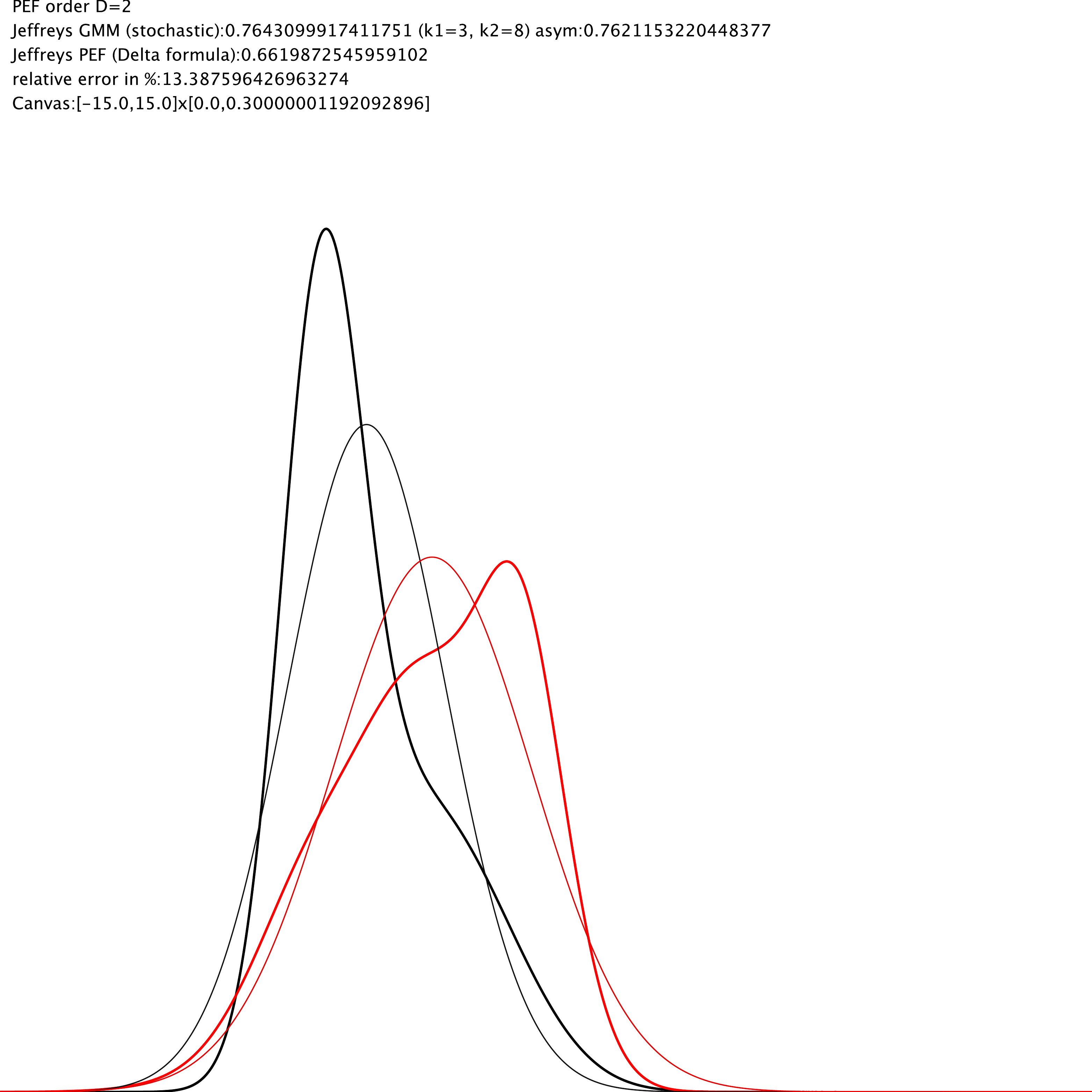}&
\includegraphics[width=\ttt\columnwidth]{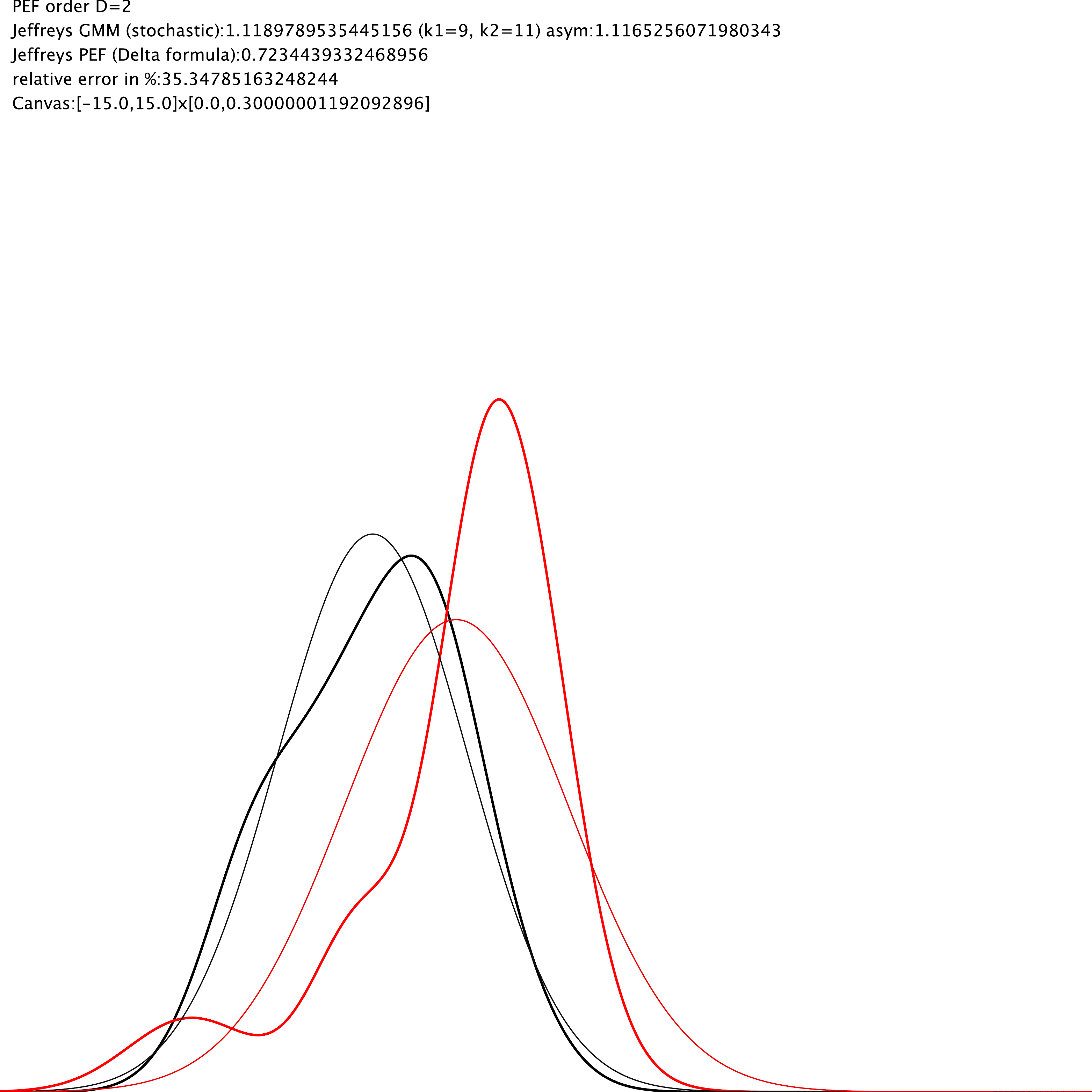}\\
$D=8$ &\includegraphics[width=\ttt\columnwidth]{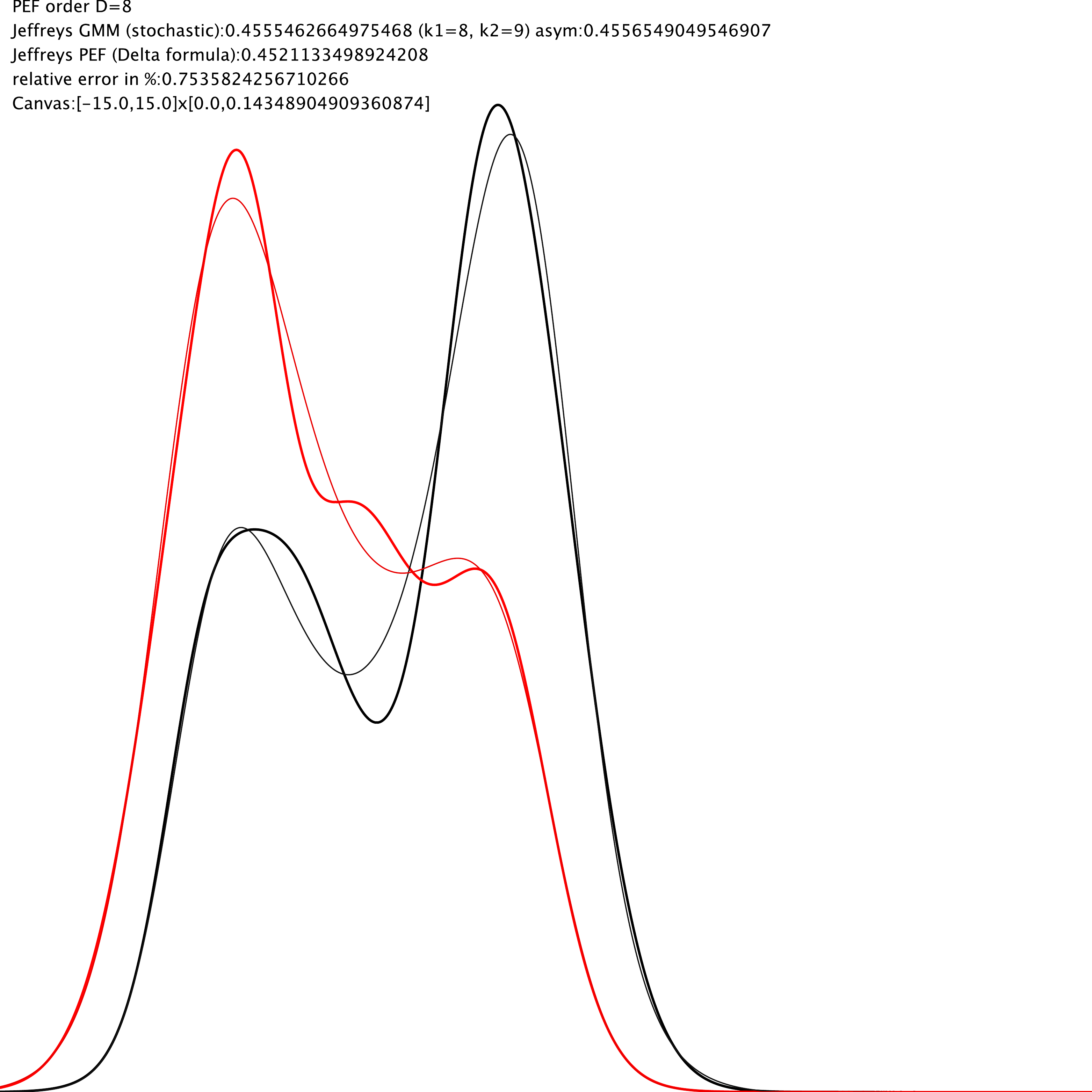}&
\includegraphics[width=\ttt\columnwidth]{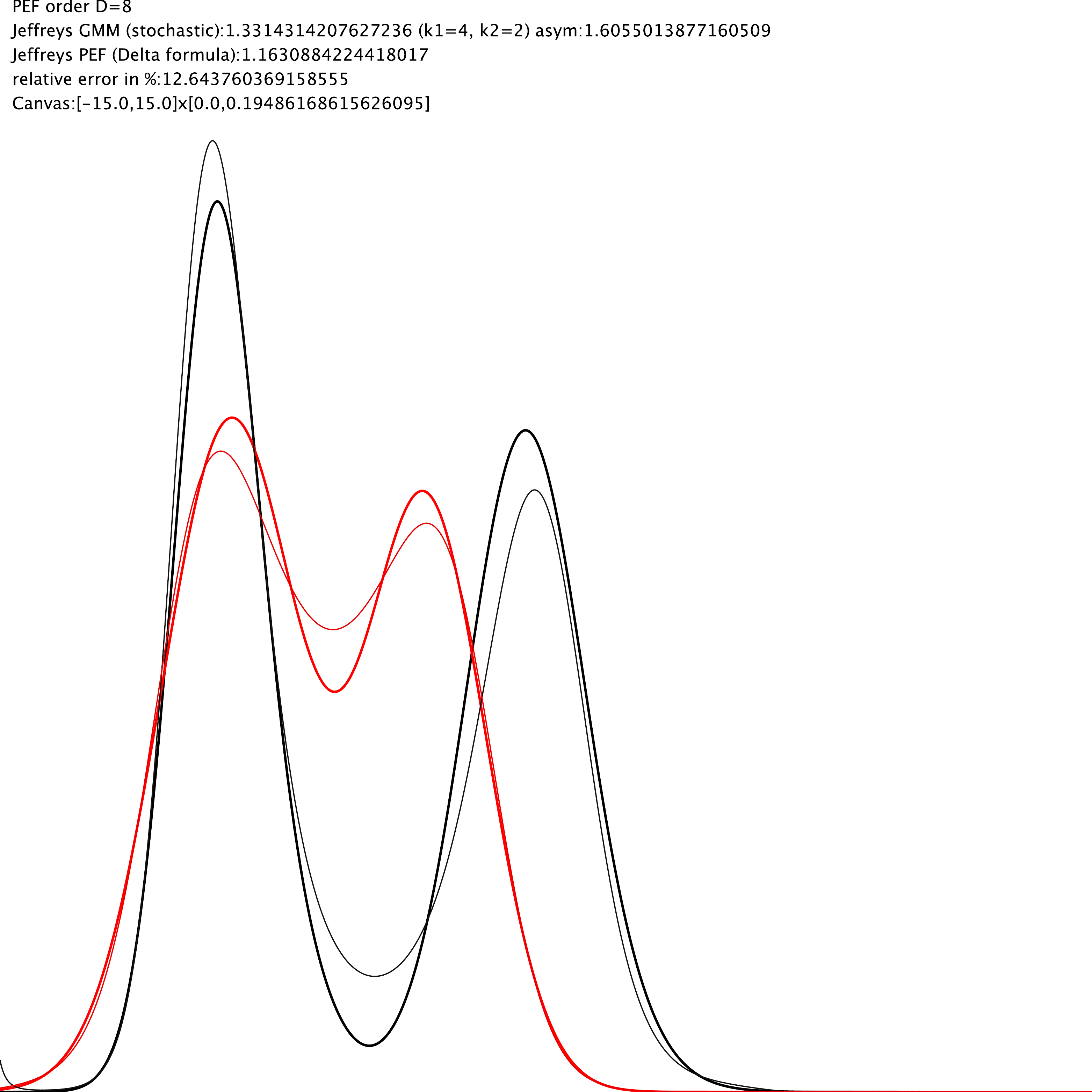}\\
$D=16$ &\includegraphics[width=\ttt\columnwidth]{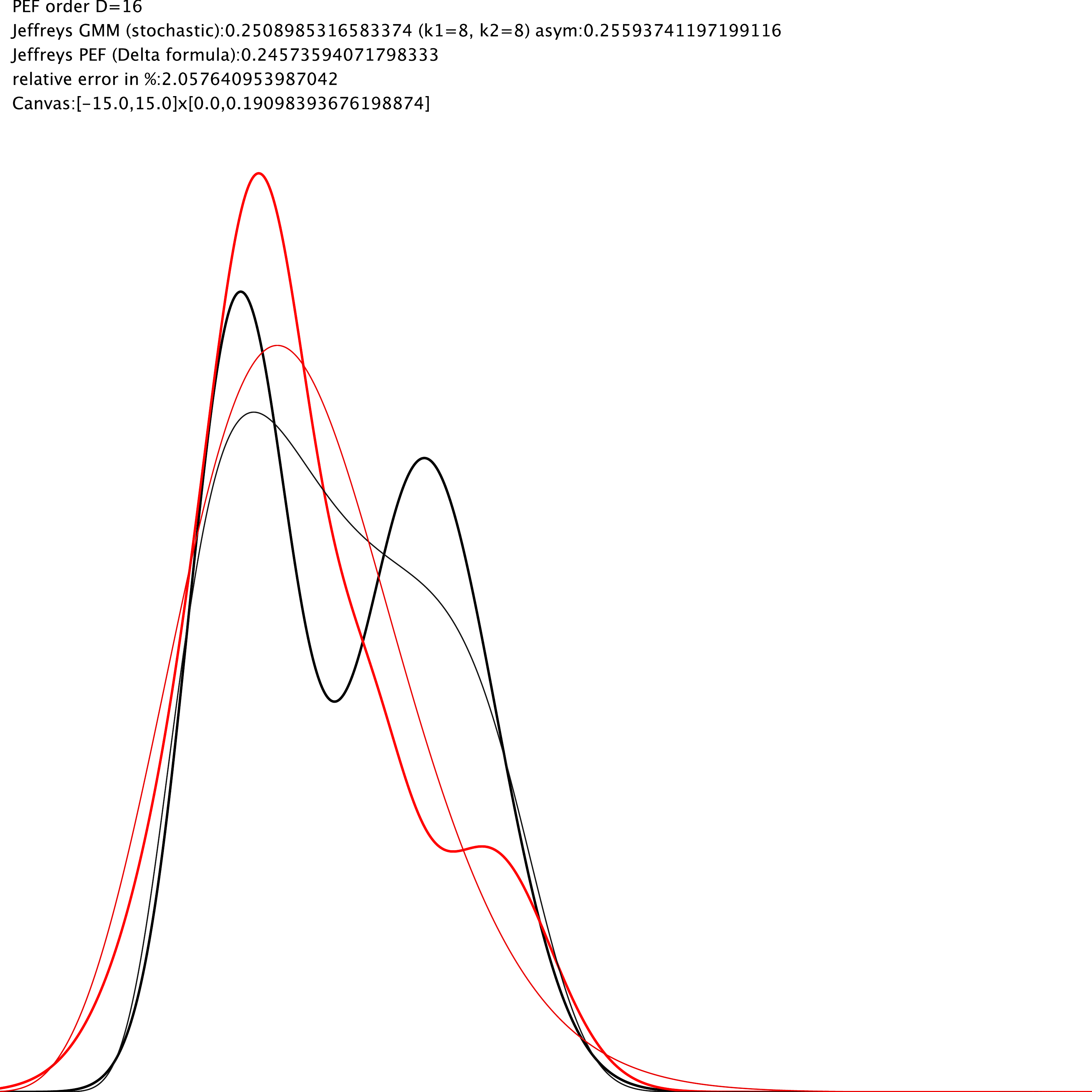}&
\includegraphics[width=\ttt\columnwidth]{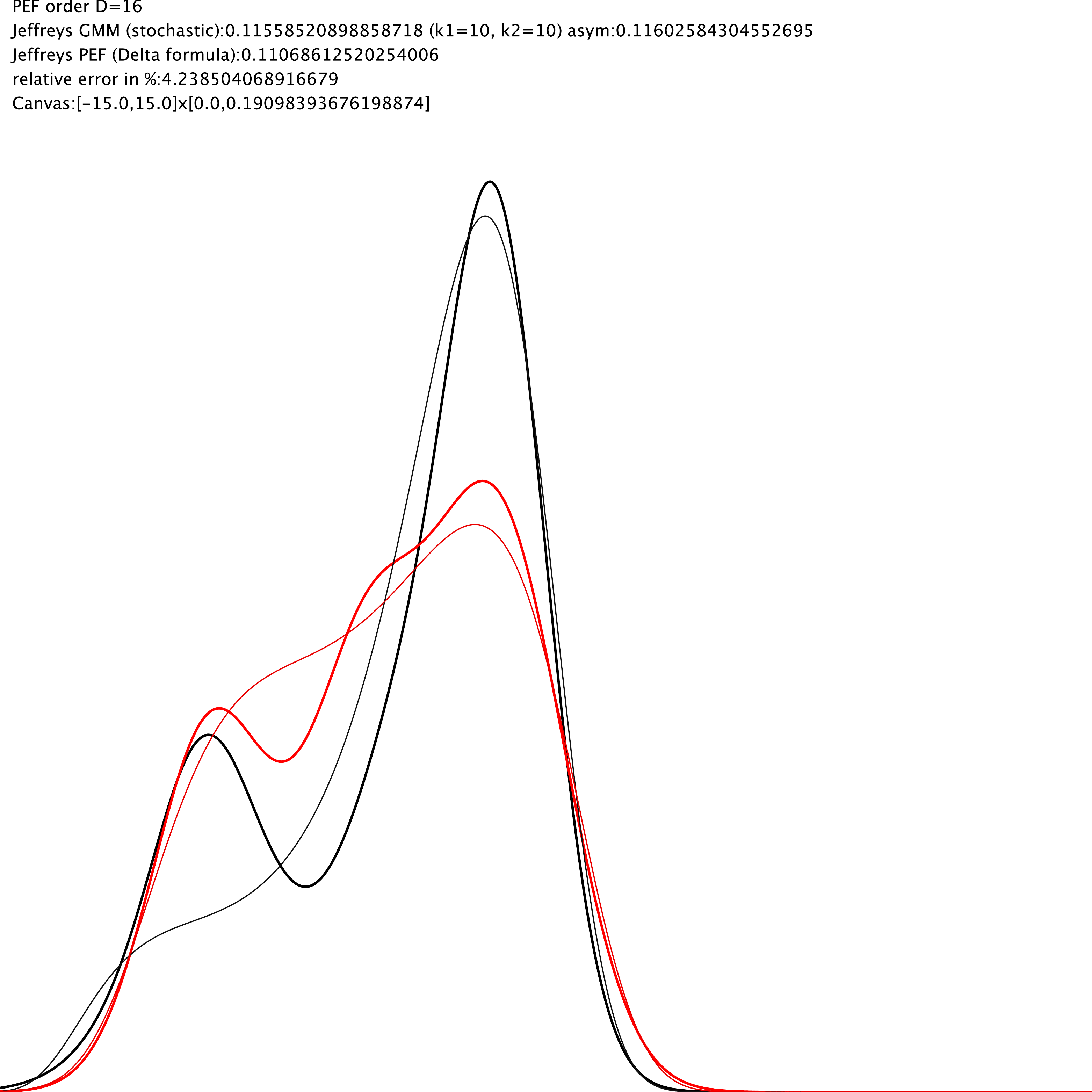}\\
\end{tabular}

\caption{Experiments of approximating the Jeffreys divergence between two mixtures by considering pairs of PEDs. Notice that only the PEDs estimated using the score matching in the natural parameter space are displayed.}%
\label{fig:exp}%
\end{figure}

Figure~\ref{fig:PEFD} illustrates the use of the order-$2$ Hyv\"arinen divergence $D_{H,2}$ to perform model selection for choosing the order of a PED.

\def\ttt{0.35\textwidth}
\begin{figure}
\centering
\begin{tabular}{cc}
\includegraphics[width=\ttt]{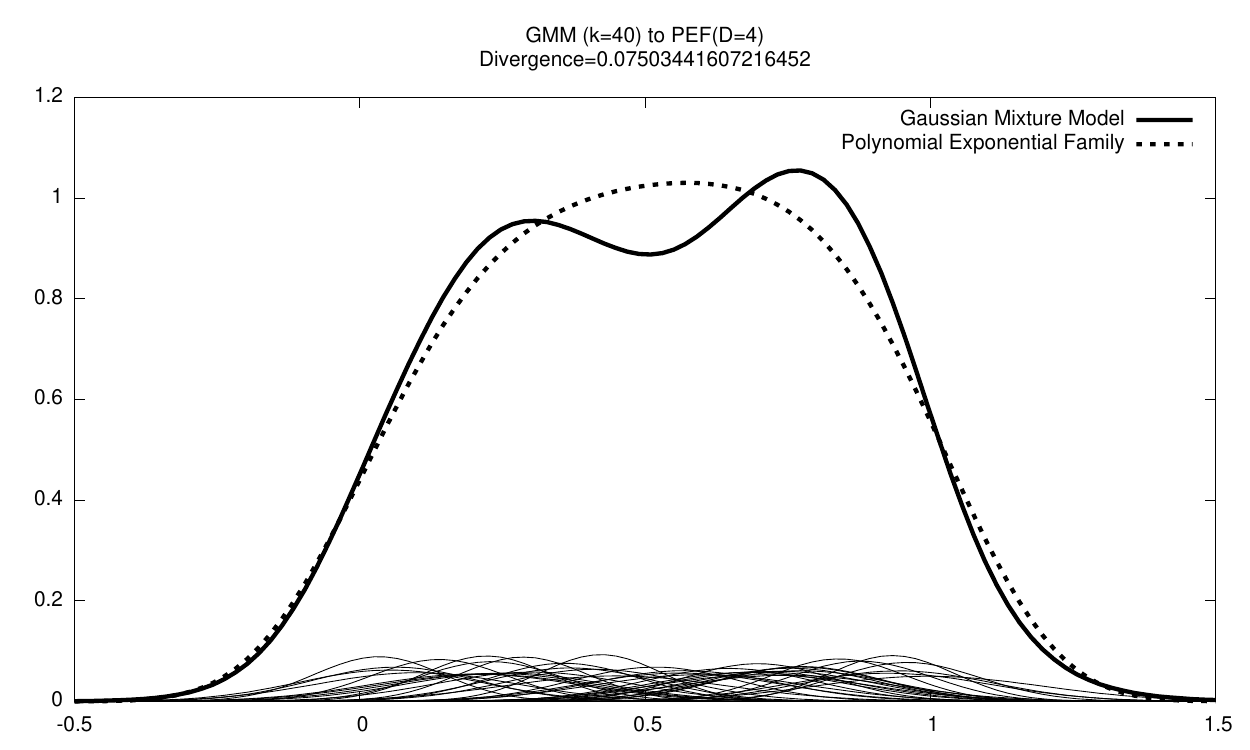} &
\includegraphics[width=\ttt]{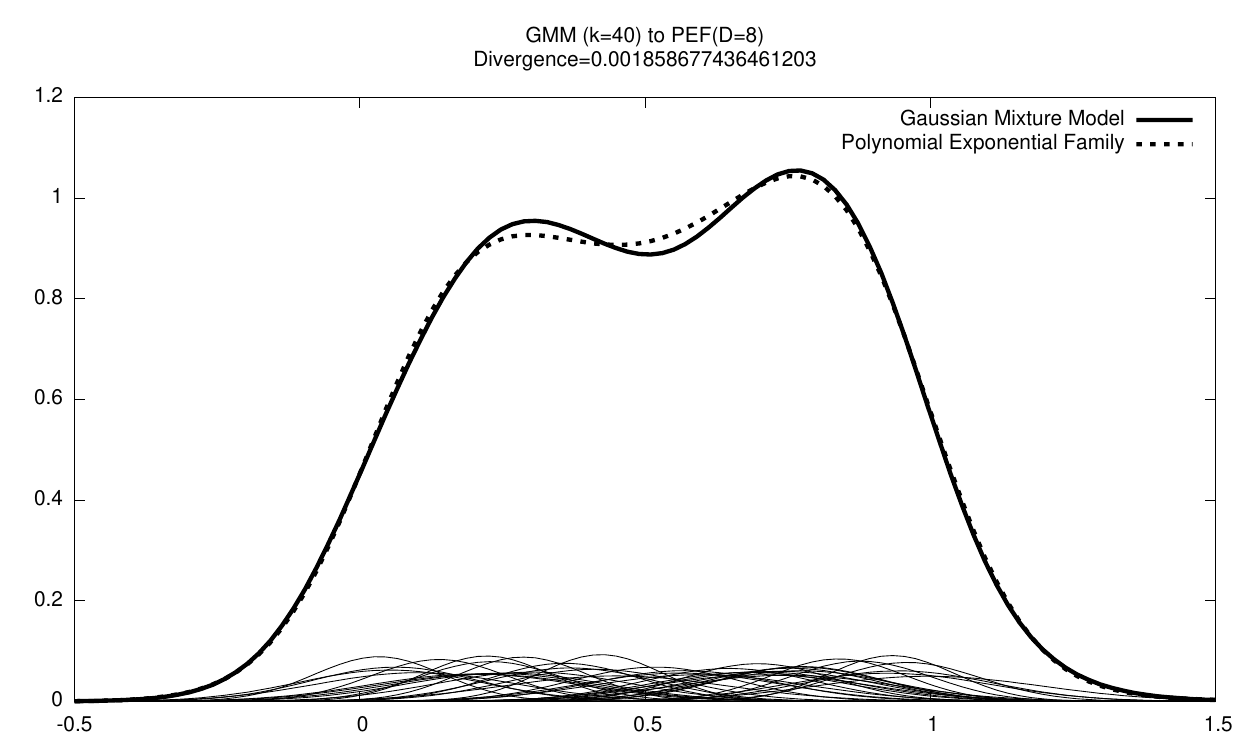} \\
\fbox{\includegraphics[width=\ttt]{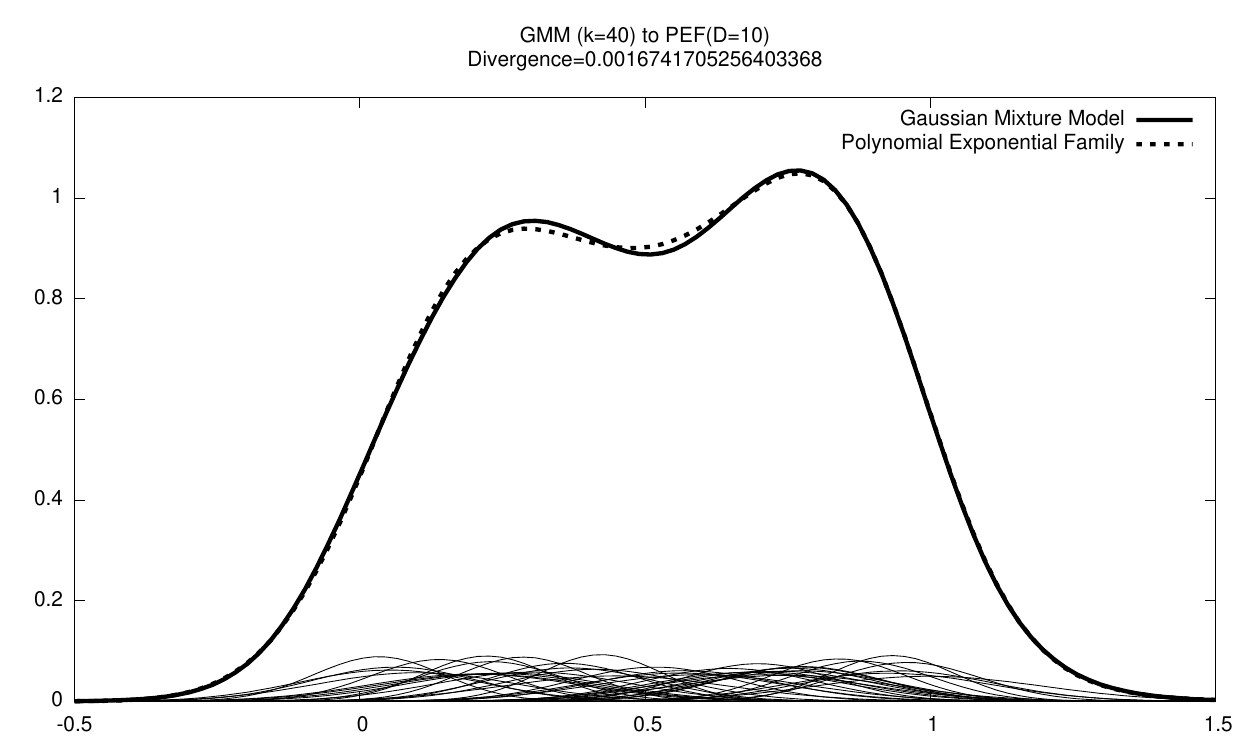}} &
\includegraphics[width=\ttt]{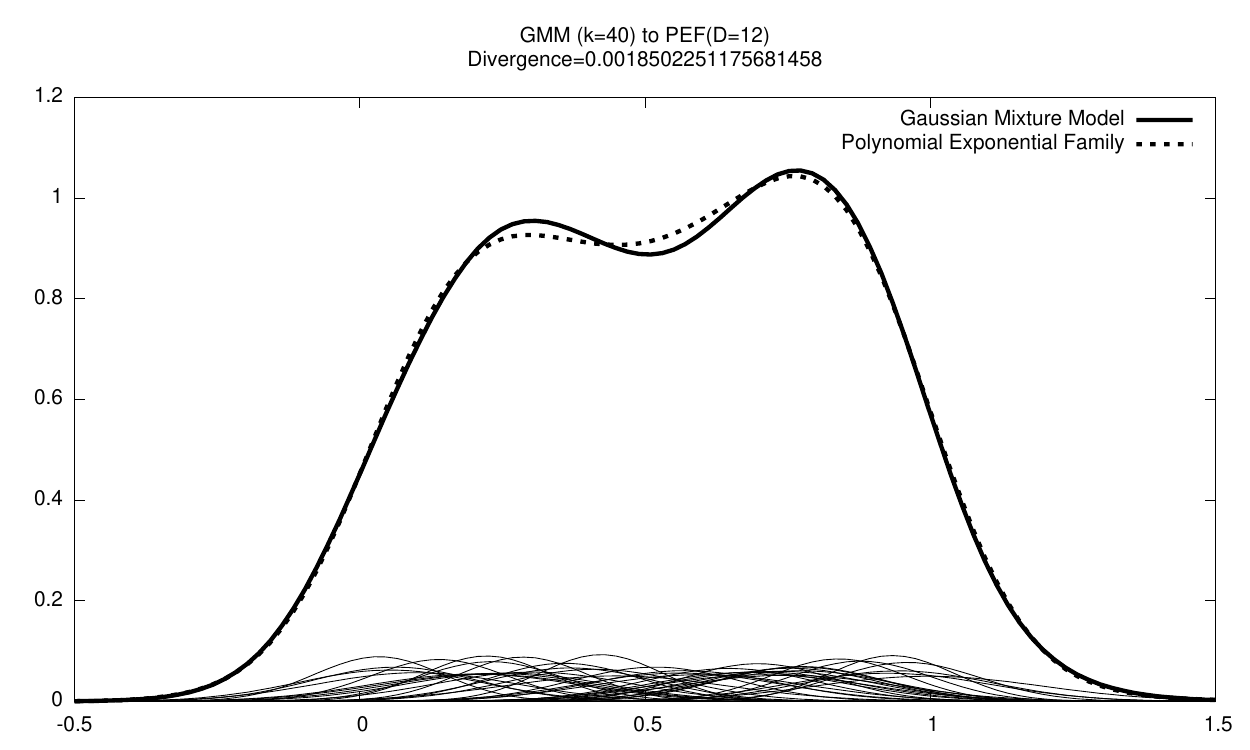} \\
\includegraphics[width=\ttt]{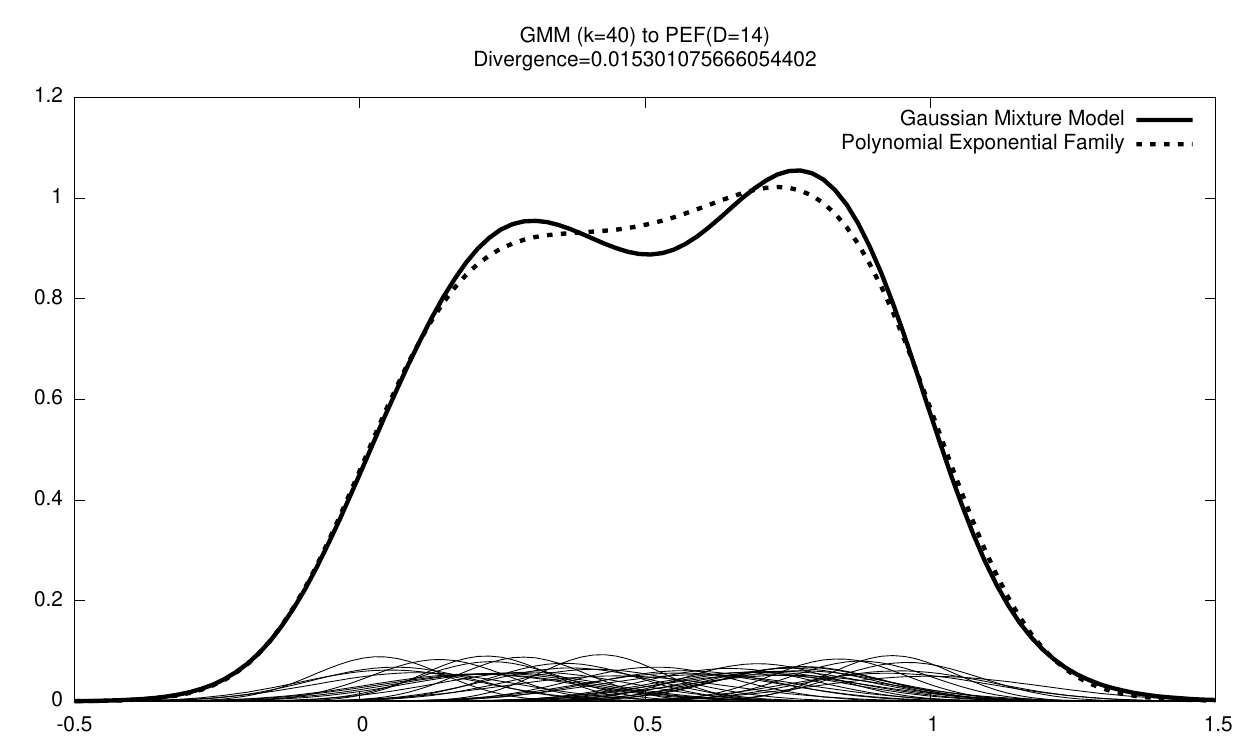} &
\includegraphics[width=\ttt]{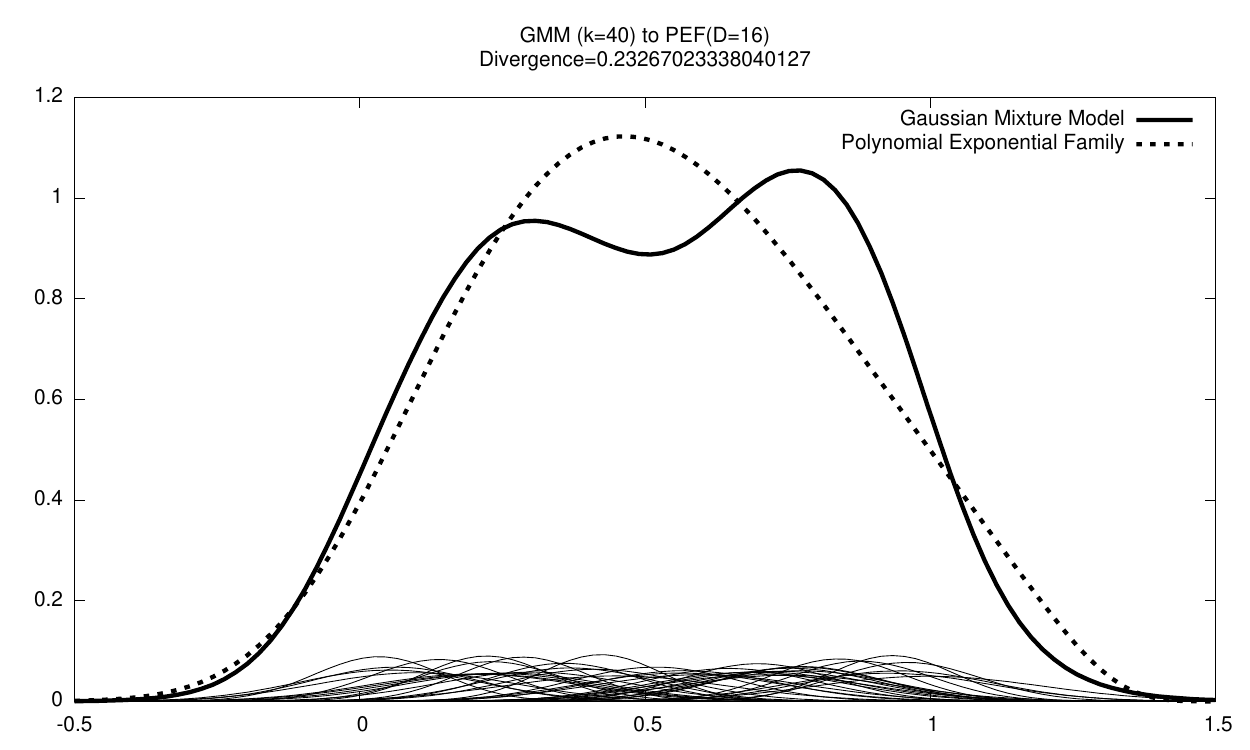} 
\end{tabular}

\caption{Selecting the PED order $D$ my evaluating the best divergence order-$2$ Hyv\"arinen divergence (for $D\in\{4,8,10, 12, 14, 16\}$) values. Here, the order $D=10$ (boxed) yields the lowest order-$2$ Hyv\"arinen divergence: The GMM is close to the PED.\label{fig:PEFD}}
\end{figure}

Finally, Figure~\ref{fig:explim} displays some limitations of the GMM to PED conversion when the GMMs have many modes.
In that case, running the conversion $\bareta_\MLE$ to get $\tilde{\theta}_T(\bareta_\MLE)$ and estimate the Jeffreys divergence by
$$
\tilde{\Delta}_J^\MLE[m_1,m_2]=(\tilde\theta_2^\MLE-\tilde\theta_1^\MLE)^\top 	(\bar\eta_2^\MLE-\bar\eta_1^\MLE),
$$ 
improves a lot the results but requires more computation.

\def\ttt{0.45}
\begin{figure}%
\begin{tabular}{cc}
\includegraphics[width=\ttt\columnwidth]{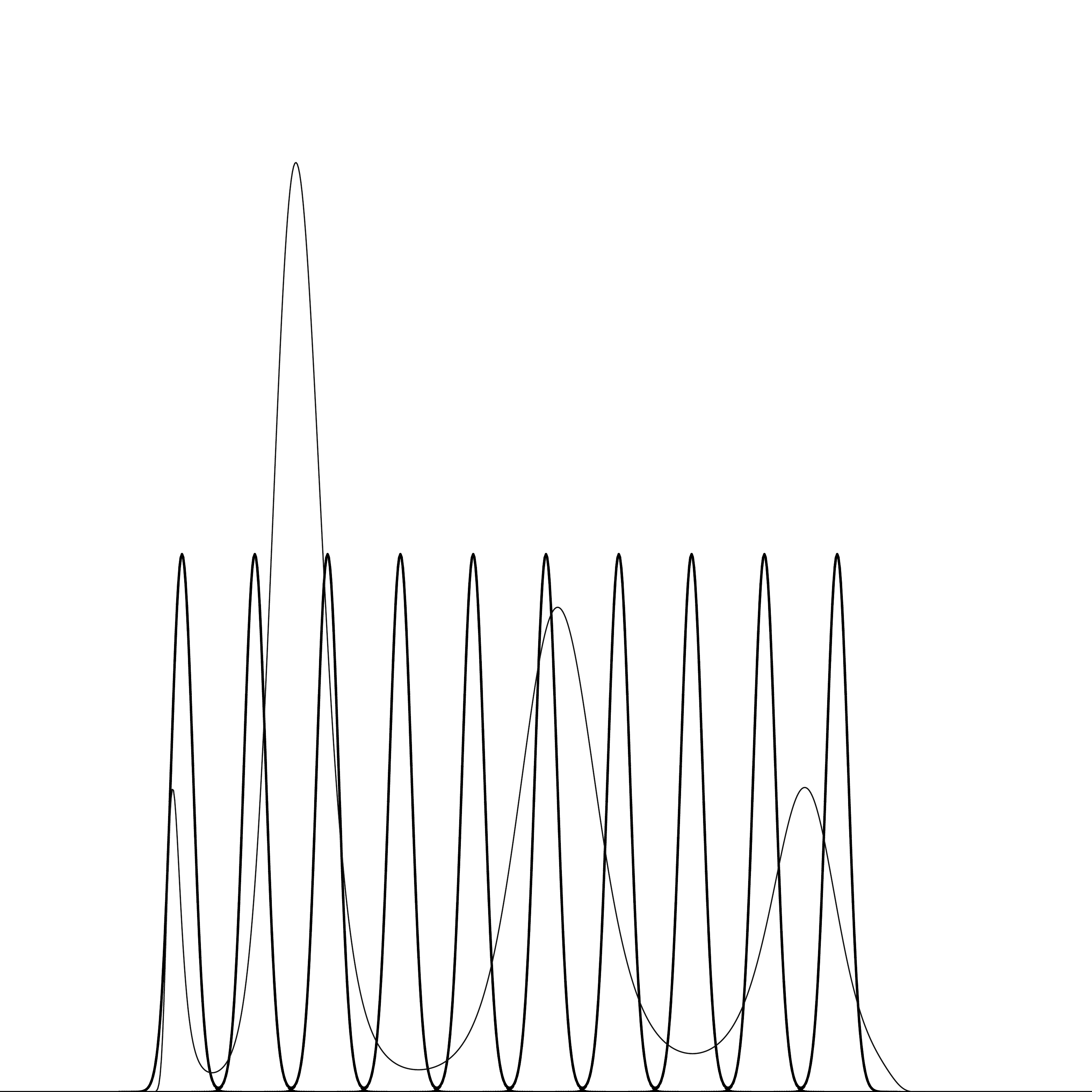}
  &\includegraphics[width=\ttt\columnwidth]{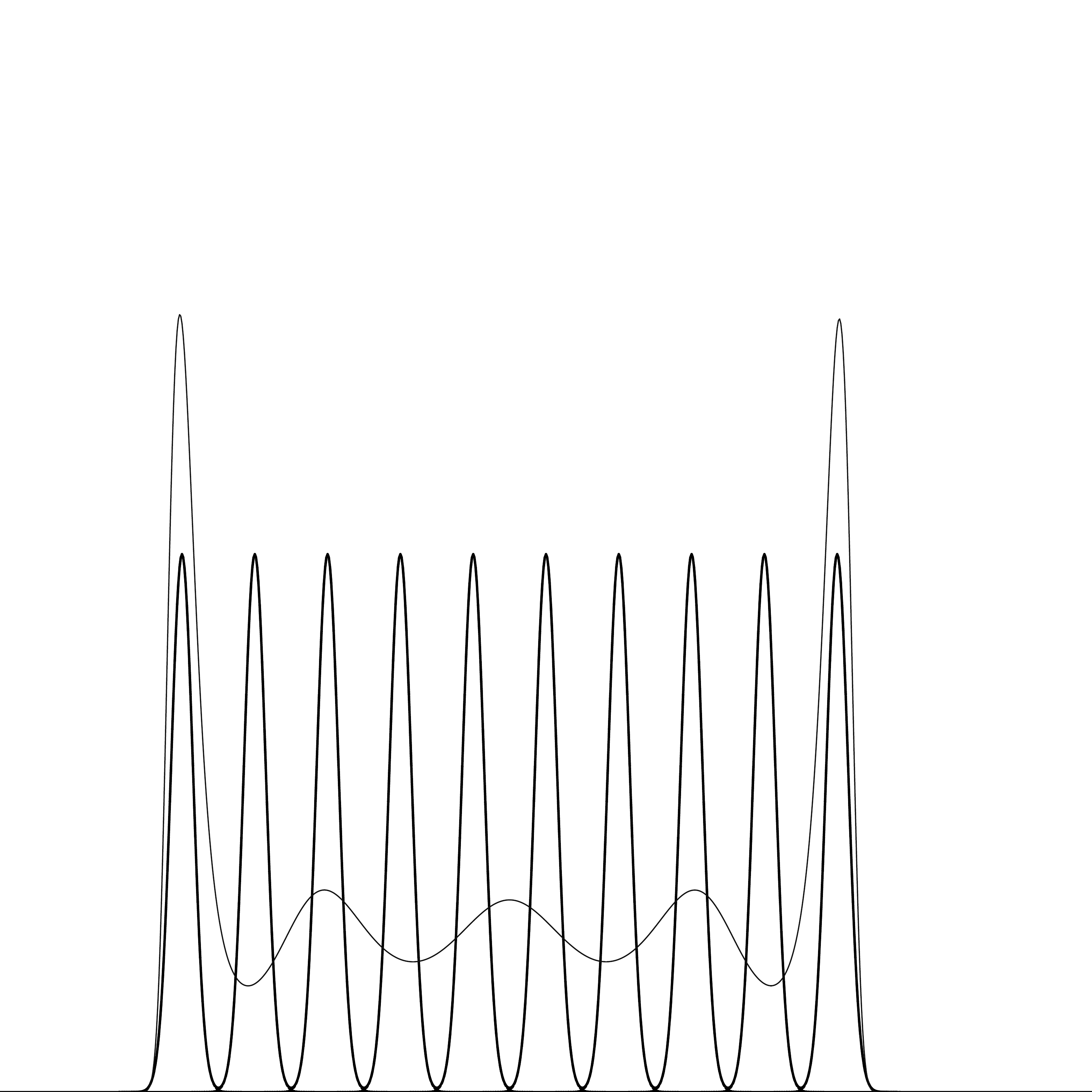}\\
	$D=32$ & $D=30$
\end{tabular}
\caption{Some limitation examples of the conversion of GMMs (black) to PEDs (grey) using the integral-based Score Matching estimator: Case of GMMs with many modes.}%
\label{fig:explim}%
\end{figure}

\def\ttt{0.2\textwidth}
\begin{figure}
\centering
\begin{tabular}{ccc}
\includegraphics[width=\ttt]{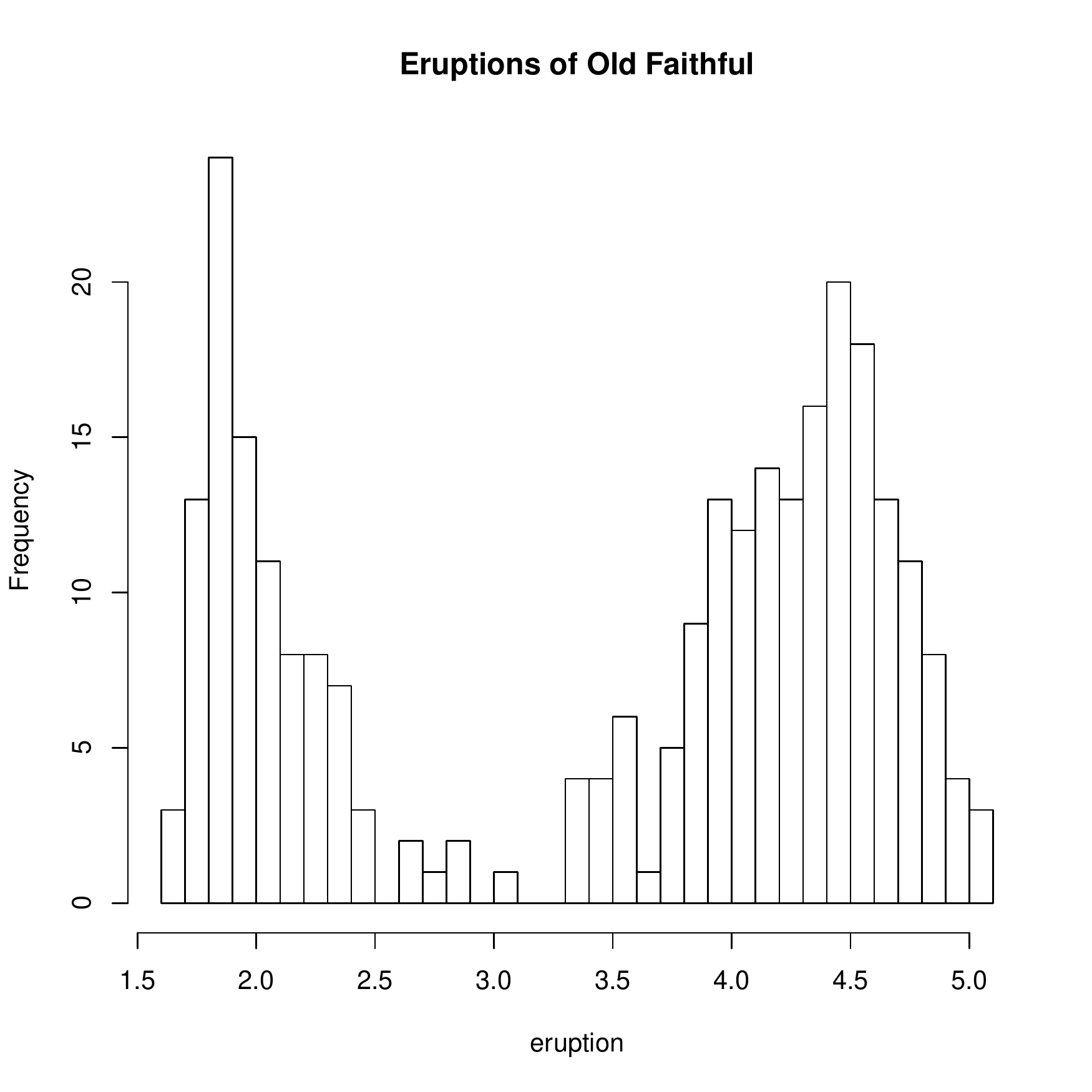} &
\includegraphics[width=\ttt]{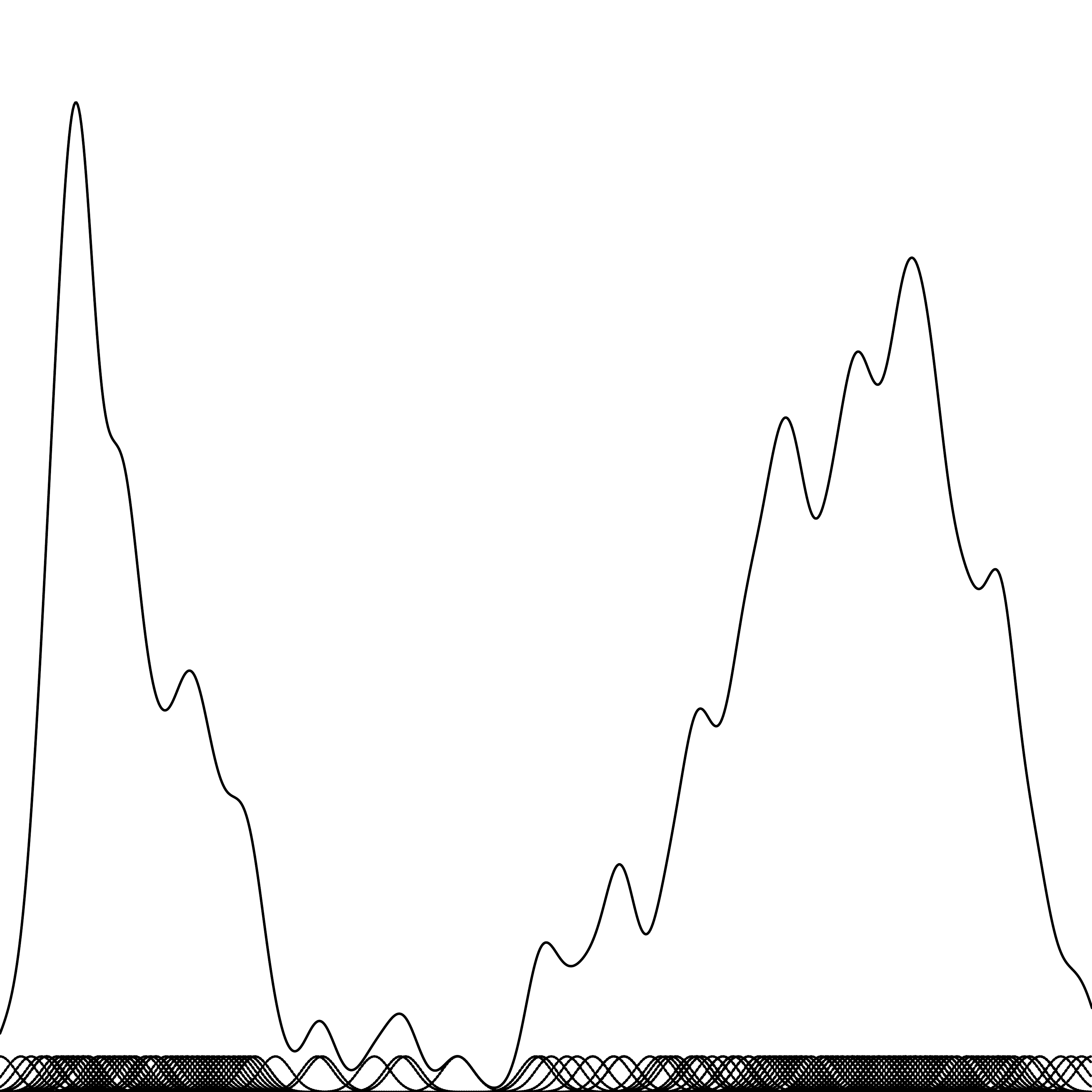} &
\includegraphics[width=\ttt]{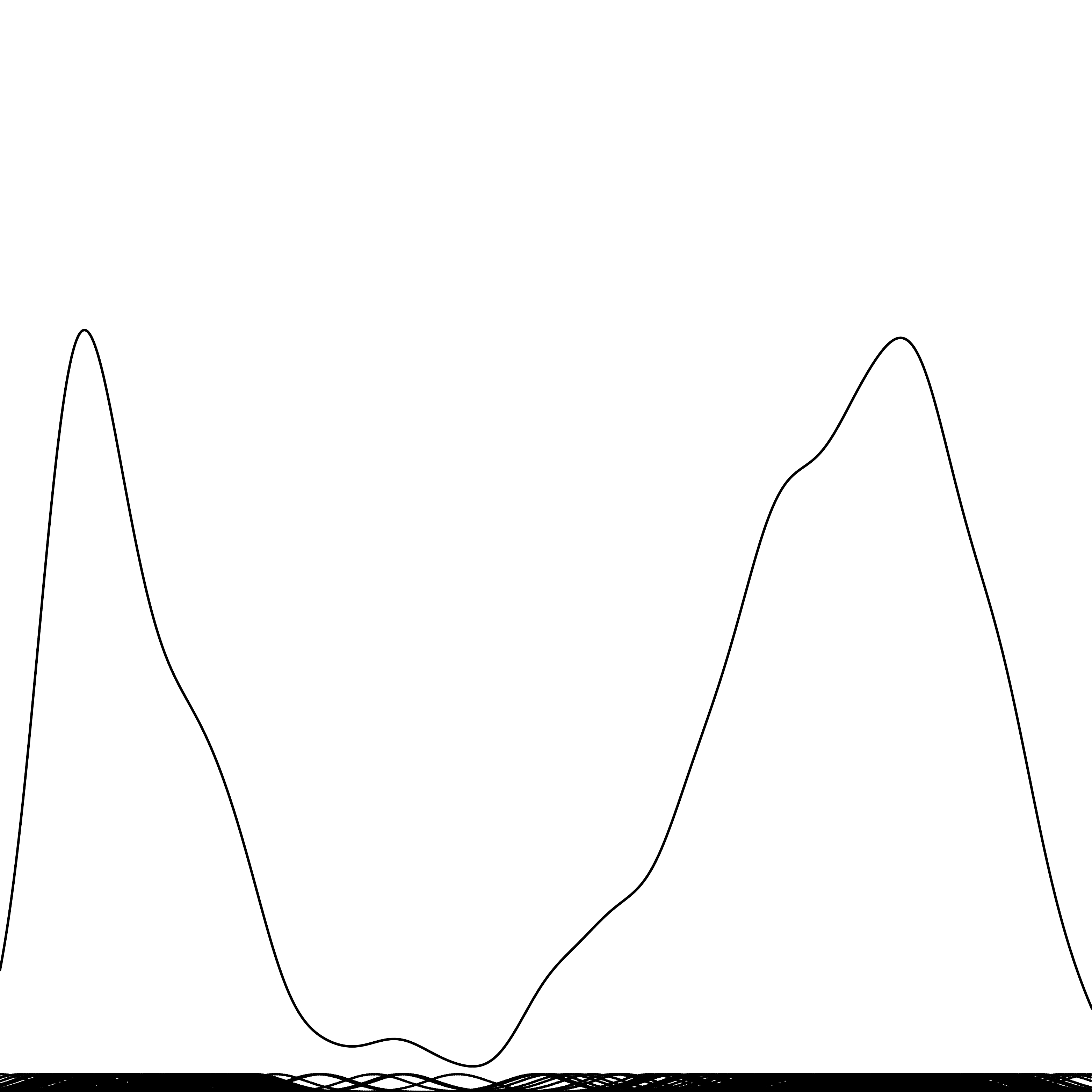} \\
Histogram (\#bins=25) & KDE with $\sigma=0.05$ & KDE with $\sigma=0.1$
\end{tabular}

\caption{Modeling the  Old Faithful geyser by a KDE (GMM with $k=272$ components, uniform weights $w_i=\frac{1}{272}$):
Histogram (\#bins=25), and KDE with $\sigma=0.05$ (middle) 
and KDE with $\sigma=0.1$ with less spurious bumps (right) \label{fig:gmmkde}}
\end{figure}

Next, we consider learning a PED by converting a GMM derived itself from a Kernel Density Estimator (KDE)~\cite{simplifykde-2013}.
We use the duration of the eruption for the Old Faithful geyser in Yellowstone National Park (Wyoming, USA):
The dataset consists of $272$ observations (\url{https://www.stat.cmu.edu/~larry/all-of-statistics/=data/faithful.dat}) 
and is included in the R language package 'stats'. The following R snippet (\url{https://www.r-project.org/}) converts the data into an histogram:
\begin{lstlisting}[backgroundcolor = \color{lightgray}]
require(stats);  
eruption <- faithful$eruptions
pdf("FaithfulHistogram.pdf")
hist(eruption,main="Eruptions of Old Faithful",breaks=25) 
dev.off()
\end{lstlisting}

Figure~\ref{fig:gmmkde} displays the GMMs obtained from the KDEs of the Old Faithful geyser dataset when choosing for each component 
$\sigma=0.05$ (left) and $\sigma=0.1$. Observe that the data is bimodal once the spurious modes (i.e., small bumps) are removed, as studied in~\cite{barron1991approximation}.
Barron and Sheu~\cite{barron1991approximation,correctionquartic-1991} modeled that dataset using a bimodal PED of order $D=4$, i.e., a quartic distribution. 
We model it with a PED of order $D=10$ using the integral-based score matching method.
Figure~\ref{fig:gmmpedgeyser} displays the unnormalized bimodal density $q_1$ (i.e., $\tilde{p}_1$) that we obtained using the integral-based score matching method (with $\calX=(0,1)$).

\def\ttt{0.5\textwidth}
\begin{figure}
\centering
\includegraphics[width=\ttt]{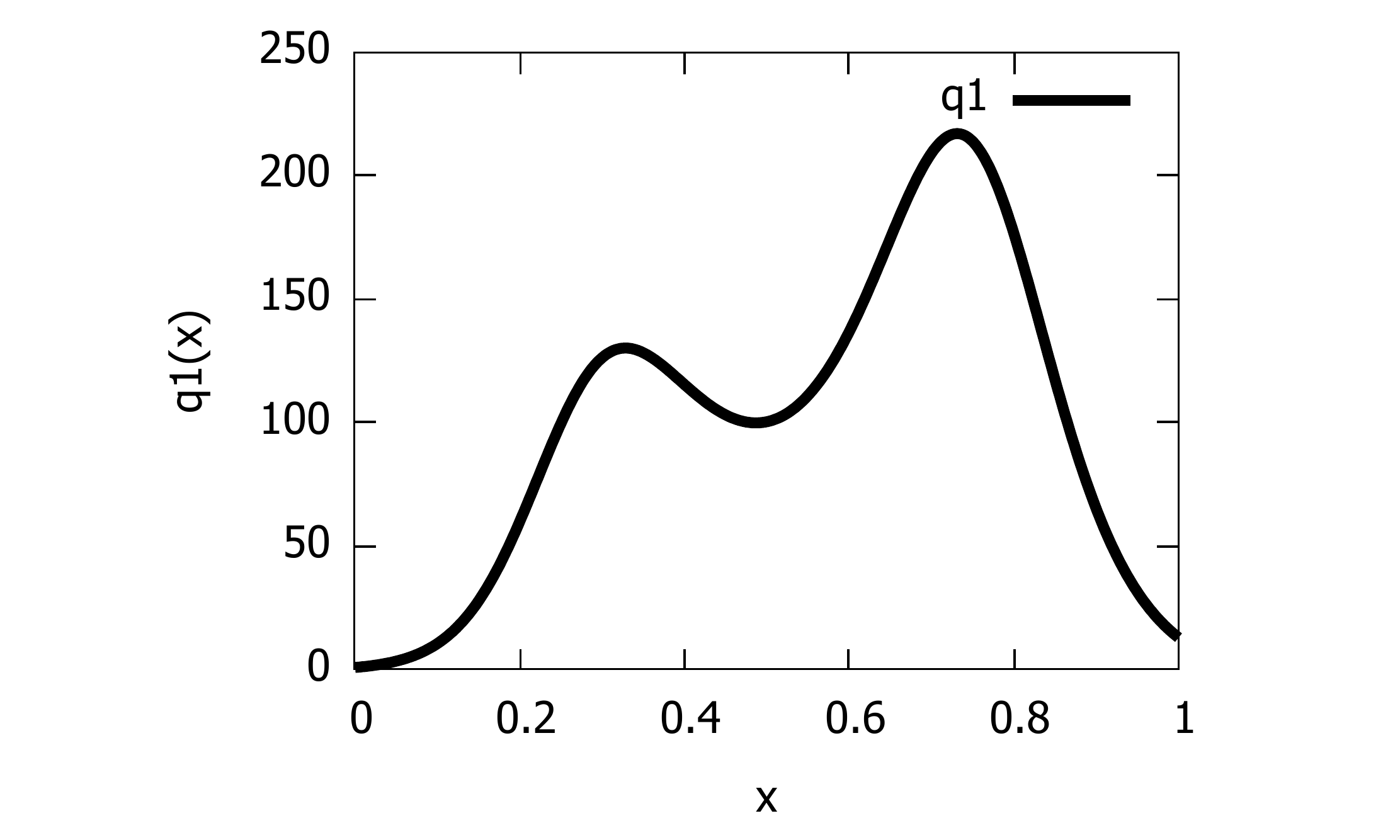}  

\caption{Modeling the  Old Faithful geyser by an exponential-polynomial distribution of order $D=10$. \label{fig:gmmpedgeyser}}
\end{figure}

The code in {\sc Maxima} to plot Figure~\ref{fig:gmmpedgeyser} is:
\begin{lstlisting}[backgroundcolor = \color{lightgray}]
q1(x):=exp((28.134510211471934*x**1)+(-71.464164508041*x**2)+(423.6511203845342*x**3)+(-1630.7270876169205*x**4)+(487.7447113037109*x**5)+(10340.70169099172*x**6)+(-25987.173209599085*x**7)+(27833.70635700226*x**8)+(-14324.646585464478*x**9)+(2902.643246746063*x**10));
plot2d([q1(x)],[x,-0,1], [xlabel,"x"], [ylabel,"q1(x)"], [legend, "q1"],[style, [lines,5,5]]);
\end{lstlisting}

\section{Conclusion and perspectives}\label{sec:concl}

Many applications require to compute the Jeffreys divergence (a symmetrized Kullback-Leibler divergence) between Gaussian mixture models.
See~\cite{xiao2010optimal,bilik2012minimum,ijcai16,vitoratou2017thermodynamic} for a few use cases.
Since the Jeffreys divergence between GMMs is provably not available in closed-form~\cite{KLnotanalytic-2004}, one often ends up implementing a costly Monte Carlo stochastic approximation of the JD.
In this paper, we first noticed the simple expression of the JD between densities $p_\theta$ and $p_{\theta'}$ of an exponential family   using their dual natural/moment parameterizations~\cite{BN-2014} $p_\theta=p^\eta$ and $p_{\theta'}=p^{\eta'}$:
$$
D_J[p_\theta,p_{\theta'}]=(\theta'-\theta)^\top (\eta'-\eta), 
$$
where $\eta=\nabla F(\theta)$ and $\eta'=\nabla F(\theta')$ for  the cumulant function $F(\theta)$ of the EF.
We then proposed a simple fast heuristic to approximate the JD between GMMs: 
First, convert a mixture $m$  to a pair $(p_{\bar\theta^\SME},p^{\bar\eta^\MLE})$ of dually parameterized polynomial exponential densities using extensions of the Maximum Likelihood and Score Matching Estimators (Theorem~\ref{thm:MLE} and Theorem~\ref{thm:SMEmixPED}), and then approximate the JD deterministically by
$$
D_J[m_1,m_2]\simeq \tilde{D}_J[m_1,m_2]= (\tilde\theta_2^\MLE-\tilde\theta_1^\MLE)^\top 	(\bar\eta_2^\MLE-\bar\eta_1^\MLE).
$$ 
The order of the polynomial exponential family may be prescribed or selected using the order-$2$ Hyv\"arinen divergence which evaluates in closed form  the dissimilarity between a  GMM and a PED density (Theorem~\ref{thm:hdpef}). 
We demonstrated experimentally that the Jeffreys divergence  between GMMs can be reasonably well approximated by $\tilde{D}_J$ for mixtures with small number of modes, with an overall speed-up of several order of magnitudes compared to the vanilla Monte Carlo sampling method.
We also propose another deterministic heuristic to estimate $D_J$ as 
$$
\tilde{D}_J^\MLE[m_1:m_2]=(\tilde\theta_2^\MLE-\tilde\theta_1^\MLE)^\top 	(\bar\eta_2^\MLE-\bar\eta_1^\MLE),
$$
where $\tilde\theta^\MLE\approx \nabla F(\bar\eta^\MLE)$ is numerically calculated using an iterative conversion procedure based on maximum entropy~\cite{MaxEnt-1992} (Section~\ref{sec:eta2theta}).
Our technique extends to other univariate mixtures of exponential families~\cite{garcia2010simplification} (e.g., mixtures of Rayleigh distributions, mixtures of Gamma distributions or mixtures of Beta distributions).
One limitation of our method is that the PED modeling of a GMM may not guarantee to obtain the same number of modes as the GMM even when we increase the order $D$ of the PEDs.
This case is illustrated in Figure~\ref{fig:mode} (right). 

\def\ttt{0.3\textwidth}
\begin{figure}%
\begin{tabular}{cc}
\includegraphics[width=\ttt]{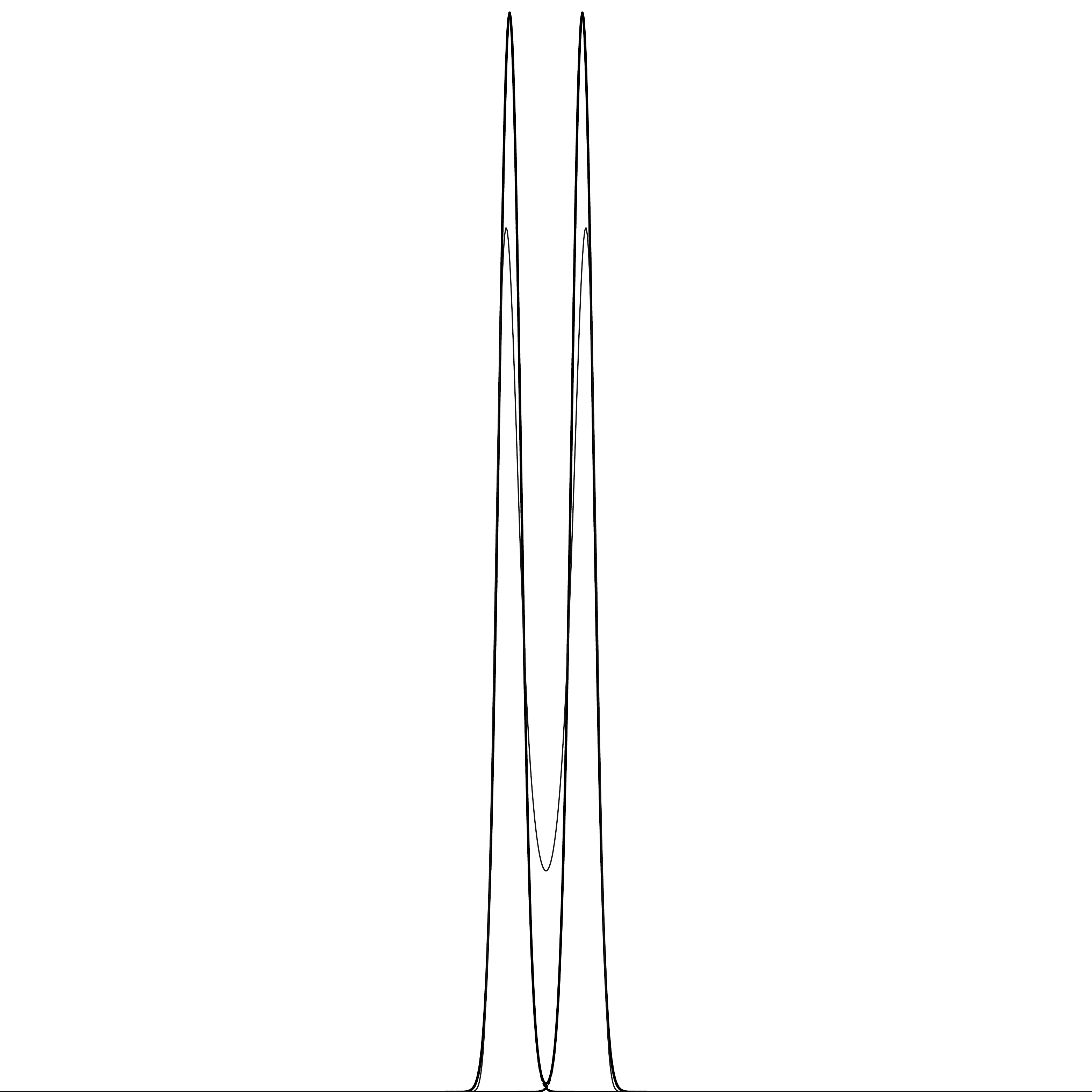}
  &\includegraphics[width=\ttt]{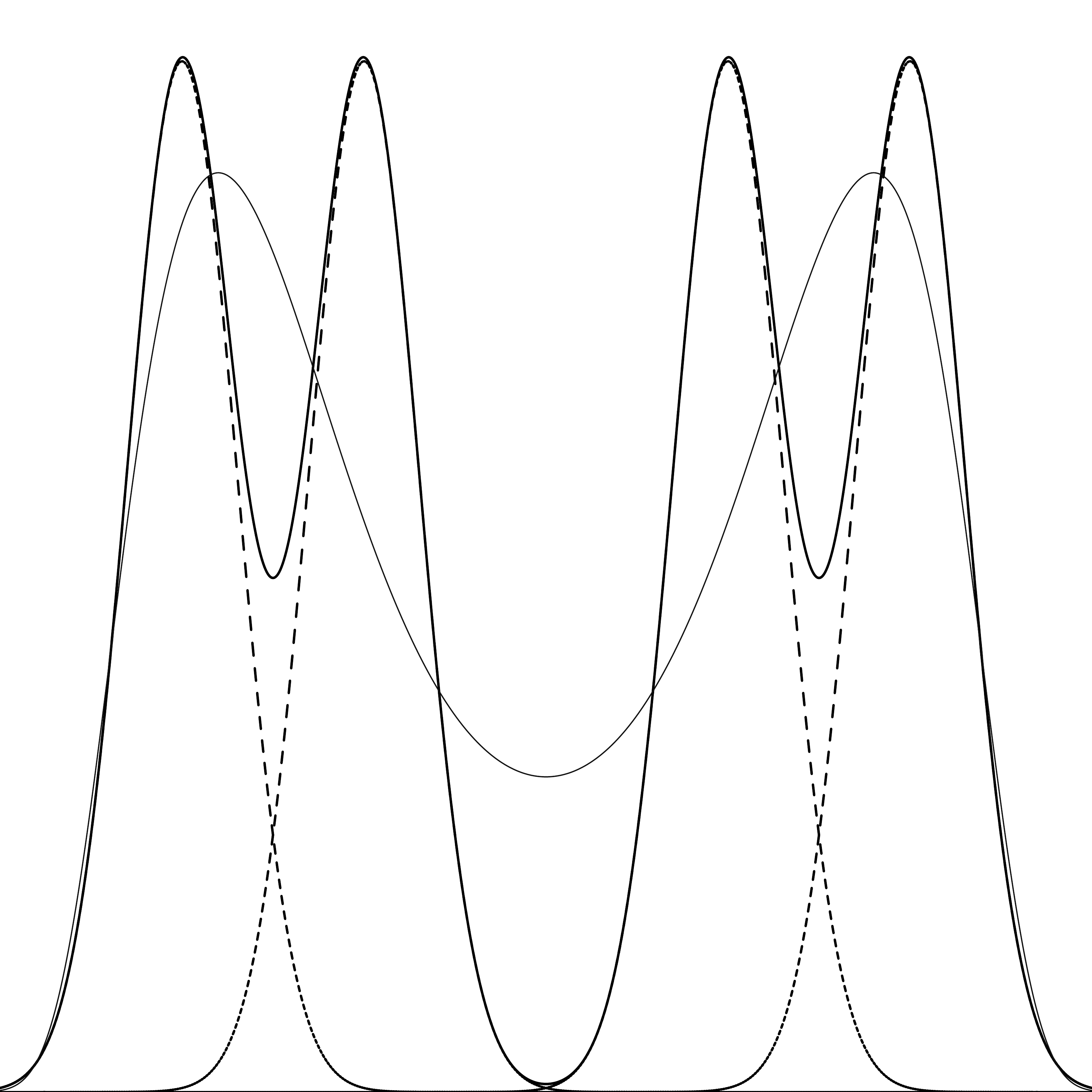}\\
	 Same number of modes ($D=4$) & Different number of modes ($D=30$)
\end{tabular}
\caption{GMM modes versus PED modes:
(left) same number and locations of modes for the GMM and the PED, and (right) $4$ modes for the GMM but only $2$ modes for the PED. }%
\label{fig:mode}%
\end{figure}

Although PEDs are well-suited to calculate Jeffreys divergence compared to GMMs, we point out that GMMs are better suited for sampling while PEDs require Monte Carlo methods (e.g., adaptive rejection sampling or MCMC methods~\cite{rohde2014mcmc}).
Also, we can estimate the Kullback-Leibler Divergence between two PEDs using rejection sampling (or other McMC methods~\cite{rohde2014mcmc}) or by using the $\gamma$-divergence~\cite{eguchi2011projective} with $\gamma$ close to zero~\cite{NN-2016} (e.g., $\gamma=0.001$).

This work opens up several perspectives for future research: For example, we may consider considered bivariate PEDs for modeling bivariate GMMs~\cite{PEF-holonomic2016}, or we may consider truncating the GMMs in order to avoid tail phenomena when converting GMMs to PEDs~\cite{truncmoment-2014,del1994singly}.

The web page of the project is\\
\centerline{\mbox{\url{https://franknielsen.github.io/JeffreysDivergenceGMMPEF/index.html}}}

\bibliographystyle{plain}

\appendix

\section{Monomial exponential families}\label{sec:MEF}
Consider the following polynomial exponential density defined on the full real line support $\calX=\bbR$ which consists of a single monomial sufficient statistic $t(x)=x^D$: 
$$
p_{D,\theta}(x):=\exp\left(\theta x^D-F_D(\theta)\right),
$$
for an even integer $D\geq 2$.
The set of such densities form a Monomial Exponential Family~\cite{nielsen2017maxent} $\calM_D=\{p_{D,\theta}(x)\ :\ \theta<0\}$ (MEF) with sufficient statistic $t(x)=x^D$.
$\calM_D$ is  a univariate order-$1$ exponential family. 
MEFs are special PEFs (with $\theta_1=\ldots=\theta_{D-1}=0$ and $\theta_D=\theta$) which yield tractable information-theoretic quantities like the KLD or the differential entropy.
Indeed, the cumulant function $F_{D}(\theta)$ is available in closed-form expression~\cite{nielsen2017maxent}:
\begin{eqnarray}
F_{D}(\theta) &=&\log \int_{\bbR} \exp(\theta x^D) \dx ,\\
&=& \log\left(\frac{2\Gamma(1/D)}{D}\right)-\frac{1}{D}\log(-\theta),
\end{eqnarray}
for $\theta<0$, where $\Gamma(\cdot)$ denotes the gamma function.
The natural parameter space is $\Theta=(-\infty,0)$.
The moment parameter is $\eta=\nabla F_D(\theta)=-\frac{1}{D\theta}>0$, and the moment space is $H=(0,\infty)$.
We have $\theta=\nabla F^*_D(\eta_D)=-\frac{1}{D\eta}<0$ and the convex conjugate  is:
$$
F_D^*(\eta)=-\log\left(\frac{2\Gamma(1/D)}{D}\right)-\frac{1}{D}(1+\log(D\eta)=-h[p_{D,\theta}].
$$
We check that the Fenchel-Young equality holds:
$$
F_{D}(\theta)+F_D^*(\eta)-\theta\eta=0.
$$  
The differential entropy of a MEF~\cite{entropyEF-2010} is $h[p_{D,\theta}(x)]=-F^*(\eta)$, and the Kullback-Leibler divergence is
\begin{eqnarray}
D_\KL[p_{D,\theta_1}:p_{D,\theta_2}] &=& B_F(\theta_2:\theta_1),\\
&=& -\frac{1}{D}\log \frac{-\theta_2}{-\theta_1}-(\theta_2-\theta_1) \left(-\frac{1}{D\theta_1}\right),\\
&=&
\frac{1}{D}\left(\frac{\theta_2}{\theta_1}-\log\frac{\theta_2}{\theta_1}-1\right).
\end{eqnarray}

It follows that the KLD is a scaled Itakura-Saito divergence~\cite{nielsen2009sided} $D_\IS$ with
$$
D_\IS[p:q]:=\frac{p}{q}-\log\left(\frac{p}{q}\right)-1.
$$
That is, we have:
$$
D_\KL[p_{D,\theta_1}:p_{D,\theta_2}]=\frac{1}{D}\,D_\IS[\theta_2:\theta_1].
$$

\begin{Example}
Let us report the MEFs for $D=2$, the zero-mean centered normal distributions~\cite{nielsen2009statistical} 
$\calM_2=\{N(0,\sigma^2)\ :\ \sigma^2>0\}$.
We have $\theta=-\frac{1}{\sigma^2}<0$, $t(x)=x^2$, $F_2(\theta)=\frac{1}{2}\log\left(\frac{\pi}{-\theta}\right)$  (since $\Gamma(1/2)=\sqrt{\pi}$),
$\eta=\sigma^2>0$ and $F_2^*(\eta)=-\frac{1}{2}-\frac{1}{2}\log(2\pi\eta)$.

The KLD between two zero-mean normal distributions $p_{0,\sigma_1^2}$ and $p_{0,\sigma_2^2}$ is
\begin{eqnarray}
D_\KL[p_{0,\sigma_1^2}:p_{0,\sigma_2^2}]  &=& \frac{1}{2} \left(
\frac{\theta_2}{\theta_1}-\log\left(\frac{\theta_2}{\theta_1}\right) -1
\right),\\
&=& \frac{1}{2}\left( \frac{\sigma_1^2}{\sigma_2^2}-\log \frac{\sigma_1^2}{\sigma_2^2} -1\right),\\
&=& D_\IS\left[-\frac{1}{2\sigma_2^2}:-\frac{1}{2\sigma_1^2}\right],\\
&=& D_\IS[\sigma_1^2:\sigma_2^2],
\end{eqnarray}
since the Itakura-Saito divergence is scale-free: $D_\IS[p:q]=D_\IS[1:\frac{q}{p}]=D_\IS[\frac{p}{q}:1]$.
Thus we have
$$
D_\KL[p_{0,\sigma_1^2}:p_{0,\sigma_2^2}]=D_\IS\left[\frac{\sigma_1^2}{\sigma_2^2}:1\right]=D_\IS[\sigma_1^2:\sigma_2^2].
$$

This matches the usual KLD between two normal distributions which can be interpreted as the sum of a squared Mahalanobis distance and half of the Itakura-Saito divergence (as noticed in~\cite{dhillon2007differential}):
\begin{eqnarray}
D_\KL[p_{\mu_1,\sigma_1^2}:p_{\mu_2,\sigma_2^2}] &=&  \frac{(\mu_2-\mu_1)^2}{2\sigma_2^2}+\frac{1}{2}\left( \frac{\sigma_1^2}{\sigma_2^2}-\log \frac{\sigma_1^2}{\sigma_2^2} -1\right),\\
&=& \frac{1}{2}\left(M_{\sigma_2^2}^2(\mu_1,\mu_2)+D_\IS\left[-\frac{1}{2\sigma_2^2}:-\frac{1}{2\sigma_1^2}\right]\right),
\end{eqnarray}
where 
$$
M_{\sigma_2^2}^2(\mu_1,\mu_2):=\frac{(\mu_2-\mu_1)^2}{\sigma_2^2}.
$$
\end{Example}

The following code in {\sc Maxima} implements and tests the various formula for the KLD between two densities of a MEF:

\begin{lstlisting}[backgroundcolor = \color{lightgray}]
/* Order of the MEF (even integer) */
D: 12;
F(theta):=log(2*gamma(1/D)/D)-(1/D)*log(-theta);
gradF(theta):=-1/(D*theta);
gradFdual(eta):=-1/(D*eta);
Fdual(eta):=-log (2*gamma(1/D)/D)-(1/D)*(1+log(D*eta));
p(x,theta):=exp(theta*(x**D)-F(theta));
/* Check Fenchel-Young inequality */
F(theta)+Fdual(gradF(theta))-theta*gradF(theta);
float(radcan(%)); /* should be zero */
/* Closed-form formula for Kullback-Leibler divergence */
LF(theta1,eta2):=F(theta1)+Fdual(eta2)-theta1*eta2;
BF(theta1,theta2):=F(theta1)-F(theta2)-(theta1-theta2)*gradF(theta2);
BdualF(eta1,eta2):=Fdual(eta1)-Fdual(eta2)-(eta1-eta2)*gradFdual(eta2);
ItakuraSaito(p,q):=(p/q)-log(p/q)-1;
closedformKL(theta1,theta2):=(1/D)*(theta2/theta1-log(theta2/theta1)-1);
closedformKLIS(theta1,theta2):=(1/D)*ItakuraSaito(theta2,theta1);
/* Differential entropy */
h(theta):=-Fdual(gradF(theta));
/* Test example */
theta:-2;
theta1:-2;
theta2:-3;
eta1:gradF(theta1);
eta2:gradF(theta2);
/* Kullback-Leibler divergence */
/* numerical integration by clamping the x-range */
print("numerical KLD:",quad_qag (p(x,theta1)*log(p(x,theta1)/p(x,theta2)), x, -100,100, 6, 'epsrel=5d-8)[1]);
/* Exact formula for the KLD */
print("Legendre-Fenchel divergence:",float(LF(theta2,eta1))); 
print("Bregman divergence:",float(BF(theta2,theta1)));
print("Dual Bregman divergence:",float(BdualF(eta1,eta2)));
print("Closed-form KLD:",float(closedformKL(theta1,theta2)));
print("Closed-form via Itakura-Saito divergence:",float(closedformKLIS(theta1,theta2)));
print("Differential entropy:",float(h(theta)));  
\end{lstlisting}

 When $D$ is an odd integer, the MEF is not defined but we can define Absolute Monomial Exponential Families~\cite{nielsen2017maxent} $\calA_D$ (AMEFs) with PDFs:
$$
p_{D,\theta}(x):=\exp(\theta\,|x|^D-F_D(\theta)).
$$
AMEFs $\calA_D=\{p_{D,\theta}(x)\ :\ \theta<0\}$ coincide with MEFs $\calM_D$ for even integers.
Since MEFs and AMEFs are exponential families, they are Maximum Entropy (MaxEnt) distributions.
That means that for any other distribution $r$, we necessarily have $h[r]\leq h[p^{D,\eta}]=-F^*_D(\eta)$ where $\eta=E_r[|x|^D]$:
$$
h[r]\leq -F^*_D(E_r[|x|^D]).
$$

\section{Stein's lemma for continuous exponential families}\label{sec:Stein}

A function on $\bbR$ is said absolutely continuous if for all $\epsilon>0$ there exists $\delta>0$ such that for all finite pairwise disjoint intervals $\{(a_i,b_i)\}_i$ with $\sum_i (b_i-a_i)<\delta$ we have $\sum_i |f(b_i)-f(a_i)|<\epsilon$. 
Let $\AC(\bbR)$ denote the set of absolutely continuous functions on $\bbR$.

\begin{Lemma}\label{lemma:Stein}
Let $X$ be a continuous random variable with exponential density 
$$
p_\theta(x)=\exp(\sum_{i=1}^D \theta_i t_i(x)-F(\theta)+k(x))
$$ 
with support $\calX=\bbR$.
For any $f\in\AC(\bbR)$ with $E[f'(X)]<\infty$, we have the following Stein identity:
\begin{equation}
E_{p_\theta}\left[\left(\sum_{i=1}^D \theta_i t_i'(x)\theta_i+k'(x)\right) f(x) \right]=-E_{p_\theta}[f'(x)].
\end{equation}
\end{Lemma}

\begin{proof}
Recall the integration by parts $(u(x)v(x))'=u'(x)v(x)+u(x)v'(x)$ so that 
$$
\int_a^b (u(x)v(x))'\dx=[u(x)v(x)]_a^b=\int_a^b u'(x)v(x)\dx+ \int_a^b u(x)v'(x)\dx.
$$
Therefore, we have $\int_a^b u'(x)v(x)\dx=[u(x)v(x)]_a^b-\int_a^b u(x)v'(x)\dx$.
Let us integrate by parts $E[f'(X)]=\int_\bbR f'(x) p_\theta(x) \dx$:
$$
E[f'(X)] = [f(x) p_\theta(x)]_{-\infty}^{\infty} - \int_\bbR f(x) p_\theta'(x) \dx .
$$
Since $E[f(X)]<\infty$, we necessarily have  $\lim_{x\rightarrow \pm\infty} p_\theta(x)f(x)=0$, 
and therefore we get $[f(x) p_\theta(x)]_{-\infty}^{\infty}=0$. 
Since $p_\theta'(x)= \left(\sum_{i=1}^D \theta_i t_i'(x)\theta_i+k'(x)\right) p_\theta(x)$, it follows that
$$
E[f'(X)] = -\int_\bbR f(x)  \left(\sum_{i=1}^D \theta_i t_i'(x)\theta_i+k'(x)\right) p_\theta(x) \dx = -E\left[\left(\sum_{i=1}^D \theta_i t_i'(x)\theta_i+k'(x)\right) f(x)\right].
$$
\end{proof}

Notice that if we define $h(x):=e^{k(x)}$ (i.e., $k(x)=\log h(x)$) then
$p_\theta(x)=\exp(\sum_{i=1}^D \theta_i t_i(x)-F(\theta))h(x)$, 
and we have $k'(x)=\frac{h'(x)}{h(x)}$, so that the Stein identity can be rewritten equivalently as:
\begin{equation}
E\left[\left(\sum_{i=1}^D \theta_i t_i'(x)\theta_i+\frac{h'(x)}{h(x)}\right) f(x)\right]=-E[f'(X)].
\end{equation}

This proof extends the original proof of Hudson~\cite{EF-Hudson-1978} (1978) who originally considered $D=1$ and $t(x)=x$.

Thus for a PED of order $D$ with $t_i(x)=x^i$, we have:
\begin{equation}
E\left[\left(\sum_{i=1}^D   i\theta_i x^{i-1}\right) f(x)\right]=-E[f'(X)].
\end{equation}

Further letting $f(x)=x^j$, we get the following identity for PEDs:
\begin{equation}\label{eq:mr}
E\left[\sum_{i=1}^D   i\theta_i x^{i-1+j}\right]=-E[j x^{j-1}].
\end{equation}

After some rewriting, this equation corresponds to Theorem~1 of~\cite{Cobb-1983} for their type $N$ distributions (PEDs).
Note that there is a missing minus in Eq.~(1.1) of~\cite{Cobb-1983}: Compare with Eq.~10 of~\cite{CobbEF-1981}.
Let $\mu_i=E[x^i]$ denote the raw (non-central) moments.
Using the linearity of the expectation in Eq.~\ref{eq:mr}, we have for any integer $j$:

\begin{equation}
\sum_{i=1}^D i\theta_i \mu_{i+j}=-j\mu_{j-1}.
\end{equation}

Based on these identities, we recover Cobb et al. method~\cite{Cobb-1983} (and hence the Score Matching Method for PEFs~\cite{Hyvarinen-2007}).

\section{Monte Carlo method: Acceptance-rejection sampling}\label{sec:ar}
Let $\tilde{p}(x)$ be an unnormalized density (e.g., a density of an exponential polynomial family), and $f(x)$ a proposal density (which can be easily sampled).
Denote by $c$ a constant such that $c\, f(x)\geq \tilde{p}(x)$ for all $x\in\calX$.
Acceptance-rejection sampling proceeds by first sampling $x_0\sim f(x)$ and then  by sampling
a uniform variate $u_0\sim U(0,c f(x_0))$. 
If $u_0>\tilde{p}(x_0)$, we reject $x_0$ and reiterate the procedure. 
Otherwise, we accept $x_0$.

Let $\tilde{\mathcal{P}}=\{[x,\tilde{p}(x)]\ :\ x\in\calX\}$ and 
$\mathcal{F}=\{[x,f(x)]\ :\ x\in\calX\}$ denote the area bounded by the functions $\tilde{p}(x)$ and $f(x)$, respectively.
It can be shown that the probability of acceptance is $\frac{1}{c}\int \tilde{p}\dx=\frac{1}{c}\mathrm{Area}(\tilde{\mathcal{P}})$.
Let $p(x)=\frac{\tilde{p}(x)}{Z_p}$ denote the normalized density.
Acceptance-rejection sampling iterates on average $\frac{c}{Z_p}$ times before accepting a random variate~\cite{bishop2006pattern}.
Notice that $c$ depends on $Z_p$ in practice.
Figure~\ref{fig:ar} depicts the process of acceptance-rejection sampling.

\begin{figure}%
\centering
\includegraphics[width=0.7\columnwidth]{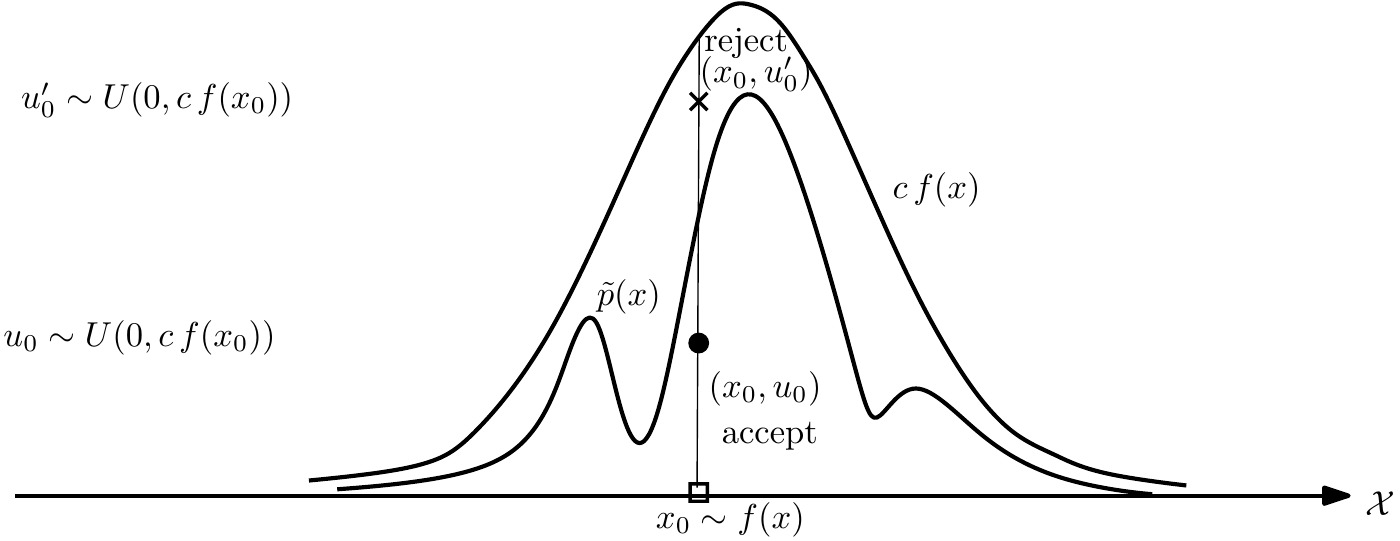}%
\caption{Acceptance-rejection sampling.}%
\label{fig:ar}%
\end{figure}

\section{An example of GMM conversion to PEDs of different orders}\label{sec:maximaexample}

In the example below in {\sc Maxima}, we display the conversion of two GMMs $m_1$ with $k_1=10$ components and $m_2$ with $k_2=11$ components into PEFs of order $8$.
The MC estimation of the JD with $s=10^6$ samples yields $0.2633\dots$ while the PEF approximation on corresponding PEFs yields
 $0.2618\ldots$ (relative error is $0.00585\dots$ or about $0.585\ldots \%$).

\begin{lstlisting}[backgroundcolor = \color{lightgray}]
assume(sigma>0);
normal(x,mu,sigma) := (1.0/(sqrt(2*%pi)*sigma))*exp(-((x-mu)**2)/(2*sigma**2) );
m1(x):=0.07364108534346539*normal(x,-6.870065732622028,1.986271783650286)+0.025167076634769497*normal(x,-6.075048963602256,1.2597407019915596)+0.1658464442557156*normal(x,-1.8146313935909912,1.8628859064709857)+0.08185153079585956*normal(x,-5.591128666334971,1.7384764507778412)+0.17161510648315073*normal(x,-1.1973521037847608,1.7318449234046134)+0.15953446622685713*normal(x,-9.310935699350253,1.254801266749996)+0.022842151933409417*normal(x,-3.0593318855130107,1.974864532982311)+0.049191546964925566*normal(x,-4.852351219823438,1.744469053647133)+0.14754497221852986*normal(x,-1.490305752594482,1.7981113425253326)+0.10276561914331742*normal(x,-6.216141690955186,1.362441455807411);
p1(x):=exp((-0.427070080652058*x**1)+(-0.1550659000258019*x**2)+(-0.011387815484428169*x**3)+(0.0013093778642312426*x**4)+(1.609814853908098E-4*x**5)+(4.97127592526588E-7*x**6)+(-2.1037139172837482E-7*x**7)+(-2.6728162063986965E-9*x**8));
m2(x):=0.12358972060770726*normal(x,-3.6623895510043063,1.8960669673645039)+0.02339461891279208*normal(x,-4.708573306071039,1.4988324123813517)+0.1506011549593897*normal(x,-5.892552442132609,1.5544383179469747)+0.09285155681935832*normal(x,-3.429016176798921,1.56860496582525)+0.10451002535239201*normal(x,-1.7627902228669008,1.469765667169659)+0.10052770665958442*normal(x,-7.640209485581262,1.4161278888109865)+0.01373064574217043*normal(x,-1.4307574402150962,1.9053029200511382)+0.10580776225290722*normal(x,-9.596750562831444,1.4463019439513571)+0.12295512911537661*normal(x,-4.567720102904078,1.2600844956599735)+0.12967771116575466*normal(x,-3.513053589366777,1.9316475019357529)+0.0323539684125673*normal(x,-7.980336695494762,1.9908903800330418);
p2(x):=exp((-0.8572007015061018*x**1)+(-0.13082046150595872*x**2)+(-0.009215306434221304*x**3)+(-0.0012216841345505935*x**4)+(6.151031668490958E-5*x**5)+(4.504153048723314E-5*x**6)+(4.319567125193681E-6*x**7)+(1.2196120465339543E-7*x**8));
plot2d([m1(x),m2(x)],[x,-10,10], [xlabel,"x"], [ylabel,"m(x)"],[legend, "m1", "m2"],[style, [lines,5,5],[lines,5,2]]);
plot2d([p1(x)],[x,-10,5], [xlabel,"x"], [ylabel,"q1(x)"], [legend, "q1"],[style, [lines,5,5]]);
plot2d([p2(x)],[x,-10,5], [xlabel,"x"], [ylabel,"q2(x)"], [legend, "q2"],[style, [lines,5,2]]);
/* Jeffreys MC:2642.581298828125ms (nb=1000000) vs Time PEF Jeffreys:0.8270999789237976ms ratio:3194.996211058532 */
/* MC Jeffreys divergence:0.26338216578112167 vs PEF: 0.2618412909304468 Error:0.005850338598686297 */
\end{lstlisting}

By inspecting the coefficients $\theta_1$ and $\theta_2$ of the PEFs, we see that coefficients fall sharply after order $5$.
Converting the mixtures to order $4$ yields a PEF JD approximation of $0.2433\dots$ with a relative error of $7.573\dots\%$.
 
\begin{lstlisting}[backgroundcolor = \color{lightgray}]
assume(sigma>0);
normal(x,mu,sigma) := (1.0/(sqrt(2*%pi)*sigma))*exp(-((x-mu)**2)/(2*sigma**2) );

m1(x):=0.07364108534346539*normal(x,-6.870065732622028,1.986271783650286)+0.025167076634769497*normal(x,-6.075048963602256,1.2597407019915596)+0.1658464442557156*normal(x,-1.8146313935909912,1.8628859064709857)+0.08185153079585956*normal(x,-5.591128666334971,1.7384764507778412)+0.17161510648315073*normal(x,-1.1973521037847608,1.7318449234046134)+0.15953446622685713*normal(x,-9.310935699350253,1.254801266749996)+0.022842151933409417*normal(x,-3.0593318855130107,1.974864532982311)+0.049191546964925566*normal(x,-4.852351219823438,1.744469053647133)+0.14754497221852986*normal(x,-1.490305752594482,1.7981113425253326)+0.10276561914331742*normal(x,-6.216141690955186,1.362441455807411);
p1(x):=exp((-0.42419488967268304*x**1)+(-0.15674407576089866*x**2)+(-0.012269008940931303*x**3)+(0.0013061387109766787*x**4)+(1.9499839508388961E-4*x**5)+(5.078100942919039E-6*x**6));
m2(x):=0.12358972060770726*normal(x,-3.6623895510043063,1.8960669673645039)+0.02339461891279208*normal(x,-4.708573306071039,1.4988324123813517)+0.1506011549593897*normal(x,-5.892552442132609,1.5544383179469747)+0.09285155681935832*normal(x,-3.429016176798921,1.56860496582525)+0.10451002535239201*normal(x,-1.7627902228669008,1.469765667169659)+0.10052770665958442*normal(x,-7.640209485581262,1.4161278888109865)+0.01373064574217043*normal(x,-1.4307574402150962,1.9053029200511382)+0.10580776225290722*normal(x,-9.596750562831444,1.4463019439513571)+0.12295512911537661*normal(x,-4.567720102904078,1.2600844956599735)+0.12967771116575466*normal(x,-3.513053589366777,1.9316475019357529)+0.0323539684125673*normal(x,-7.980336695494762,1.9908903800330418);
p2(x):=exp((-0.9006356445545727*x**1)+(-0.13457130798621164*x**2)+(-7.529018543852573E-4*x**3)+(7.859916869099592E-4*x**4)+(2.6290055488964458E-5*x**5)+(-1.1836096479592277E-6*x**6));
plot2d([m1(x),m2(x)],[x,-10,10], [xlabel,"x"], [ylabel,"m(x)"],[legend, "m1", "m2"],[style, [lines,5,5],[lines,5,2]]);
plot2d([p1(x)],[x,-10,5], [xlabel,"x"], [ylabel,"q1(x)"], [legend, "q1"],[style, [lines,5,5]]);
plot2d([p2(x)],[x,-10,5], [xlabel,"x"], [ylabel,"q2(x)"], [legend, "q2"],[style, [lines,5,2]]);
/* Time Jeffreys MC:2642.4267578125ms (nb=1000000) vs Time PEF Jeffreys:0.09099999815225601ms ratio:29037.657268864357*/
/* MC Jeffreys divergence:0.2632422749014559 vs PEF Jeffreys: 0.2433048263139042 Error:0.07573801964374921 */
\end{lstlisting}

\section{Acronyms and notations}

\begin{tabular}{ll}
PDF & Probability density function\\
EF & Exponential Family \\
EPF & Exponential-Polynomial Family\\
PEF & Polynomial Exponential Family \\
PED & Polynomial Exponential Density\\
MEF & Monomial Exponential Family \\
MGF & Moment Generating Function \\
AMEF & Absolute Monomial Exponential Family\\
GMM & Gaussian Mixture Model\\
$k$-GMM & GMM with $k$ components\\
ILSM & Iterative Linear System Method~\cite{MaxEnt-1992}\\
$D_\KL$ & Kullback-Leibler divergence\\
$D_\JS$ & Jensen-Shannon divergence\\
$B_F$ & Bregman divergence\\
$L_F$ & Legendre-Fenchel divergence\\
$D_H$ & Hyv\"arinen divergence (Fisher divergence)\\
$D_{H,\alpha}$ & $\alpha$-order Hyv\"arinen divergence\\
$D_J$ & Jeffreys divergence\\
$\hat{D}_J^{\calS_m}$ & Monte Carlo stochastic estimation of Jeffreys divergence\\
$\tilde{D}_J$ & Our approximation heuristic of Jeffreys divergence with two PED pairs\\
$\Delta_J$ & Approximation of $D_J$ with two PED pairs\\
$\bareta^\MLE$ & Integral-based Maximum Likelihood Estimator\\
$\bartheta^\SME$ & Integral-based Score Matching Estimator\\
$\tilde{\eta}^\SME$ & Converting approximately $\bartheta^\SME$ to moment parameter\\
$\tilde{\theta}^\MLE$ & Converting approximately $\bareta^\MLE$ to natural parameter\\
$\tilde{\bartheta}_T^\MLE$ & Approximation of $\nabla F^*(\bareta^\MLE)$ using $T$ iterations of~\cite{MaxEnt-1992}\\
$\Delta_J^\MLE$ & Approximation of $D_J$ using MLE   only and $\tilde{\theta}^\MLE$\\
$\Delta_J^\SME$ & Approximation of $D_J$ using SME   only and $\tilde{\eta}^\SME$\\
$p_\theta$ & EF density parameterized using natural parameter\\
$p^\eta$ & EF density parameterized using moment parameter\\
$\Theta$ & Natural parameter space\\
$H$ & Moment parameter space (read $H$ as greek Eta)\\
$F(\theta)$ & Cumulant function of an EF\\
$Z(\theta)$ & Partition function of an EF ($Z(\theta)=e^{F(\theta)}$)\\
$\tilde{p}_\theta=q_\theta$ & unnormalized EF density parameterized using natural parameter\\
$\mu_k(p)$ & raw moment or non-central moment $E_p[X^k]$.\\
$m_\theta(u)$ & moment generating function (MGF)\\ 
$p_{\mu,\sigma}$ & PDF of a normal distribution with mean $\mu$ and standard deviation $\sigma$\\
$\calE_t$ & EF with sufficient statistics $t(x)$\\
$\calE_D$ & PEF of order $D$, $t(x)=(x,x^2,\ldots,x^D)$\\
$\calE_4$ & Quartic EF\\
$\calM_D$ & MEF of order $D$\\
$\calA_D$ & AMEF of order $D$
\end{tabular}

\end{document}